\documentclass[letterpaper]{article}

\usepackage{fullpage}
\usepackage{pgffor}
\usepackage{xspace}

\usepackage{style}
\usepackage{bold-matrices}
\usepackage{bold-vectors}
\usepackage{thmtools}
\usepackage{thm-restate}

\newcommand{\dem}{\operatorname{dem}}
\renewcommand{\cong}{\operatorname{cong}}

\title{Submodular Hypergraph Partitioning: Metric Relaxations and Fast Algorithms via an Improved Cut-Matching Game}
\author{Antares Chen\\
\texttt{antaresc@uchicago.edu}\\
University of Chicago
\and
Lorenzo Orecchia\\
\texttt{orecchia@uchicago.edu}\\
University of Chicago
\and
Erasmo Tani\\
\texttt{etani@uchicago.edu}
\\
University of Chicago
}

\date{\today}

\begin{document}
\maketitle


\begin{abstract}
Despite there being significant work on developing spectral~\cite{chan2018spectral,lau2022cheeger,kwok2022cheeger}, and metric embedding~\cite{louis2016approximation} based approximation algorithms for hypergraph generalizations of conductance, little is known regarding the approximability of hypergraph partitioning objectives beyond this.

This work proposes algorithms for a general model of hypergraph partitioning that unifies both undirected and directed versions of many well-studied partitioning objectives. The first contribution of this paper introduces \emph{polymatroidal cut functions}, a large class of cut functions amenable to approximation algorithms via metric embeddings and routing multicommodity flows. We demonstrate an $O(\sqrt{\log n})$-approximation, where $n$ is the number of vertices in the hypergraph, for these problems by rounding relaxations to metrics of negative-type.

The second contribution of this paper generalizes the cut-matching game framework of Khandekar \etal \cite{khandekar2007cut} to tackle polymatroidal cut functions. This yields the first almost-linear time $O(\log n)$-approximation algorithm for standard versions of undirected and directed hypergraph partitioning~\cite{kwok2022cheeger}. A technical consequence of our construction is that a cut-matching game which greatly relaxes the set of allowed actions for both players can be used to partition hypergraphs with negligible impact on the approximation ratio. We believe this to be of independent interest.
\end{abstract}

\newpage

\section{Introduction}

In the past 20 years, increasing complexity in real-world data has necessitated generalizing graph models to represent higher order relations~\cite{agarwalHigherOrderLearning2006, catalyurekHypergraphpartitioningbasedDecompositionParallel1999}. Hypergraphs have served as the quintessential object to model relations involving multiple objects, and, as a consequence, the study of hypergraph algorithms has featured prominently in recent practical~\cite{bensonHigherorderOrganizationComplex2016,tsourakakis2017scalable} and theoretical developments~\cite{chan2018spectral,chan2012linear, louis2015hypergraph,louis2016approximation}.

Hypergraph partitioning is a fundamental algorithmic problem in this broader landscape. From an unsupervised machine learning perspective, hypergraph partitioning enables detecting significant clusters in complex higher-order networks, such as social~\cite{tsourakakis2017scalable,yang2017hypergraph} or biological networks~\cite{feng2021hypergraph, klamt2009hypergraphs}. From a theoretical perspective, partitioning algorithms are a fundamental primitive to be exploited in the development of divide-and-conquer methods for other hypergraphs problems. In contrast to graphs, where there is only one way for a partition to intersect an edge, hypergraphs admit multiple definitions for the cost of cutting a hyperedge. This yields a broad spectrum partitioning objectives, and allows the end-user to better design partitioning problems that fit specific applications. While there has been significant work on developing spectral approximation algorithms for hypergraph generalizations of conductance~\cite{chan2018spectral,lau2022cheeger,kwok2022cheeger}, with few exceptions~\cite{louis2016approximation}, little is known about the approximability of large classes of hypergraph partitioning objectives beyond this. Even less is known regarding fast algorithms for these tasks.

This paper proposes \emph{a general model for partitioning hypergraphs}. The first half of this paper identifies a large class of hypergraph partitioning objectives, which we call \emph{polymatroidal cut functions}. This class is a subset of submodular cut functions previously considered by~\cite{li2017inhomogeneous, liSubmodularHypergraphsPLaplacians2018} and captures both undirected and directed versions of almost all partitioning measures found in prior theoretical work. We show that these objectives possess properties that imply they are particularly well-suited for \emph{approximation methods based on metric embeddings and routing multi-commodity flows}. This leads us to polynomial time $O(\sqrt{\log n})$-approximations for these objectives based on rounding $\ell_2^2$-metric relaxations. The second half of this paper constructs fast algorithms that achieve similar $O(\log n)$-approximations by \emph{generalizing the cut-matching framework} of Khandekar \etal~\cite{khandekar2007cut}. This produces the first almost-linear time polylogarithmic approximation for both undirected and directed hypergraph conductance.

Let us now define concepts relevant to hypergraph partitioning. A \emph{hypergraph} is a tuple $G = (V, E)$, where $E \subseteq 2^V$ is a set of \emph{hyperedges}. The \emph{rank} of a hyperedge $h \in E$ is its cardinality $\lvert h \rvert$. If every $h \in E$ satisfies $\lvert h \rvert = k$, the hypergraph $G$ is called \emph{$k$-uniform}. Hence, an undirected graph is a $2$-uniform hypergraph. A key step in formalizing hypergraph cut problems is to define a suitable \emph{cut function} $\delta_h: 2^h \to [0,1]$ on each hyperedge $h \in E$. This expresses the cost incurred when $h$ is cut by $(A, \bar{A})$ into $A \cap h$ and $\bar{A} \cap h$. Central to our work are submodular cut functions, defined by~\cite{li2017inhomogeneous, liSubmodularHypergraphsPLaplacians2018}:

\begin{definition}[Submodular Cut Function] \label{def.cut-function}
A function $\delta_h: 2^h \to [0,1]$ is a {\it submodular cut function} if it is submodular and satisfies $\delta_h(\varnothing) = \delta_h(h) = 0$.
\end{definition}

\noindent
The space of cut functions for an edge $\{ i, j \}$ is spanned by the undirected edge cut function
$\delta_{\{i,j\}}(S) = \lvert \vone^S_i - \vone^S_j\rvert$ and the directed edge cut function $\delta_{(i,j)}(S) = (\vone^S_i - \vone^S_j)_+$, giving a complete classification of cut functions for rank-$2$ hyperedges. Higher-rank hyperedges support broader choices of cut functions.
The most common example is the {\it standard hypergraph cut function}
\cite{chan2018spectral, louis2015hypergraph}: $\delta^{\cut}_h(S) \defeq \min\{ 1, \lvert S \rvert, \lvert h\setminus S \rvert \}$, which takes value $1$ if $h$ is cut by $S$ and $0$ otherwise. If the hypergraph $G = (V, E)$ is associated with positive hyperedge weights $\vw \in \Z^E_{> 0}$, we can succintly write the cut function for the entire graph as $\delta_G(S) \defeq \sum_{h \in E} w_h
\cdot \delta_h(S)$.

With these definitions, we are interested in minimizing the \emph{submodular hypergraph ratio-cut objective}.

\begin{definition}[Submodular Hypergraph Ratio-Cut] Given a hypergraph $G = (V, E)$ with positive hyperedge weights $\vw \in \mathbb{Z}^E_{>0}$, non-negative vertex weights $\vmu \in \mathbb{Z}^V_{\geq 0}$, and a collection of submodular cut functions $\{\delta_h\}_{h \in E}$, the \emph{ratio-cut objective} on $G$ is defined for all $S \subseteq V$ as:
\begin{equation*}
\Psi_G(S)
\defeq \frac{\delta_G(S)}{\min\{\mu(S),\mu(\bar{S})\}}
= \frac{\sum_{h \in E} w_h \cdot \delta_h(S\cap h)}{\min\{\mu(S),\mu(\bar{S})\}} \, .
\end{equation*}
The \emph{ratio-cut} of $G$ is $\Psi^*_G \defeq \min_{S \subset V} \Psi_G(S)$.
\end{definition}

\noindent
When specialized to graphs\footnote{In the rest of the paper, for any undirected (resp. directed) graphs $G$, unless otherwise specified, we assume that $\delta_G$ and $\Psi_G$ are formed using the undirected (resp. directed) cut functions.} ratio-cut objectives include  both graph expansion ($\vmu= \ones$), graph conductance ($\forall i \in V, \mu_i = \sum_{h \in E : i \in h} w_h$). It also captures their directed counterparts, by replacing the undirected cut function $\delta_{\{i,j\}}$ with its directed analogue $\delta_{(i,j)}$.

\paragraph{Previous Work on Submodular Hypergraph Partitioning}
The standard hypergraph partitioning problem was first studied in the contexts of parallel numerical algorithms~\cite{catalyurekHypergraphpartitioningbasedDecompositionParallel1999} and scientific computing~\cite{devineParallelHypergraphPartitioning2006}.

Since then, many high-quality approximation algorithms have been developed for approximating hypergraph expansion and conductance. Louis and Makarychev~\cite{louis2016approximation} gave the first randomized polynomial-time $O\big( \sqrt{\log n} \big)$-approximation algorithm for hypergraph expansion. This matches the best known approximation for the graph expansion~\cite{ARV2009}, and uses the same relaxation to metrics of negative-type. Louis~\cite{louis2015hypergraph} and Chan \etal~\cite{chan2018spectral} developed a spectral approach towards approximation hypergraph conductance, achieving Cheeger-like guarantees. Their algorithms require solving semidefinite programs (SDP) to approximate the hypergraph spectral gap, returning a partition with conductance at most $O\big( \sqrt{\Psi^*_G \log \max_{h \in E} |h|} \big)$. Their results are known to be tight under the small-set expansion conjecture. In more recent work, Lau \etal used a similar approach to obtain a $O(\sqrt{\Psi^*_G \log(1/\Psi^*_G})$ guarantee for directed graph conductance~\cite{lau2022cheeger}.

Much less is known about the approximability of hypergraph ratio-cut objectives for general submodular cut functions. Following the work of Yoshida~\cite{yoshidaCheegerInequalitiesSubmodular2018}, Li and Milenkovic~\cite{li2017inhomogeneous, liSubmodularHypergraphsPLaplacians2018} proposed submodular cut functions. For submodular hypergraph conductance, rounding an SDP relaxation yields algorithms with conductance guarantees $O\big( \sqrt{\Psi^*_G \cdot \log |V|} \big)$~\cite{yoshidaCheegerInequalitiesSubmodular2018} and $O\big( \sqrt{\Psi^*_G \cdot \max_{h \in E} |h|} \big)$~\cite{liSubmodularHypergraphsPLaplacians2018}. Li and Milenkovic~\cite{liSubmodularHypergraphsPLaplacians2018} further conjecture that improving the dependence on hypergraph rank is $\ClassNP$-hard. The work of Yoshida~\cite{yoshidaCheegerInequalitiesSubmodular2018} and the survey by Veldt \etal~\cite{veldtHypergraphCutsGeneral2020} provide excellent discussions of submodular cut functions and their applications.

\paragraph{Previous Work on Cut-Matching Games}
Algorithms for the above settings require solving SDP relaxations over $\lvert V \rvert \times \lvert V \rvert$ symmetric matrices, with at least $\lvert E \rvert$ constraints. No faster algorithms are currently known beyond generic SDP solvers~\cite{jiangFasterInteriorPoint2020}, which run in $\Omega\big( \lvert V \rvert^2 \lvert E \rvert \big)$ time. Even for the simplest form of standard hypergraph partitioning, the resulting SDP relaxation is a mixed packing-and-covering program for which no almost-linear time algorithm is currently known~\cite{jambulapatiPositiveSemidefiniteProgramming2021}.

In order to develop fast approximation algorithms, one then considers primal-dual methods that solve the dual to metric relaxations for ratio-cut problems. This poses an immediate issue; the dual to these relaxations are given by multi-commodity flow problems whose demands are encoded by a $\vmu$-weighted complete graph. Such problems may require quadratic time to solve~\cite{aroraSqrtLognApproximation2010} due to the large number of demands. The cut-matching game~\cite{khandekarGraphPartitioningUsing2009} was originally a method of approximately reducing this multi-commodity flow to a polylogarithmic number of exact single-commodity flow computations. Each single-commodity flow problem would route a perfect matching in the graph, until the union of matchings approximates the desired demand graph.

Through the cut-matching game, Khandekar \etal~\cite{khandekarGraphPartitioningUsing2009} obtained a $O(\log^2 |V|)$-approximation to undirected graph expansion by exactly solving $O(\log^3 |V|)$ maximum flow problems. Later, Orecchia \etal~\cite{orecchiaPartitioningGraphsSingle2008} reduced both of these measures by $O(\log n)$. Subsequent works have vastly generalized the applicability of the cut-matching game: Louis~\cite{louis2010cut} extends~\cite{khandekarGraphPartitioningUsing2009} to approximate directed graph expansion, Long and Saranurak extend~\cite{khandekar2007cut} to approximate vertex and hypergraph expansion, and~\cite{nanongkai2017dynamic} demonstrate how to use the cut-matching game using approximate single-commodity flow solves. The cut-matching game has since become ubiquitous in designing fast deterministic algorithms for various static, and dynamic graph problems~\cite{bernstein2022deterministic, chuzhoyDistancedMatchingGame2023, chuzhoy2020deterministic, chuzhoyNewAlgorithmDecremental2019}.

\paragraph{Our Contributions}

In this paper, we achieve the following.
\begin{itemize}
\item We identify the class of polymatroidal cut functions, a subset of submodular cut functions amenable to metric and flow techniques. For this class, we develop notions of hypergraph flows and flow-embeddings of graphs into hypergraphs which certify lower bounds to the ratio-cut objective.

\item Using metric embedding techniques, we give polynomial-time $O(\sqrt{\log n})$-approximation algorithms for the minimum ratio-cut problem over submodular hypergraphs with polymatroidal cut functions. Our algorithms generalize the seminal result of Arora, Rao and Vazirani~\cite{aroraSqrtLognApproximation2010} for sparsest cut to the largest class of problems to date.

\item We extend the cut-matching game framework to approximate minimum ratio-cuts on submodular hypergraphs with polymatroidal cut functions. Our approach is based on approximating a continuous non-convex formulation of the minimum ratio-cut via a family of ``local'' convex submodular minimization problems. We then use the cut-matching framework to ``boost'' solutions to problems in this family, producing an $O(\log n)$-approximation algorithm for minimum ratio-cut that only requires a polylogarithmic number of approximate submodular minimization solves. This amounts to the first almost-linear time polylogarithmic approximation for both undirected, and directed hypergraph conductance.
\end{itemize}

A technical byproduct of our cut-matching game is that it allows for a cut player to play arbitrary disjoint vertex sets, and requires only a single $O(1)$-approximate submodular minimization solve per round. This may impact the deployment of cut-matching games in practice by lowering the required precision for the current running-time bottleneck of algorithms based on this framework~\cite{orecchia2022practical, veldt2023cut} as previous applications required either a single $o_n(1)$-approximate solves, or a polylogarithmic number of $O(1)$-approximate solves.

Our cut-matching game additionally avoids combinatorial restrictions found in previous applications of the framework. For example, the submodular minimization problem present in approximating hypergraph conductance is equivalent to a single-commodity flow whose demand graph is neither a perfect matching in the undirected case, nor Eulerian in the directed case. This may simplify the analyses of down-stream applications of the cut-matching game such as those found in~\cite{bernstein2022deterministic, chuzhoyDistancedMatchingGame2023, chuzhoy2020deterministic, chuzhoyNewAlgorithmDecremental2019, long2022near, nanongkai2017dynamic}.

\paragraph{Concurrent Work}
Subsequent to initial versions of this paper~\cite{ameranisEfficientFlowbasedApproximation2023}, Veldt~\cite{veldt2023cut} provided a generalization of the cut-matching framework to approximate minimum ratio-cut problems specified by cardinality-based symmetric submodular cut functions, a subset of polymatroidal cut functions. Their approach utilizes exact maximum flow solves to compute hypergraph cut preservers, a gadget which lower bounds the ratio-cut objective only when the cut function is cardinality-based~\cite{veldtHypergraphCutsGeneral2020, veldt2021approximate}. Unlike this work,~\cite{veldt2023cut} provides compelling empirical evaluations which should certainly be considered.

Independent of this work, Lau~\etal\cite{lau2023fast} give an almost-linear time $O(\sqrt{\log n})$-approximation for standard undirected and directed hypergraph cut functions. Their algorithm solves the dual to a relaxation of hypergraph conductance given by adding $\ell_2^2$-metric constraints to a minimum reweighted eigenvalue problem~\cite{lau2022cheeger}. This yields a multi-commodity flow problem, which they demonstrate to be approximable using a sequence of single-commodity flows via an interesting generalization of Sherman's algorithmic chaining~\cite{shermanBreakingMulticommodityFlow2009}. Their algorithm is not based on the cut-matching game, and require similar combinatorial restrictions when routing the single-commodity flow as found in previous works.

\paragraph{Paper Organization}
The paper is organized as follows:
\begin{itemize}
\item Section~\ref{sec.results} provides a technical overview of the results, and proof techniques found in this paper.

\item Section~\ref{sec.preliminaries} outlines preliminary tools and notation.

\item Section~\ref{sec.properties.and.examples} introduces polymatroidal cut functions, and provides key examples of functions in this class.

\item Sections~\ref{sec.hypergraph-flows} and \ref{sec.hypergraph.flow.embeddings} introduce hypergraph flow, and hypergraph flow embeddings. These serve as fundamental objects in constructing certificates of approximate optimality in the cut-matching game.

\item Section~\ref{sec.sdp-algorithm} details our $O(\sqrt{\log n})$-approximation for minimum ratio-cut given polymatroidal cut functions, describing the $\ell_2^2$-metric relaxation, and embedding results required to round an integral solution.

\item Section~\ref{sec.ci} describes our non-convex approach to the minimum ratio-cut problem and its connection to the cut-matching game.

\item Section~\ref{sec.alg-cm} describes our cut-matching game framework and uses it to construct an $O(\log n)$-approximation algorithm to the minimum ratio-cut problem with polymatroidal cut functions.
\end{itemize}

\noindent
The proof of several technical statements encountered throughout the paper are postponed to the appendix.

\section{Our Results}
\label{sec.results}

\subsection{The Metric Approach to Submodular Hypergraph Ratio-Cut}
\label{sec.results.metric}

The best polynomial-time approximation algorithms for graph ratio-cut problems are obtained by combining metric relaxations, either to general metrics~\cite{leightonMulticommodityMaxflowMincut1999} or $\ell_2^2$-metrics~\cite{ARV2009}, with metric embedding results for rounding. They achieve respectively an $O(\log n)$ and a $O(\sqrt{\log n})$-approximation ratio.
We study the applicability of this approach to the submodular hypergraph ratio-cut problem. At the outset, we remark that we should not hope to obtain polylogarithmic approximation for general submodular cut functions without further assumptions, as Svitkina and Fisher (Section 3 in~\cite{svitkinaSubmodularApproximationSamplingbased2010}) exhibit a submodular cut function for which any $o(\sqrt{n})$-approximation requires exponentially many queries to a value oracle.

Given this negative result, it is necessary to restrict the problem to a smaller class of cut functions.
At a high level, for a metric relaxation to work, the cut function must exhibit some kind of monotonicity, so that its value is preserved under low-distortion metric embeddings. However, a cut function cannot be monotone without being equal to $0$, so the monotone structure must appear in some other way.  To solve this issue, we introduce the class of \emph{polymatroidal cut functions}, which are the infimal symmetrization of monotone non-decreasing submodular functions (see below).
We believe that this new class captures most cut functions that are amenable to approximation via metric relaxation. As partial evidence to this statement, we show in Section~\ref{sec.polymatroidal} that the fractional dual objects associated with polymatroidal cut functions are closely related to polymatroidal flows, a well-studied generalization of network flows with many of the same favorable flow-cut gap properties~\cite{chekuri2012multicommodity}.

\paragraph{Polymatroidal Cut Functions and Metric Embeddings}
 Our first contribution is to define {\it polymatroidal cut functions}, a subclass of submodular cut functions given by the infimal convolutions of non-decreasing submodular functions.

\begin{definition}[Polymatroidal Cut Function]\label{def.monotone-submodular-cut-function}
A cut function $\delta_h: 2^h \to \R_{\geq 0}$ is {\it polymatroidal} if, for all $S \subseteq h$, it can be expressed as
\begin{align*}
\delta_h(S) & = \min \big\{F^{-}_h(S),\;F^{+}_h(h \setminus S) \big\},  
\end{align*}
where the associated functions $F_h^{-},F^{+}_h: 2^h \rightarrow \R$ are non-decreasing submodular functions such that $F^{-}_h(\varnothing) = F^{+}_h(\varnothing)= 0.$ When the associated functions $F^-_h$ and $F^+_h$ are identical, we refer to the cut function $\delta_h$ as  {\it symmetric polymatroidal}.
\end{definition}

Note that $\delta_h$ penalizes a cut $(A,\overline{A})$ intersecting $h$ in a directed manner, with $F^{-}$ penalizing $A \cap h$ and $F^{+}$ penalizing $\overline{A} \cap h$.  For a symmetric polymatroidal cut function $\delta_h$, such penalty is the same, so that $\delta_h(A \cap h) = \delta_h(\overline{A} \cap h)$, generalizing the cut function of undirected graphs and hypergraphs.
We show in Section~\ref{sec.polymatroidal} that Definition~\ref{def.monotone-submodular-cut-function} captures most of the cut functions proposed in previous work, including all directed and undirected standard hypergraph cut functions. Moreover, our setup also captures possible combinations of all these types of cut functions into a single framework.
We highlight two additional examples that are relevant in applications:
\begin{itemize}
\item $F_h$ is a cardinality-based non-decreasing submodular function~\cite{veldt2021approximate}, i.e., a non-decreasing concave function of the cardinality. This is useful when we wish to modify the star-expansion cut function by diminishing the penalty on balanced partitions of $h,$ e.g., by taking $\delta_h(S) = \min\big\{ |S|^p,|h \setminus S|^p \big\}$ for $p \in (0,1).$
\item $F_h$ is the entropy of a subset of random variables associated with vertices of $h$. When the hypergraph $G$ is the factor graph of an undirected graphical model, this yields a variant of the minimum information partition~\cite{ hidakaFastExactSearch2018,narasimhanQClustering2005}.
\end{itemize}

In Section~\ref{sec.sdp-algorithm}, we investigate the polynomial-time approximability of ratio-cut problems over submodular hypergraphs with polymatroidal cut functions. In particular, we prove the following result.
\begin{theorem}\label{thm.main-metric-approx}
There exists a randomized polynomial-time $O\big( \sqrt{\log |V|} \big)$-approximation algorithm
for solving the minimum ratio-cut problem on weighted submodular hypergraphs equipped with polymatroidal cut functions.
\end{theorem}

This result  extends the $O\big( \sqrt{\log n} \big)$-approximation of Louis and Makarychev~\cite{louis2016approximation} for standard hypergraph sparsest cut to all ratio-cut problems for the more general class of polymatroidal functions.
Our algorithm for Theorem~\ref{thm.main-metric-approx} solves a novel $\ell_2^2$-metric relaxation of the ratio-cut problem for weighted submodular hypergraphs (see the program~\ref{eqn.gen-vector-program} in Section~\ref{sec.sdp-algorithm}). This relaxation crucially exploits the form of the Lov\'asz extension of polymatroidal cut functions.
To the best of our knowledge, our results constitutes the broadest known generalization of the seminal result of Arora, Rao and Vazirani for sparsest cut, capturing a wealth of partitioning objectives, including directed and undirected graph ratio-cut, vertex-based ratio cut and hypergraph ratio-cut within a single simple analysis. It also improves on the best known approximation ratio for many other objectives in this class, such as cardinality-based cut functions~\cite{veldt2021approximate}.

\paragraph{Polymatroidal Cut Functions and Hypergraph Flows}
In the second part of the paper, we explore the efficient solution of the minimum submodular hypergraph ratio cut problem via flow methods based on the cut-matching game. We start by developing the notions of flow-cut duality and flow embeddings over hypergraphs. For this purpose, in Section~\ref{sec.hypergraph-flows}, we define a novel notion of \emph{hypergraph flows} for weighted submodular hypergraphs equipped with polymatroidal cut functions.
Specifically, for a weighted submodular hypergraph $G=(V,E, \vw, \vmu)$, we show that the subgradients of polymatroidal cut functions behave like network flows over the \emph{factor graph}, a graph analogue of the input submodular hypergraph that includes an auxiliary vertex for each hyperedge $h \in E$ and an edge $\{i,h\}$ for each $i \in h.$ For each hyperedge $h \in E$, the monotone functions $F_h^-$ and $F_h^+$ define capacity constraints on the flows entering and leaving the corresponding auxiliary node. In this process, we highlight new connections between submodular hypergraphs and the polymatroidal flows described by Lawler and Martell~\cite{lawler1982computing} and extended to the multi-commodity flow setting by Chekuri \etal~\cite{chekuri2012multicommodity}. The main algorithmic result in this part of the paper is a \emph{flow decomposition result for hypergraph flows} (Theorem~\ref{theorem.hypergraph-flow-decomposition}) which allows us to construct flow embeddings of directed and undirected graphs into hypergraphs while lower bounding the size of the hypergraph cut with the size of the corresponding graph cut. This will become a fundamental building block for our cut-matching game construction.

\subsection{A Non-Convex Optimization Approach To Minimum Ratio-Cut}
\label{sec.results.ci}

To generalize the cut-matching game framework to the setting of submodular hypergraphs with polymatroidal cut functions, we devise a novel optimization approach based on a continuous non-convex formulation \eqref{eq.rc-noncvx} of the submodular hypergraph ratio-cut problem.
Given a weighted submodular hypergraph $G = (V, E, \vw, \vmu)$, and cut functions $\{ \delta_h \}_{h \in E}$ with Lov\'asz extensions $\{\bar{\delta}_h\}_{h \in E}$, we define~\eqref{eq.rc-noncvx} as the following non-convex optimization problem over $\R^V$:
\begin{align*}
\bar{\Psi}_G(\vx)
&\defeq \frac{\sum_{h \in E} w_h \bar{\delta}_h(\vx)}{\min_u \|\vx - u \vone\|_{\vmu,1}} \\
\bar{\Psi}_G^*
&\defeq \min_{\vx \in \R^V} \bar{\Psi}_G(\vx)
\stepcounter{equation}
\tag{\texttt{RC-NonCvx}}
\label{eq.rc-noncvx}
\end{align*}
The fact that solving~\eqref{eq.rc-noncvx} is equivalent to solving the minimum ratio-cut problem (expressed in the following lemma) is a direct consequence of the submodularity of the cut functions. It is proved in full in Appendix~\ref{sec.appendix.omitted}.

\begin{restatable}{lemma}{lemcontinuous}
\label{lem.continuous-to-discrete-equivalence}
Given a weighted hypergraph $G^*=(V,E,\boldsymbol\mu,\vw)$, we have $\Psi_G^* = \bar{\Psi}_G^*$. Furthermore, there exists an algorithm that, given any $\vx \in \R^V$, recovers a cut $S \subseteq V$ satisfying
\begin{equation*}
\Psi_G(S) \leq \bar{\Psi}_G(\vx)
\end{equation*}
in time $O\big( |V| \log |V| + \sparsity(G) \big)$, where the sparsity of $G$ is defined as $ \sparsity (G) \defeq \sum_{h\in E} |h|.$
\end{restatable}

The main idea in our approach is to address the non-convexity of the $\bar{\Psi}_G$ objective by replacing the non-concave denominator with a linear lower bound. In particular, we can exploit the dual characterization of the $\ell_1$-norm in the lower bound to write:
\begin{equation}\label{eq.calc}
\min_{u \in \R} \lVert \vx - u \ones \rVert_{1, \vmu}
= \min_{u \in \R} \max_{\lVert \vy \rVert_\infty \leq 1} \langle \vy, \vx - u \ones \rangle_{\vmu}
= \max_{\substack{\lVert \vy\rVert_\infty \leq 1 \\ \langle \vy, \ones \rangle_{\vmu} = 0}} \langle \vy, \vx \rangle_{\vmu}
\end{equation}

This calculation directly leads us to define a  novel \emph{localized, convex formulation} of the \eqref{eq.rc-noncvx} problem for each vector $\vs \in \R^V$ with $\|\vs\|_{\infty} \leq 1$ and $\langle \vs, \vone \rangle = 0,$ which we call a \emph{seed}:

\begin{align*}
\bar{\Psi}_{G,\vs}(\vx)
&\defeq \frac{\sum_{h \in E} w_h \bar{\delta}_h(\vx)}{\max\{0,\langle \vs, \vx \rangle_{\vmu}\}} \\
\bar{\Psi}_{G,\vs}^*
&\defeq \min_{\vx \in \R^V} \bar{\Psi}_{G,\vs}(\vx)
\stepcounter{equation}
\tag{\texttt{Local-RC}}
\label{eq.rc-local}
\end{align*}
Intuitively, the program~\eqref{eq.rc-local} seeks a distribution over cuts $\vx$ with small expected cut size and  high correlation with the seed $\vs.$
In Section~\ref{sec.ci}, we investigate the properties of this formulation, including its equivalence with an integral cut problem, the ratio-cut improvement problem (Definition~\ref{def.rc-improve} in Section~\ref{sec.ci}), which
generalizes the cut improvement problem of Andersen and Lang~\cite{Andersen-Lang} to submodular hypergraphs.
For polymatroidal cut functions, we show that the natural notion of a dual solution for the convex problem~\ref{eq.rc-local} is exactly a hypergraph flow with demand vector $\diag(\vmu) \vs.$
By applying our flow decomposition result, we can turn such a hypergraph flow into a \emph{dual graph certificate}, a graph $D$ such that $\Psi^*_{G,\vs} \cdot \delta_D(S) \leq \delta_G(S)$ for all $S \subseteq V.$
To construct these objects, we prove the following algorithmic result:
\begin{theorem}[Informal. See Theorems~\ref{thm.general-solver} and~\ref{thm.maxflow-solver}]
\label{thm.informal-ci}
An approximate primal-dual solution to $\eqref{eq.rc-local}$ can be computed by solving $O(\log |V|)$-many decomposable submodular minimization problems for general polymatroidal cut functions or $O(\log |V|)$-many $\nicefrac{1}{2}$-approximate maximum flow problems over a graph of size $O(\sparsity(G)).$
\end{theorem}

Finally, and crucially for our purposes, for all $\vx \in \R^V,$ we have $\bar{\Psi}_G(\vx) \leq \bar{\Psi}_{G,\vs}(\vx)$ by Equation~\ref{eq.calc}. Hence, any (not necessarily optimal) solution $\vx$ to \eqref{eq.rc-local} for a seed $\vs$ yields a solution to \eqref{eq.rc-noncvx} of value at most $\bar{\Psi}_{G,\vs}(\vx)$ and, via Lemma~\ref{lem.continuous-to-discrete-equivalence}, a cut $S$ with $\Psi_G(S) \leq \bar{\Psi}_{G,\vs}(\vx).$
In other words, if we can somehow find a seed $\vs$ such that $\bar{\Psi}^*_{G,\vs} \leq \alpha \Psi^*_G$, then solving the \eqref{eq.rc-local} will yield an $\alpha$-approximation algorithm for $\Psi^*_G.$

This reduction may not seem very useful, as finding the required seed may be as hard as solving the original problem.
However, we can now exploit the dual feedback obtained when solving \eqref{eq.rc-local} for a sequence of seeds $ \vS = (\vs_1, \vs_2, \ldots, \vs_T)$ to effectively guide our search for a good seed and boost our approximation ratio.
In particular, suppose we can adaptively construct a sequence of $T$ seeds $\vS$ such that the sum $H =\sum_{t=1}^T D_t$
of the corresponding dual demand graphs $D_t$ has large ratio-cut, i.e., $\Psi_H^* \geq \beta.$  Then, it is easy to show (see Theorem~\ref{thm.cm.approx} and its proof) that we must have $\min_{t=1}^T \Psi^*_{G,\vs_t} \leq \nicefrac{T}{\beta} \cdot \Psi^*_G .$ Hence, solving \eqref{eq.rc-local} for one of the seeds in the sequence $\vS$ must yield a $\nicefrac{T}{\beta}$-approximation to the minimum ratio-cut problem.
It turns out that the problem of constructing such sequence $\vS$ is precisely the task addressed by the cut-matching framework of Khandekar \etal~\cite{khandekarGraphPartitioningUsing2009}, which we generalize and strengthen in our next contribution.

\subsection{An Improved Cut-Matching Game}
\label{sec.results.cm}

We now describe our new definition of the cut-matching game.
A key distinction between our setup and the previous lies in what constitutes admissible actions played by the cut and matching players.

\begin{restatable}{definition}{cutaction}\label{def.cut-action}
Let $V$ be a set of vertices with vertex weights $\vmu \in \Z^{V}_{\geq 0}$. A {\it cut action} is a pair $(A, B)$ of non-empty disjoint subsets $A, B \subseteq V$ such that $\vmu(A) \leq \vmu(B).$
\end{restatable}

\noindent
The cut player is no longer restricted to returning a partition of the vertex set, let alone a bisection. The set of valid actions for a matching player is similarly relaxed, as the edges of the response do not need to form a perfect matching. It is only required that enough of the available degree is routed across the cut action.

\begin{restatable}{definition}{approximatematchingresponse}\label{def.approx-matching-action}
An \emph{approximate matching response} to a cut action $(A, B)$ is an weighted bipartite graph $D = \big( A \cup B, E_D, \vw^D, \vmu \big)$ satisfying the following conditions.
\begin{enumerate}
\item Bounded degree: for all $i \in V$
\begin{equation*}
\deg_D(i)
\leq \begin{cases}
  \mu_i \; &i \in A \\
  \frac{\vmu(A)}{\vmu(B)} \cdot \mu_i \;& i \in B
\end{cases} \, .
\end{equation*}

\item Largeness: $\vw^D \big( E(A, B) \big) \geq {\vmu(A) \over 2}$.
\end{enumerate}
We may emphasize that $D$ is a \emph{directed approximate matching response}, if the edges of $D$ are directed from $A$ to $B.$
\end{restatable}

\noindent
This definition is motivated by the fact that an approximate matching response naturally arises from the dual graph certificate obtained by approximately solving the \eqref{eq.rc-local} program above. The matching response will be undirected when we wish to approximate the minimum ratio-cut of a hypergraph with symmetric cut functions and directed in the general case of asymmetric cut functions. Let us now define the cut-matching game.

\begin{restatable}{definition}{generalizedcutmatchinggame}
\label{def.cut-matching-game}
Let $n > 0$, and $\vmu \in \Z^V_{\geq 0}$ be a weighting over a set of vertices $V$ where $\lvert V \rvert = n$. A $(n, \vmu)$-\emph{generalized cut-matching game} is a multi-round, two-player game between a cut player $\cC$, and a matching player $\cM$ that proceeds as follows. In the $t$-th round of the game, following definitions~\ref{def.cut-action} and \ref{def.approx-matching-action}:
\begin{enumerate}
\item $\cC$ plays a cut action $(A_t, B_t)$,

\item $\cM$ responds with an approximate matching response $D_t$ to $(A_t, B_t)$.
\end{enumerate}
The \emph{state graph} after $t$ rounds is $H_t = \sum_{s=1}^t D_s$ the edge-wise sum of the matching response seen thus far.
\end{restatable}

\noindent
For comparison, the original cut-matching game of~\cite{khandekarGraphPartitioningUsing2009} is captured by the above definition when one specifically considers graph expansion $\vmu = \ones$, restricts cut actions to be exact bisections $(A, \bar{A})$ where $\lvert A \rvert = \lvert \bar{A} \rvert$, and requires exact matching responses: $\vw^D\big( E(A, \bar{A}) \big) = \vmu(A)$.

To approximate minimum ratio-cut using our cut-matching games, we seek a cut player $\cC$ that outputs a sequence of cuts which, in as few cuts as possible, forces the matching player $\cM$ to place edges that make the minimum ratio-cut objective for the state graph $H_T$ large. We call a cut player \emph{good} if it can achieve this.

\begin{restatable}{definition}{goodcutstrategy}
\label{def.good-cut-strategy}
A \emph{cut strategy} is a randomized algorithm $\cA_{\cut}$ that takes as input the current state graph $H$, and outputs a cut action $(A, B)$. A cut strategy $\cA_{\cut}$ is $\big( f(n), g(n) \big)$-\emph{good} if a cut player $\cC$ using $\cA_{\cut}$ at every iteration, achieves:
$$
\Psi_{H_{g(n)}} \geq f(n) \, \,
$$
with constant probability, for any (potentially adaptive) sequence of approximate matching responses $D_1, \ldots, D_{g(n)}$.
\end{restatable}

Note that, when dealing with symmetric cut functions, $H$ will be undirected, so that  $\Psi_H$ refers to the undirected ratio-cut objective.  If the cut functions are asymmetric, then $\Psi_H$ is naturally replaced with the directed ratio-cut objective.

These definitions reveal the mechanism by which the cut-matching framework adaptively constructs a sequence of seeds to~\eqref{eq.rc-local}: the seed we choose at every iteration of cut-matching game is just a $\vmu$-weighted indicator vector for the cut action of that round. A good strategy guarantees that, no matter what dual graph certificates we receive for each seed, at the end of the game at $T=g(n)$ iterations, we can guarantee $H_T$ has ratio-cut objective at least $\beta = f(n)$. This shows that a good strategy yields an approximation ratio of $O(\nicefrac{g(n)}{f(n)})$
for minimum ratio cut. In Section~\ref{sec.alg-cm}, we formalize this result as Theorem~\ref{thm.cm.approx}.
Our main technical theorem (Theorem~\ref{thm.cm.good-cut-strategy}) regarding the new cut-matching game shows the existence of good strategies for the cut-player that yield $O(\log n)$-approximation algorithms for the minimum ratio-cut problem over submodular hypergraphs, for both the cases of symmetric and general polymatroidal cut functions. The proof of this theorem is based on a new geometric result about separated sets in vector embeddings, which we discuss below.
Together with Theorem~\ref{thm.informal-ci}, our results on good cut-player strategies directly imply the following results for the solution of the minimum ratio-cut problem. Their short proof appear in Appendix~\ref{sec.appendix.omitted.alg-cm}.

\begin{restatable}{corollary} {genpoly}
\label{thm.cm.approx-logn}
There exists an $O(\log n)$-approximation algorithm for the minimum ratio-cut problem over submodular hypergraphs with polymatroidal cut functions whose running time is dominated by the solution of a polylogarthmic number of decomposable submodular minimization problems over the hypergraph cut function.
\end{restatable}

\begin{restatable}{corollary}{polycut}
There exists an $O(\log n)$-approximation algorithm for the minimum ratio-cut problem over submodular hypergraphs equipped with the directed or standard hypergraph cut function whose running time is almost linear in the hypergraph sparsity.
\end{restatable}

\paragraph{Separated Sets}
In our new cut-matching game, the cut player must pay considerable more care in choosing a cut action $(A, B)$ than in the original version, as the approximate nature of the matching response means that the matching player may easily keep some small subset of vertices disconnected from the rest of the graph. Suppose for instance that the cut player restricted itself to play $\mu$-bisections. Then, the matching player could select a small cut $T$, with $\mu(T) \leq \nicefrac{1}{10} \cdot \mu(V)$ and never place any edge across the boundary of $T.$ The only way the cut player can force progress on $T$ is to depart from playing bisections and play a cut action $(A,B)$ where $A$ is close or equal to $T.$

Let us see how this intuition is formalized in our design of cut-player strategies. As in the analysis of Orecchia \etal~\cite{orecchiaPartitioningGraphsSingle2008}, our cut strategy maintains at each iteration $t$ a spectral embedding for the current state graph $H_t$. We make progress if we can force the matching player to connect vertices that are far away in this embedding. We show that we can achieve this, even against an approximate matching player, by choosing a cut action coming from the following notion of separated sets:
\begin{definition}[$\sigma$-separated sets]
\label{def.separated-sets}
    Given an embedding $\{\vv_i\}_{i\in [n]}$ and a measure $\vmu \in \R^n_{\geq 0}$ we say two sets $S, T \subseteq [n]$ are \emph{$\sigma$-separated} if:
    \[
        \vmu(S) \cdot \vmu(T) \cdot \min_{\substack{i\in S\\ j\in T}} \norm{\vv_i - \vv_j}^2 \geq \sigma \cdot \sum_{i,j \in \binom{V}{2}} \mu_i \mu_j \cdot \norm{\vv_i - \vv_j}^2.
    \]
\end{definition}

This definition generalizes the idea of well-separated sets~\cite{ARV2009}, which was applied only to large balanced sets, to possibly unbalanced sets, as long as $S$ and $T$ contribute a $\sigma$-fraction to the total variance of the embedding.
In Section~\ref{sec.separated}, we show how to efficiently find separated sets and prove the following result.

\begin{theorem}
\label{thm:separated}
Given an embedding $\{\vv_i\}_{i\in [n]} \subseteq \R^d$ and a measure $\vmu \in \Z^n_{\geq 0}$, with $\mu(V) = \poly(n)$, there exists a randomized algorithm running in time $\tilde{O}(nd)$ that outputs $\sigma$-separated sets $S$ and $T$ for $\sigma = \Omega(\nicefrac{1}{\log n})$ with high probability.
\end{theorem}

\section{Preliminaries}\label{sec.preliminaries}

\paragraph{Hypergraphs} The sparsity $\sparsity(G)$ of a hypergraph is defined as $ \sparsity (G) \defeq \sum_{h\in E} |h|.$ Throughout the paper, we assume that the edge weights $\vw$ and vertex weights $\vmu$ of our instance submodular hypergraph $G=(V, E, \vw,\vmu)$ are polynomials in $|V|.$ In particular, we have ${\vmu}(V) \in \poly(|V|).$

\paragraph{Linear Algebra Notation}
Throughout this paper, we adopt the convention of representing scalars $x$ in unbolded lowercase, vectors $\vx$ in boldface lowercase, and matrices $\mX$ as boldface uppercase.
Given vectors $\vx, \vy \in \R^n$, we will use $\vx \circ \vy$ to denote the Hadamard (entry-wise) product between $\vx$ and $\vy$.
We will  use $\lvert \vx \rvert$ to denote the vector whose entries are the absolute value of the entries in $\vx$.
We let $\langle \vx, \vy \rangle$ denote the standard Euclidean inner product, $\lVert \cdot \rVert$ the Euclidean norm, and $\lVert \cdot \rVert_p$ the $\ell_p$-norm.

We define an inner product between $\vx, \vy \in \R^n$ weighted by a non-negative vector $\vmu \in \R^n_{\geq 0}$ as
$
\langle \vx, \vy \rangle_{\vmu} \defeq \sum_{i=1}^n \mu_i \cdot x_i\cdot y_i \, .
$
The quantity $\lVert \vx \rVert_{\vmu} \defeq \sqrt{\langle \vx, \vx \rangle_{\vmu}}$ will denote the norm induced by the inner product $\langle \cdot, \cdot \rangle_{\vmu}$, while for any $p > 0$, we use $\lVert \vx \rVert_{\vmu, p}$ to denote the $\ell_p$-norm weighted by $\vmu$
$
\lVert \vx \rVert_{\vmu, p}
\defeq \left( \sum_{i=1}^n \mu_i \cdot \lvert x_i \rvert^p \right)^{1/p} \, .
$
Note consequently that $\lVert \vx \rVert_{\vmu} = \lVert \vx \rVert_{\vmu, 2}$. For $\vx, \vy \in \R^n$, we write $\vx \perp \vy$ to indicate the relation: $\langle\vx,\vy\rangle = 0$.

Given any vector $\vx \in \R^n$, we define the positive and negative parts of $\vx$ as the vectors $\vx_+,\vx_-\in \R^n$ with
        $\vx_+(i) \defeq  \max\{0,\vx(i)\}$ and
        $\vx_-(i) \defeq  \max\{0,-\vx(i)\}.$

If $\mX, \mY \in \R^{n \times n}$, then $\langle \mX, \mY \rangle = \sum_{i,j=1}^n X_{ij} Y_{ij}$ denotes the Frobenius inner product. We let $\cS^{n \times n}$ be the set of $n$ by $n$ symmetric matrices. Given $\mX,\mY \in \cS^{n \times n}$ we denote by $\mX \succeq 0$ the statement that $\mX$ is positive semidefinite, and by $\mX \preceq \mY$ the relation $0 \preceq \vY-\vX$.  We use $\diag(\vx) \in \R^{n \times n}$ to denote the real matrix with the coordinates of $\vx$ placed along its diagonal. We will reserve specifically the matrix $\mM \defeq \diag(\vmu)$ for the non-negative vector $\vmu \in \R^n_{\geq 0}$.

Given a finite set $V$, a vector $\vx \in \R^V$, and a subset $S \subseteq V$, we let $\vx(S) = \sum_{v\in S} x_v$. We also denote by $\vx_S \in \R^V$, the restriction of $\vx$ to coordinates in $S$.
In this way, the vector $\ones_S$ denotes the $\{0,1\}$-indicator of set $S$.

\paragraph{Preliminaries on Graphs}

Given a graph $H = (V, E, \vw)$, denote its combinatorial Laplacian as $\mL(H)$:
\begin{equation*}
\mL(H) \defeq \mD(H) - \mA(H),
\end{equation*}
If $H$ is directed, then we will associate with it a directed Laplacian $\mL_+$ given by
\begin{equation*}
\mL_{+}(H) = \sum_{(i,j) \in E} w_{ij} \cdot \big( \mL_{ij} + \mE_{ii} - \mE_{jj} \big)
\end{equation*}
where $\mL_{ij}$ is the Laplacian of an undirected graph $\big(V, \big\{ \{ i, j \} \big\} \big)$, and $\mE_{ii}$ is $1$ on $ii$-th entry and $0$ everywhere else. We will denote $\tilde{H}$ as the symmetrization of $H$, i.e. the weighted undirected graph formed by taking the edge-wise sum ${\mA(G) + \mA(G)^{\top}}$ of $H$ and its reverse graph.


\paragraph{Factor Graphs} A hypergraph $G = (V, E_G, \vw, \vmu)$ is naturally associated with its {\it factor graph}, an unweighted graph $\hat{G} = (V \cup E_G, \hat{E}_G)$ that is bipartite between the \emph{variable} vertices representing vertices $i \in V$, and \emph{factor} vertices representing hyperedges $h \in E_G$. For $i \in V$ and $h \in E_G$, the edge $\{ i, h \}$ belongs to $\hat{E}_G$ if and only if $i \in h$, i.e.:
\begin{equation*}
\hat{E}_G = \{ \{ i, h \} \in V \times E_G : i \in h \}.
\end{equation*}

\begin{figure}
    \centering
    \tikzset{every picture/.style={line width=0.75pt}} 

\begin{tikzpicture}[x=0.75pt,y=0.75pt,yscale=-1,xscale=1]

\draw    (398,52) -- (453,96) ;
\draw    (329,69) -- (373,81) ;
\draw   (38,60) .. controls (72,28) and (141,29) .. (136,50) .. controls (131,71) and (124,79) .. (132,116) .. controls (140,153) and (4,92) .. (38,60) -- cycle ;
\draw   (106,117) .. controls (90,99) and (75.99,56.19) .. (111,43) .. controls (146.01,29.81) and (201,46) .. (224,52) .. controls (247,58) and (229,111) .. (216,130) .. controls (203,149) and (122,135) .. (106,117) -- cycle ;
\draw  [color={rgb, 255:red, 0; green, 0; blue, 0 }  ,draw opacity=1 ][fill={rgb, 255:red, 0; green, 0; blue, 0 }  ,fill opacity=1 ][line width=0.75]  (44,67) .. controls (44,64.79) and (45.79,63) .. (48,63) .. controls (50.21,63) and (52,64.79) .. (52,67) .. controls (52,69.21) and (50.21,71) .. (48,71) .. controls (45.79,71) and (44,69.21) .. (44,67) -- cycle ;
\draw  [color={rgb, 255:red, 0; green, 0; blue, 0 }  ,draw opacity=1 ][fill={rgb, 255:red, 0; green, 0; blue, 0 }  ,fill opacity=1 ][line width=0.75]  (113,50) .. controls (113,47.79) and (114.79,46) .. (117,46) .. controls (119.21,46) and (121,47.79) .. (121,50) .. controls (121,52.21) and (119.21,54) .. (117,54) .. controls (114.79,54) and (113,52.21) .. (113,50) -- cycle ;
\draw  [color={rgb, 255:red, 0; green, 0; blue, 0 }  ,draw opacity=1 ][fill={rgb, 255:red, 0; green, 0; blue, 0 }  ,fill opacity=1 ][line width=0.75]  (115,114) .. controls (115,111.79) and (116.79,110) .. (119,110) .. controls (121.21,110) and (123,111.79) .. (123,114) .. controls (123,116.21) and (121.21,118) .. (119,118) .. controls (116.79,118) and (115,116.21) .. (115,114) -- cycle ;
\draw  [color={rgb, 255:red, 0; green, 0; blue, 0 }  ,draw opacity=1 ][fill={rgb, 255:red, 0; green, 0; blue, 0 }  ,fill opacity=1 ][line width=0.75]  (214,62) .. controls (214,59.79) and (215.79,58) .. (218,58) .. controls (220.21,58) and (222,59.79) .. (222,62) .. controls (222,64.21) and (220.21,66) .. (218,66) .. controls (215.79,66) and (214,64.21) .. (214,62) -- cycle ;
\draw  [color={rgb, 255:red, 0; green, 0; blue, 0 }  ,draw opacity=1 ][fill={rgb, 255:red, 0; green, 0; blue, 0 }  ,fill opacity=1 ][line width=0.75]  (201,123) .. controls (201,120.79) and (202.79,119) .. (205,119) .. controls (207.21,119) and (209,120.79) .. (209,123) .. controls (209,125.21) and (207.21,127) .. (205,127) .. controls (202.79,127) and (201,125.21) .. (201,123) -- cycle ;
\draw   (106,52) .. controls (104,33) and (140,27) .. (180,31) .. controls (220,35) and (243,51) .. (236,69) .. controls (229,87) and (108,71) .. (106,52) -- cycle ;
\draw  [color={rgb, 255:red, 0; green, 0; blue, 0 }  ,draw opacity=1 ][fill={rgb, 255:red, 0; green, 0; blue, 0 }  ,fill opacity=1 ][line width=0.75]  (325,69) .. controls (325,66.79) and (326.79,65) .. (329,65) .. controls (331.21,65) and (333,66.79) .. (333,69) .. controls (333,71.21) and (331.21,73) .. (329,73) .. controls (326.79,73) and (325,71.21) .. (325,69) -- cycle ;
\draw  [color={rgb, 255:red, 0; green, 0; blue, 0 }  ,draw opacity=1 ][fill={rgb, 255:red, 0; green, 0; blue, 0 }  ,fill opacity=1 ][line width=0.75]  (394,52) .. controls (394,49.79) and (395.79,48) .. (398,48) .. controls (400.21,48) and (402,49.79) .. (402,52) .. controls (402,54.21) and (400.21,56) .. (398,56) .. controls (395.79,56) and (394,54.21) .. (394,52) -- cycle ;
\draw  [color={rgb, 255:red, 0; green, 0; blue, 0 }  ,draw opacity=1 ][fill={rgb, 255:red, 0; green, 0; blue, 0 }  ,fill opacity=1 ][line width=0.75]  (396,116) .. controls (396,113.79) and (397.79,112) .. (400,112) .. controls (402.21,112) and (404,113.79) .. (404,116) .. controls (404,118.21) and (402.21,120) .. (400,120) .. controls (397.79,120) and (396,118.21) .. (396,116) -- cycle ;
\draw  [color={rgb, 255:red, 0; green, 0; blue, 0 }  ,draw opacity=1 ][fill={rgb, 255:red, 0; green, 0; blue, 0 }  ,fill opacity=1 ][line width=0.75]  (495,64) .. controls (495,61.79) and (496.79,60) .. (499,60) .. controls (501.21,60) and (503,61.79) .. (503,64) .. controls (503,66.21) and (501.21,68) .. (499,68) .. controls (496.79,68) and (495,66.21) .. (495,64) -- cycle ;
\draw  [color={rgb, 255:red, 0; green, 0; blue, 0 }  ,draw opacity=1 ][fill={rgb, 255:red, 0; green, 0; blue, 0 }  ,fill opacity=1 ][line width=0.75]  (482,125) .. controls (482,122.79) and (483.79,121) .. (486,121) .. controls (488.21,121) and (490,122.79) .. (490,125) .. controls (490,127.21) and (488.21,129) .. (486,129) .. controls (483.79,129) and (482,127.21) .. (482,125) -- cycle ;
\draw    (373,81) -- (400,116) ;
\draw    (453,96) -- (486,125) ;
\draw    (400,116) -- (453,96) ;
\draw    (398,52) -- (373,81) ;
\draw    (398,52) -- (446,59) ;
\draw    (446,59) -- (499,64) ;
\draw [fill={rgb, 255:red, 255; green, 255; blue, 255 }  ,fill opacity=1 ]   (499,64) -- (453,96) ;
\draw  [color={rgb, 255:red, 0; green, 0; blue, 0 }  ,draw opacity=1 ][fill={rgb, 255:red, 255; green, 255; blue, 255 }  ,fill opacity=1 ][line width=0.75]  (369,81) .. controls (369,78.79) and (370.79,77) .. (373,77) .. controls (375.21,77) and (377,78.79) .. (377,81) .. controls (377,83.21) and (375.21,85) .. (373,85) .. controls (370.79,85) and (369,83.21) .. (369,81) -- cycle ;
\draw  [color={rgb, 255:red, 0; green, 0; blue, 0 }  ,draw opacity=1 ][fill={rgb, 255:red, 255; green, 255; blue, 255 }  ,fill opacity=1 ][line width=0.75]  (442,59) .. controls (442,56.79) and (443.79,55) .. (446,55) .. controls (448.21,55) and (450,56.79) .. (450,59) .. controls (450,61.21) and (448.21,63) .. (446,63) .. controls (443.79,63) and (442,61.21) .. (442,59) -- cycle ;
\draw  [color={rgb, 255:red, 0; green, 0; blue, 0 }  ,draw opacity=1 ][fill={rgb, 255:red, 255; green, 255; blue, 255 }  ,fill opacity=1 ][line width=0.75]  (449,96) .. controls (449,93.79) and (450.79,92) .. (453,92) .. controls (455.21,92) and (457,93.79) .. (457,96) .. controls (457,98.21) and (455.21,100) .. (453,100) .. controls (450.79,100) and (449,98.21) .. (449,96) -- cycle ;

\end{tikzpicture}
    \caption{A hypergraph (left) and its factor graph (right). In the factor graph image, the variable vertices are shown in solid black, while the factor vertices are shown as empty circles.}
    \label{fig:Factor Graph}
\end{figure}

\paragraph{Submodular Functions}
A set function $F: 2^V \rightarrow \R$ is submodular if for all $A, B \subseteq V$, we have:
\begin{equation*}
F(A \cup B) \leq F(A) + F(B) - F(A \cap B).
\end{equation*}
From the above, one can easily see that the set of submodular functions is closed under taking conic combinations.  Following standard conventions, we will always assume that $F(\varnothing) = 0$ unless specified otherwise.

Any submodular set function $F$ can be associated with a convex body $\cB(F)$, known as its base polytope:
\begin{equation}
\label{eq.base-polytope-definition}
\cB(F) \defeq \big\{
  \vx \in \R^V \, : \,
  \langle \vx, \ones \rangle = F(V) \, \wedge \,  \forall S \subseteq V,\,
  \langle \vx, \ones_S \rangle \leq F(S)
\big\}.
\end{equation}
%
When $F$ is monotone non-decreasing and $F(\emptyset) = 0,$ it is convenient to employ the positive submodular polyhedron:
\begin{equation*}
\cP_+(F) \defeq \big\{
  \vx \in \R^V_{\geq 0} \, : \, \forall S \subseteq V,\,
  \langle \vx, \ones_S \rangle \leq F(S) \big\}.
\end{equation*}
Crucially, a submodular set function $F$ can be extended to its convex closure, known as its Lov\'{a}sz extension and denoted by $\bar{F} : \R^V \rightarrow \R.$
There are a number of equivalent characterizations of the Lov\'{a}sz extension and its connection to the base polytope $\cB$ and the positive polyhedron $\cP_+$. Rather than listing them here, we will refer directly to the results in Chapter 3 and 4 of the excellent reference by Bach~\cite{bach2013learning}. We make an exception for the following fact, which plays a crucial role in the analysis of our approximation algorithms based on metric embeddings. See Proposition 4.8 in~\cite{bach2013learning} for a proof.
\begin{fact}\label{fct.lovaszmonotone}
Let $F: 2^V \to \R_{\geq 0}$ be a monotone non-decreasing submodular function with $F(\emptyset) = 0$. Its Lov\'asz extension $\bar{F}$ is monotone, non-decreasing over $\R^V.$
\end{fact}

\paragraph{Online Linear Optimization and Regret Minimization}
Fix a convex set of \emph{actions} $\cX \subseteq \cS^{n \times n}$. The setting of \emph{online linear optimization} over symmetric matrices considers the following multi-round game: in each round $t$, a player plays an action $\mX_t \in \cX$, receives a feedback matrix $\mF_t \in \cS^{n \times n}$, and suffers a loss of $\langle \mF_t, \mX_t \rangle$. The goal of the player is to play actions which minimize total loss incurred. In general, one cannot hope for a strategy that is competitive against an arbitrary, adaptive adversary. Thus, one instead seeks to \emph{minimize regret}: the difference, up to a time horizon $T > 0$, between total loss and the loss suffered by the best fixed action in hindsight.

We will require a regret minimization strategy to produce a sequence of $\mX_t$ for this setting. Fortunately, many prior results~\cite{arora2007combinatorial, allen2015spectral, hazan2012near, kale2007efficient, LorenzosThesis, orecchia2012balanced_separator} have studied this, and a Matrix Multiplicative Weight Update algorithm works off-the-shelf.

\begin{theorem}[Theorem 3.3.3 in~\cite{LorenzosThesis}]
\label{thm.regret-bound}
Fix $T > 0$, and a weight vector $\vmu \in \Z_{\geq 0}^n$. Let $\mF_1, \ldots, \mF_T$ be a sequence of feedback matrices such that there exists $\rho > 0$ where $\mF_{t} \succeq - \frac{1}{\rho} \cdot \mM$ for every $t > 0$. For any step size $0 < \eta < \rho$, define $\mX_{t}$ to be given by the matrix multiplicative weight update
\begin{equation*}
\stepcounter{equation}
\tag{\texttt{matrix-mwu}}
\label{eq.mmwu}
\begin{aligned}
\mW_{t+1} &= \mmwu_{\eta, \vmu}( \mF_1, \ldots, \mF_t ) \\
&\defeq \mM^{-1/2} \exp \bigg( - \eta \cdot \sum_{k=1}^{t} \mM^{-1/2} \mF_k \mM^{-1/2} \bigg) \mM^{-1/2} \\
\mX_{t+1} &= \frac{\mW_{t+1}}{\langle \mM, \mW_{t+1} \rangle}
\end{aligned}
\end{equation*}
Then, the following inequalities hold for any $\mU \in \Delta^{n \times n}_{\vmu}$. We have:
\begin{equation}
\label{eq.regret-bound.standard}
\frac{1}{T} \sum_{t=1}^T \langle \mF_t, \mU \rangle
\geq \frac{1}{T} \sum_{t=1}^T \langle \mF_t, \mX_t \rangle - \frac{\eta}{T} \sum_{t=1}^T \, \big\lVert \mM^{-1/2} \mF_t \mM^{-1/2} \big\rVert_{\spectral}^2 - \frac{\log n}{\eta T} \, .
\end{equation}
Furthermore, if $\mF_{t} \succeq - \frac{1}{\rho} \cdot \mM$ for every $\rho > 0$ (i.e. $\mF_t \succeq \mZero$), and there exist $\epsilon > 0$ such that $\mF_t \preceq \epsilon \cdot \mM$, then:
\begin{equation}
\label{eq.regret-bound.local-norm}
\frac{1}{T} \sum_{t=1}^T \big\langle \mF_t, \mU \big\rangle
\geq \frac{1 - \epsilon \eta}{T} \cdot \sum_{t=1}^T \big\langle \mF_t, \mX_t \big\rangle - \frac{\log n}{\eta T} \, .
\end{equation}
\end{theorem}

\noindent
We remark that equation~\eqref{eq.regret-bound.standard} is the standard regret bound for the Matrix Multiplicative Weight Update algorithm, while~\eqref{eq.regret-bound.local-norm} is the regret bound proven using local norm convergence.

\section{Polymatroidal Cut Functions and Hypergraph Flows}
\label{sec.polymatroidal}

In this section, we establish the main properties of the polymatroidal cut functions in Definition~\ref{def.monotone-submodular-cut-function}. In the second part of the section, we introduce the dual notion of hypergraph flows as natural lower bounds to polymatroidal cut functions and highlight their connection with previously studied polymatroidal flows~\cite{chekuri2012multicommodity}. Finally, we define the notion of flow embedding of a graph into a hypergraph.

\subsection{Properties and Examples}\label{sec.properties.and.examples}

We start by verifying that any polymatroidal cut function is submodular. A proof of the following fact can be found in Appendix \ref{sec.appendix.omitted}:

\begin{restatable}{fact}{submodularsymmetrization}
\label{fact.submodular-symmetrization}
Let $V$ be an arbitrary finite set and $F,G:2^V \to \R$ be monotone non-decreasing submodular functions. If $\delta:2^V \to \R$ is the function defined by:
\[
    \delta(S) \defeq \min\{F(S),G(V \setminus S)\},
\]
for all $S \subseteq V$, then $\delta$ is submodular.
\end{restatable}

\paragraph{Examples of Polymatroidal Cut Functions} We now show that the class of polymatroidal cut functions encompasses many popular cut objectives. Our first example is the directed edge cut function in graphs. This function is defined on a directed edge $(i,j)$ as:
\[
    \delta_{(i,j)}(S) = \begin{cases}
        1 &\text{if } i\in S\text{ and }j
        \in \overline{S},\\
        0 &\text{otherwise}.
    \end{cases}
\]

It is easy to check that this function is polymatroidal, since it can be rewritten as:
\[
    \delta_{(i,j)}(S) = \min\{ |S \cap \{i\}|, |(h\setminus S) \cap \{j\}| \}
\]
for all $S \subseteq h$.
Next, we consider the standard hyperedge cut function. Given a hyperedge $h \in E$, this function is defined as:
\[
    \delta^\cut_h(S) \defeq \begin{cases}0 &\text{if }S\cap h \in \{\emptyset, h\},\\
    1 &\text{otherwise}.\end{cases}
\]

Note that the undirected cut function in graphs is a special case of $\delta^\cut$. This function is symmetric polymatroidal, as we can let $F_h (S) = \min\{1, |h\cap S|\} = G_h(S)$ for any $h\in E$ to obtain:
\[
    \delta_h^{\cut}(S) = \min\{1, |h\cap S|, |h\cap \overline{S}|\},
\]
which is $1$ if and only if $S$ intersects $h$ non-trivially, and zero otherwise.

Moreover, the class of directed hypergraph cut functions recently introduced by Lau et al. (Section 6.2 of~\cite{lau2023fast}) to define directed hypergraph expansion is also a sub-class of polymatroidal cut functions, as we now show. Given a hyperedge $h\in E$ and a partition of $h$ into a tail set $T_h$ and a headset $H_h$, we can capture the directed cut function of Lau et al., by letting $F_h(S)= \min \{1,|T_h\cap S|\}$ and $G_h(\overline{S}) = \min \{1,|H_h\cap \overline{S}|\}$.

Finally, Veldt \etal \cite{veldt2021approximate} consider a class of cut functions which are cardinality-based and submodular. In Appendix B.2 of their paper, they argue that every such cut function $\delta_h$ can be written as:
\[
    \delta_h(S) = g_h(\min\{|S \cap h|, |h\cap \overline{S}|\}),
\]
for some concave and monotonically increasing function $g_h$. Since $g_h$ is monotonically increasing, the above gives:
\[
    \delta_h(S) = \min\{g_h(|S \cap h|), g_h(|h\cap \overline{S}|)\},
\]
and hence all functions in this class are also examples of polymatroidal cut functions (note that the composition of a monotone, concave function with the cardinality function is always monotone and submodular).

\paragraph{Lov\'asz Extensions of Polymatroidal Cut Functions}
We next establish the form of the Lov\'{a}sz extension of a polymatroidal cut function. We will make heavy use of this form in building our metric relaxations and in rounding them via low-distortion metric embeddings in Section~\ref{sec.sdp-algorithm}.
The proof of the following fact is found in Appendix~\ref{sec.appendix.omitted}:

\begin{restatable}{fact}{formoflovaszextension}
\label{fact.form-of-lovasz-extension}
    Given a polymatroidal cut function $\delta_h$, we have, for all $\vx \in \R^h$:
    \[
        \overline{\delta}_h(\vx) = \min_{\nu \in \R}\bar{F}^-_h((\vx - \nu \vone_h)_+ ) + \bar{F}^+_h((\vx - \nu \vone_h)_-).
    \]
\end{restatable}

It is instructive to consider how the previous results are instantiated for the two canonical examples of the directed edge cut function and the standard hypergraph cut function.
For the former on an arc $(i,j)$, we have $\bar{F}^{-}(\vx) = \vx_i$ and $\bar{F}^{+}(\vx) = \vx_j,$ so that:
$$
\forall \vx \in \R^{\{i,j\}} ,\; \overline{\delta}_{(i,j)} (\vx) = \min_{\nu \in \R} (\vx(i)-\nu)_{+} + (\vx(j)-\nu)_{-} = (\vx(i) - \vx(j))_+ .
$$
For the latter, on a hyperedge $h$:
$$
\forall \vx \in \R^h ,\;
\overline{\delta}^{\cut}_h(\vx )
= \min_{\nu \in \R} \|(\vx - \nu \vone)_{+}\|_\infty + \|(\vx - \nu \vone\|)_{-}\|_\infty = .
$$

In the case of a symmetric polymatroidal cut function, the characterization of Fact~\ref{fact.form-of-lovasz-extension} can take a more compact form, as shown in Appendix~\ref{sec.appendix.omitted}:
\begin{restatable}{fact}{formoflovaszextensionsym}
\label{fact.lovasz-extension-symmetric}
    Given a symmetric polymatroidal cut function $\delta_h$ with $F_h \defeq F^+_h = F^-_h$, we have, for all $\vx \in \R^h$:
    \[
        \min_{\nu \in \R} \bar{F}_h(|\vx - \nu \vone_h|) \leq \overline{\delta}_h(\vx) \leq 2 \cdot   \min_{\nu \in \R} \bar{F}_h(|\vx - \nu \vone_h|).
    \]
\end{restatable}

\subsection{Base Polytope, Hypergraph Flows and Polymatroidal Networks}
\label{sec.hypergraph-flows}

It is natural to attempt to lower bound the cost of cuts in a submodular hypergraph $G=(V,E_G, \vw, \vmu)$ by considering collections of vectors $\vY = (\vy_h \in \cB(\delta_h))_{h \in E_G}$, where each vector $\vy_h$ belongs to the base polytope of the polymatroidal cut function $\delta_h$ under consideration. In this section, we interpret these vectors as hypergraph flows over $G$ and show that they take the form of polymatroidal flows~\cite{lawler1982computing,chekuri2012multicommodity} over the factor graph $\hat{G}$.

\paragraph{Base Polytope of Polymatroidal Cut Function} We start by characterizing the base polytope $\cB(\delta_h)$ for a polymatroidal cut function $\delta_h.$ By the definition of base polytope in Equation~\ref{eq.base-polytope}, we have
$$
\cB(\delta_h) = \{ \vy \in \R^h : \langle \vy , \vone_h \rangle = 0 \, \wedge \, \forall S \subset h,\,
\langle \vy, \vone_S \rangle \leq \min\{F_h^{-}(S), F_h^+(h\setminus S)\}\}.
$$
The first constraint implies that for all $S \subset h,$ we must have $\langle \vy , \vone_S \rangle = - \langle \vy, \vone_{h\setminus S} \rangle.$ Hence, we can rewrite:
\begin{align}
\label{eq.base-polytope}
\cB(\delta_h) &= \{ \vy \in \R^h \perp \vone_h :  \forall S \subset h,\;
\langle \vy, \vone_S \rangle \leq F_h^{-}(S) \, \wedge \, - \langle \vy, \vone_{h \setminus S} \rangle \leq F_h^{+}(h \setminus S) \} \nonumber\\
              &= \{ \vy \in \R^h \perp \vone_h :  \forall S \subset h,\; - F_h^{+}(S) \leq \langle \vy, \vone_S \rangle \leq F_h^{-}(S)\}\\
             &= \{ \vy \in \R^h \perp \vone_h : (\vy)_+ \in \cP_+(F^{-}_h) \, \wedge \, (\vy)_- \in \cP_+(F^{+}_h)\}.\nonumber
\end{align}
We will interpret an element $\vy$ of $\cB(\delta_h)$ as a flow over the neighborhood of factor vertex $h$ in the factor graph $\hat{G}$, with $\vy_i > 0$ indicating flow going from $i$ to $h$ and $\vy_i < 0$ flow going from $h$ to $i$. The constraint $\vy \perp \vone_h$ simply represents flow conservation at $h.$ With this setup, for each set $S \subset h,$ the function $F_h^{-}$ (resp. $F_h^{+}$) constrains the amount of flow that can be routed from the set $S$ to $h$ (resp. to the set $S$ from $h$).

\paragraph{Example: Base Polytope for Standard Hypergraph Cut Function} Consider the case of the standard hypergraph cut function $\delta^{\cut}_h(S)$ defined in the previous section. Then, it is easy to see that:
\[
    \cB(\delta^{\cut}_h) = \{\vy \in \R^h\perp \vone_h \mid \norm{\vy}_1 \leq {2}\}.
\]
As discussed above, we can think of each entry $\vy(i)$ of $\vy$ for $i\in h$, as the amount of flow flowing into $h$ from vertex $i$, and the constraint $\norm{\vy}_1 \leq 2$ enforces that the total amount of flow through the hyperedge $h$ is no more than $1$. This is equivalent to a vertex capacity constraint on the factor vertex corresponding to the hyperedge $h$ in the factor graph $\hat{G}$.

\paragraph{Example: Base Polytope for Directed Hypergraph Cut Functions} Consider the class of directed hypergraph cut functions defined by Lau \etal in Section 6.2 of \cite{lau2023fast}. Each of these cut functions $\delta_h$ has an associated head set $H_h \subseteq h$ and an associated tail set $T_h \subseteq h$ and $\delta_h(S)$ has value one if and only if the head set intersects $H_h\cap S\neq \emptyset$ and $\overline{S}\cap \emptyset \neq \emptyset$. Their base polytope is then given by:
\[
    \cB(\delta_h) = \{\vy \in \R^h \perp \vone_h \mid \vy(H_h)\leq 1 \land \forall i\in T_h \setminus H_h : \vy(i)\leq 0\}.
\]

\noindent
The notion of hypergraph flow, which we introduce next, formalizes the network flow interpretation for all hyperedges $h \in E_G.$

\paragraph{Hypergraph Flows}
Given a weighted submodular hypergraph $G=(V,E, \vw \in \mathbb{Z}^E_{>0},  \vmu \in \mathbb{Z}^V_{\geq 0})$, a {\it hyperedge flow} over a hyperedge $h \in E$ is a vector $\vy_h \in \R^h \perp \vone_h.$ A {\it hypergraph flow} $\vY \in \bigoplus_{h\in E} \R^h \perp \vone_h$ is a direct sum of hyperedge flows $\vY = (\vy_h)_{h \in E}.$
Analogous to the case for graph flows, we define a notion of congestion based on the maximum violation of the polymatroidal constraints defining each $\cB(\delta_h)$:
\begin{equation}
\label{eq.congestion}
\cong_G(\vY) = \min \{ \rho \geq 0 : \forall h \in E_G,\, \vy_h \in \rho \cdot \vw_h \cdot \cB(\delta_h)\}.
\end{equation}
The demand vector $\dem(\vY) \in \R^V \perp \vone$ of a hypergraph flow $\vY = (\vy_h)_{h \in E}$ is
\begin{equation}
\label{eq.demand}
\dem_i(\vY) = \sum_{ h \in E: h \ni i} \vy_h(i).
\end{equation}
We will make use of the vector space structure inherited from $\vY \in \bigoplus \R^h \perp \vone_h$ to define linear combinations of hypergraph flows $c\vY_1 + \vY_2$ for any $c \in \R.$ The demand vector also behaves linearly as $\dem(c\vY_1+\vY_2) = c \cdot \dem(\vY_1) + \dem(\vY_2).$

Finally, notice that hypergraph flows over hyperedges of rank $2$ are equivalent to graph flows as, for an edge $\{i,j\} \in E$, $\vy_{ij}(i)$ represents the flow from $i$ to $j$, while $\vy_{ij}(j) = -\vy_{ij}(j)$ is the flow from $j$ to $i$.

\paragraph{Connection to Polymatroidal Flows}  In the polymatroidal network flow model \cite{chekuri2012multicommodity, lawler1982computing}, one has a directed (2-uniform) graph $H$ representing a flow network where the edges are not individually capacitated, but rather for every vertex $v\in V(H)$, there are monotone submodular functions $\rho^-$ and $\rho^+$ constraining the amount of flow entering and exiting a vertex respectively.

Given an hypergraph flow $\vY=\{\vy_h\}_{h\in E}$ over $G$, one can interpret it as a graph flow over the factor graph $\hat{G}$, where the flow from vertex $i$ to factor vertex $h$ over the edge $\{i,h\}$ is given by $\vy_h(i)$. When the congestion (as defined in \eqref{eq.congestion}) is required to be at most one, the constraints imposed on this flow on $\hat{G}$ can be cast as constraints defining a flow in a polymatroidal network supported on $\hat{G}$. Here, each constraint on the hyperedge flow for hyperedge $h \in E(G)$ defined by its base polytope, can be encoded as polymatroidal constraints on the graph flow into and out of the vertex $h$ in $\hat{G}$.

Despite this connection with polymatroidal network, we remark that, to the best of our knowledge, the hypergraph partitioning problems studied in this paper cannot be expressed in the setup proposed by Chekuri~\cite{chekuri2012multicommodity} for defining sparsest cut problems over polymatroidal networks.

\subsection{Hypergraph Flow Embeddings}\label{sec.hypergraph.flow.embeddings}

To construct efficient flow-based algorithms for the minimum ratio-cut problem in the second part of this paper, we will need a way to generalize the notion of flow embedding of a demand graph into a graph to the setting of hypergraphs. In this section, we leverage the definition of hypergraph flows to provide this generalization. In particular, we  define the flow embedding of a (directed) graph into a weighted submodular hypergraphs with polymatroidal cut functions. We also show that these embeddings yield lower bounds on the hypergraph cut functions in terms of the graph cuts of the embedded graph. Analogously to graphs, we also show that any hypergraph flow can be converted into a flow embedding via flow-path decompositions.

\begin{definition}[Hypergraph Flow Embedding]
\label{def.flow-embedding}
Let $G = \big( V, E_G, \vw^G, \vmu \big)$ be a weighted submodular hypergraph  equipped with a collection of polymatroidal cut functions $\{ \delta_h \}_{h \in E_G}$ and let  $H = \big( V, E_H, \vw^H \big)$ be a weighted directed graph on the same vertex set.  We say that the graph $H$  embeds into $G$ as a flow with congestion $\rho \geq  0$, denoted $H \preceq_\rho G$ if there exist hypergraph flows $\{\vY^e \}_{e \in E_H}$ such that:
\begin{enumerate}
\item For all $e=(i,j) \in E_H$, the flow $\vY^e$ routes $w^H_{e}$ units of flow from vertex $i$ to vertex $j$, i.e.,
$$
\dem(\vY^e)= w^H_{e} \cdot (\vone_i - \vone_j) \, .
$$
\item For each $h \in E_G$, the flows ${\vY^e}_{e \in E_H}$ respect the polymatroidal capacity constraints, i.e., for all $h \in E_G$:
\begin{align*}
\sum_{e \in E_h} (\vy_h^e)_+ \in \rho \cdot w_h \cdot \cP_+(F_h^-),\\
\sum_{e \in E_h} (\vy_h^e)_- \in \rho \cdot w_h \cdot \cP_+(F_h^+).
\end{align*}
\end{enumerate}
\end{definition}
The capacity constraints can be understood by analogy with the capacity constraints for multi-commodity flows over a graph. As in that case, the constraints in the definition ensure that flows from different commodity along the same arc do not cancel out.

We will use this notion of embeddability  to certify lower bounds to the ratio-cut objective for general polymatroidal cut functions. We will crucially make use of the following theorem in our development of the cut-matching game for solving ratio-cut problems over weighted submodular hypergraphs. Its proof is in Section~\ref{sec.appendix.omitted}.

\begin{restatable}{theorem}{flowembedding}
\label{thm.flow-embedding}
Let $G = \big( V, E_G, \vw^G, \vmu \big)$ be a weighted hypergraph equipped with a collection of polymatroidal cut functions $\{ \delta_h \}_{h \in E_G}$, and $H = \big( V, E_H, \vw^H \big)$ be a weighted directed graph. If $H \preceq_{\rho} G$, then
\begin{equation*}
\forall \vx \in \R^V\,,  \sum_{( i, j) \in E_H} w^H_{ij} \cdot (\vx_i - \vx_j)_{+}
\leq \rho \cdot \sum_{h \in E_G} w^G_h \cdot \overline{\delta}_h(\vx) \; \textrm{ and }\; \forall S \subseteq V\,,\, \delta_H(S) \leq \rho \cdot \delta_G(S).
\end{equation*}
For symmetric polymatroidal cut functions $\{ \delta_h \}_{h \in E_G}$ , then $H \preceq_{\rho} G$ implies the following bound on the undirected symmetrization $\tilde{H} = (V,E_{\tilde{H}}, \vw^{\tilde{H}})$:
$$
\forall \vx \in \R^V\,,  \sum_{\{i, j\} \in E_{\tilde{H}}} w^{\tilde{H}}_{ij} \cdot \lvert \vx_i - \vx_j \rvert
\leq 2 \cdot \rho \cdot \sum_{h \in E_G} w^G_h \cdot \overline{\delta}_h(\vx) \; \textrm{ and }\; \forall S \subseteq V\,,\, \delta_{\tilde{H}(S)} \leq 2 \cdot \rho \cdot \delta_G(S).
$$
\end{restatable}

We now state the main algorithmic result of this section, which allows us to construct hypergraph flow embedding from any single hypergraph flow. Just like in the case of graphs, this result is based on constructing a flow-path decomposition of the original flow. The particular form of polymatroidal cut functions plays monotonicity of the functions $F^+$ and $F^-$ plays a crucial role in bounding the congestion of the embedding, which cannot be controlled in the same way for general submodular cut functions. The proof appears in Appendix~\ref{sec.appendix.omitted}.

\begin{restatable}{theorem}{flowdecomposition}
\label{theorem.hypergraph-flow-decomposition}
Let $G = \big( V, E_G, \vw^G, \vmu \big)$ be a weighted hypergraph equipped with a collection of polymatroidal cut functions $\{ \delta_h \}_{h \in E_G}$. Given a hypergraph flow $\vY$ over $G$, there is an algorithm that, in time $\tilde{O}\big( \sparsity(G) \big)$, computes a weighted directed bipartite graph $H = \big( V, E_h, \vw^H \big)$ such that $H \preceq_{\cong_G(\vY)} G$ and $|E_H| = \tilde{O}\big( \sparsity(G) \big)$. Furthermore,
all vertices $i \in V$ with non-negative (resp. negative) demand $\dem_i(\vY)$ are source nodes (resp. sink nodes) in $H$, each with out-degree (resp. in-degree) equal $\dem_i(\vY).$
\end{restatable}

\section{An $O(\sqrt{\log n})$-Approximation via $\ell_2^2$-Metric Embeddings}
\label{sec.sdp-algorithm}

In this section, we prove Theorem~\ref{thm.main-metric-approx} by giving an $\ell_2^2$-metric relaxation and rounding argument yielding a randomized polynomial-time $O(\sqrt{\log n})$-approximation algorithm for the minimum ratio-cut problem over submodular hypergraphs equipped with polymatroidal cut functions. As a warm-up, we give the simpler argument for symmetric polymatroidal cut functions, before proving the result for general polymatroidal cut functions in Section~\ref{sec.general}

\subsection{Warm-up: Symmetric Polymatroidal Cut Functions}

We begin by providing a relaxation for the symmetric version of the problem, in which $F_h^+ = F_h^-$ for every $h \in E$. All the results in this subsection will focus exclusively on this special case, but they will be extended to the general case of all polymatroidal cut functions in the following subsection.

Given a weighted submodular hypergraph $G = (V, E, \vw, \vmu)$, we construct a vector-program relaxation of the continuous non-convex formulation in Equation~\eqref{eq.rc-noncvx} by associating to each vertex $i \in V \cup E$ of the factor graph $\hat{G}$
a vector $\vv_i$. For a hyperedge $h \in E$ and $i \in h,$ we can then relax the terms $|\vx_i - \nu_h\vone_h|$, which make up the argument of $\bar{F}_h$ in  Fact~\ref{fact.lovasz-extension-symmetric}, with the distance $\vd^h_i \defeq \|\vv_i - \vv_h\|^2.$ We obtain the following semidefinite program in its vector embedding form:
\begin{equation*}
\stepcounter{equation}
\tag{\texttt{RC-Metric}}
\label{eqn.vector-program}
\begin{aligned}
& \underset{\vv_i}{\min}
& & \sum_{h \in E} w_h \cdot \bar{F}_h ( \vd^h) &\\
& \textup{s.t.}
&& \hspace{-2mm}\sum_{\{i,j\} \subseteq V} {\mu(i) \mu(j) \over \mu(V)} \cdot \|\vv_i - \vv_j\|_2^2 \geq 1 &\\
& & & \|\vv_i -\vv_j \|^2_2 \leq \|\vv_i - \vv_k \|_2^2 + \|\vv_k - \vv_j\|_2^2
& \forall \, i, j, k \in V \cup E \\
& & &\vd^h_i =   \|\vv_i - \vv_h\|^2 & \forall h \in E, i \in h \\
& & & \vv_i \in \R^n, \vd^h \in \R^h_{\geq 0}
& \hspace{1cm} \forall \, i \in V \cup E, \, \forall \, h \in E
\end{aligned}
\end{equation*}

The following simple lemma is proved in Appendix \ref{sec.appendix.omitted}.

\begin{restatable}{lemma}{symrelaxation}
\label{lem.relaxation}
The \eqref{eqn.vector-program} vector program is, up to a constant, a relaxation of the minimum ratio-cut program for the hypergraph $G=(V,E,\vw,\vmu)$ with symmetric polymatroidal cut functions $\{\delta_h\}_{h \in E}.$
\end{restatable}

\paragraph{Rounding via Metric Embedding} To round the convex program \eqref{eqn.vector-program}, we apply the following well-known embedding result, which is implicit in the seminal work of Arora, Rao and Vazirani~\cite{ARV2009}.

\begin{theorem}[\cite{ARV2009}]
\label{thm:embedding}
Let $\{\vv_i \in \R^d \}_{i \in V \cup E}$ be a feasible solution to the $\ell_2^2$-metric relaxation \eqref{eqn.vector-program}. Then, there exists an efficiently computable $1$-Lipschitz\footnote{Recall that, given a metric space $(X,d)$, a function $f:X \to \R $ is said to be $K$-Lipschitz for some constant $K$ if $|f(x) - f(y)| \leq K\cdot d(x,y)$ for every $x,y \in X$.} map $\phi:\R^d \to \R$ satisfying:
$$
  \frac{\sum_{i, j \in V} \mu(i) \mu(j) \cdot |\phi(\vv_i)- \phi(\vv_j)|}
  {\sum_{i, j \in V} \mu(i) \mu(j) \cdot \|\vv_i - \vv_j\|^2}
  \geq \Omega \left(\frac{1}{\sqrt{\log |V|}}\right).
$$
\end{theorem}

To complete the proof of Theorem~\ref{thm.main-metric-approx} for the special case of symmetric polymatroidal cut functions, we make use of the embedding in Theorem~\ref{thm:embedding}. The proof crucially relies on the monotonicity of the Lov\'asz extension $\bar{F}$ (Fact~\ref{fct.lovaszmonotone}). The ability to perform this step can be taken as a justification for our definition of polymatroidal cut functions.

\begin{lemma}\label{lem.sqrt-logn-rounding-symmetric}
There is a polynomial-time algorithm that given a solution $\{\vv_i\}_{i \in V \cup E}$ to \eqref{eqn.vector-program} of objective value $\kappa :=  \sum_{h\in E} w_h \cdot \bar{F}_h(\vd^h)$, returns a cut $S \subseteq V$ such that:
\[
    \Psi_G(S) \leq O \big( \kappa \cdot \sqrt{\log |V|} \big) \, .
\]
\end{lemma}
\begin{proof}
Apply Theorem \ref{thm:embedding} to obtain $\vx \in \R^{V}$ and $\veta \in R^E$ with  $\vx(i) \defeq \phi(\vv_i)$ for all $i \in V$ and $\boldsymbol{\eta}(h) \defeq \phi(\vv_h) $. The $1$-Lipschitzness of $\phi$ guarantees, for every $h \in E$ and $i \in h$:
\begin{equation*}
  \lvert \vx(i) - \boldsymbol{\eta}(h) \rvert =
  \lvert \phi(\vv_i) - \phi(\vv_h) \rvert \leq \|\vv_i - \vv_h\|^2 = \vd^h_i.
\end{equation*}
By the monotonicity of $\bar{F}_h$ in Fact~\ref{fct.lovaszmonotone}, we have:
\begin{equation}\label{eq:rounding-numerator}
  \sum_{h \in E} w_h \cdot \min_{\nu_h \in \R} \bar{F}_h(|\vx_h - \nu_h \vone_h|)
  \leq  \sum_{h \in E} w_h \cdot \bar{F}_h(|\vx_h - \boldsymbol{\eta}(h) \vone_h|)
  \leq \sum_{h \in E} w_h \cdot \bar{F}_h(\vd^h)
  = \kappa.
\end{equation}
It now suffices to obtain a lower bound on the denominator in the continuous formulation~\eqref{eq.rc-noncvx}:
\begin{equation}\label{eq:rounding-denominator}
  {\min_{\gamma \in \R} \|\vx - \gamma \ones\|_{1, \vmu}} \geq {1\over 2}{\sum_{i,j\in V} {\vmu(i) \vmu(j) \over \vmu(V)}\cdot |\vx_i - \vx_j|}
  \geq  \Omega\left( {1 \over \sqrt{\log |V|}}\right) {\sum_{i,j\in V} {\vmu(i) \vmu(j) \over \vmu(V)} \| \vv_i - \vv_j\|^2_2 }
  \geq \Omega\left( {1 \over \sqrt{\log |V|}}\right).
\end{equation}
Finally, Equations \eqref{eq:rounding-numerator} and \eqref{eq:rounding-denominator}, yield, together with Fact~\ref{fact.lovasz-extension-symmetric}:
\[
   {\sum_{h \in E} w_h  \bar{\delta}_h (\vx) \over \min_{\gamma \in \R} \|\vx - \gamma\ones\|_{1, \vmu}} \leq 2 \cdot
  {\sum_{h \in E} w_h \min_{\nu \in \R} \bar{F}_h(|\vx - \nu 1|) \over \min_{\gamma \in \R} \|\vx - \gamma\ones\|_{1, \vmu}}
  \leq
  O(\kappa \cdot \sqrt{\log n}).
\]
The vector $\vx$ can then be efficiently rounded to a subset $S\subseteq V$ by Lemma~\ref{lem.continuous-to-discrete-equivalence}.
\end{proof}

\subsection{General Polymatroidal Cut Functions}
\label{sec.general}

For the case of general polymatroidal cut function, our metric relaxation must be able to capture the signed terms $(\vx_h - \nu_h \vone_h)_+$ and $(\vx_h - \nu_h \vone_h)_-$ that appear in the Lov\'asz extension in Fact~\ref{fact.form-of-lovasz-extension}. For this purpose, we will use the directed $\ell_2^2$-semimetrics introduced by Charikar, Makarychev and Makarychev~\cite{charikar2006directed} and, for every $h \in E$ and $i \in h$, relax $(\vx_i - \nu_h)_+$ to $\vd^{h,-}_i \defeq \norm{\vv_i - \vv_h}^2 + \norm{\vv_i}_2^2- \norm{\vv_h}_2^2$ and $(\vx_i - \nu_h)_-$ to $\vd^{h,+}_i \defeq \norm{\vv_i - \vv_h}^2 + \norm{\vv_h}_2^2- \norm{\vv_i}_2^2$.
We obtain the following semidefinite program:
\begin{equation*}
\stepcounter{equation}
\tag{\texttt{RC-Directed-Semimetric}}
\label{eqn.gen-vector-program}
\begin{aligned}
& \underset{\vv_i, \vell^h}{\min}
& & \sum_{h \in E} w_h \cdot \left(\bar{F}_h^-(\vd_h^-) + \bar{F}_h^+(\vd_h^+) \right)\\
& \textup{s.t.}
& & \sum_{\{i,j\} \subseteq V} {\mu(i) \mu(j) \over \mu(V)} \cdot \|\vv_i - \vv_j\|_2^2 \geq 1 \\
& & & \|\vv_i -\vv_j \|^2_2 \leq \|\vv_i - \vv_k \|_2^2 + \|\vv_k - \vv_j\|_2^2
& \forall \, i, j, k \in V \cup E \\
& & & \vd^{h,-}_i = \norm{\vv_i - \vv_h}^2 + \norm{\vv_i}_2^2- \norm{\vv_h}_2^2
& \forall h \in E, i \in h\\
& & & \vd^{h,+}_i = \norm{\vv_i - \vv_h}^2 + \norm{\vv_h}_2^2- \norm{\vv_i}_2^2
& \forall h \in E, i \in h\\
& & & \vv_i \in \R^n, \vd^{h,+}, \vd^{h,-} \in \R^h_{\geq 0}
& \forall \, i \in V \cup E, \, \forall \, h \in E
\end{aligned}
\end{equation*}

We have the following lemma, which we prove in  Appendix~\ref{sec.appendix.omitted}.
\begin{restatable}{lemma}{genrelaxation}
    The \eqref{eqn.gen-vector-program} vector program is, up to a constant, a relaxation of the minimum submodular hypergraph ratio-cut problem for the hypergraph $G=(V,E,\vw,\vmu)$ with polymatroidal cut functions $\{\delta_h\}_{h \in E}.$
\end{restatable}

\paragraph{Rounding via Metric Embedding}
To round a solution to the program \eqref{eqn.gen-vector-program} to a solution of the original ratio-cut problem \eqref{eq.rc-noncvx}, we require a version of the embedding result of Theorem~\ref{thm:embedding} for directed semimetrics. The following result is a simple consequence of the rounding algorithm of Agarwal~\etal~\cite{agarwal2005log} for a similar relaxation of directed expansion. We prove it for completeness in Appendix~\ref{sec.appendix.omitted}.

\begin{restatable}{theorem}{directedembedding}
\label{thm:directed-embedding}
    Given a feasible solution $\{\vv_i \in \R^d\}_{i\in V\cup E}$ to \ref{eqn.gen-vector-program} and a measure $\vmu \in \R_{\geq 0}^V$, there exists a polynomial-time computable map $\phi:\R^d \to \R$ such that:
    \begin{enumerate}
        \item $(\phi(\vv)- \phi(\vw))_+ \leq \norm{\vv - \vw}^2 + \norm{\vv}^2 - \norm{\vw}^2$ for all $\vv, \vw \in \R^d$,
        \item and:
        \[
        \frac{\sum_{i,j \in V} \mu(i) \mu(j) |\phi(\vv_i) - \phi(\vv_j)|} {\sum_{i,j \in V} \mu(i) \mu(j) \norm{\vv_i - \vv_j}^2_2 } \geq \Omega \left(\frac{1}{\sqrt{\log |V|}}\right)
        \]
    \end{enumerate}
\end{restatable}

We can now complete the proof of Theorem~\ref{thm.main-metric-approx} for the general polymatroidal case by showing that the map $\phi$ in Theorem~\ref{thm:directed-embedding} yields a $O(\sqrt{\log n})$-approximate solution to the minimum ratio-cut problem.

\begin{lemma}
    Given a solution $(\{\vv_i\}_{i\in V \cup \in E})$ to \ref{eqn.gen-vector-program} of value $\kappa$, we can produce in polynomial time a set $S \subseteq V$ such that:
    $$
    \Psi_G(S) \leq O(\kappa \cdot \sqrt{\log |V|}).
    $$
\end{lemma}
\begin{proof}
Let $\phi$ be the map whose existence is guaranteed by Theorem~\ref{thm:directed-embedding}. Let $\vx\in \R^{V}$ and $\veta \in \R^E$  be the vector given by $\vx_i = \phi(\vv_i)$ for all $i \in V$ and $\eta_h = \phi(\vx_h)$ for all $h \in E.$ By the first condition in Theorem~\ref{thm:directed-embedding}, we have, for all $h \in E$ and $i \in h$:
$$
(\vx_i - \eta_h)_+ \leq \vd^{h,-}_i \, \textrm{ and } \, (\vx_i -\eta_h)_- \leq \vd^{h,+}_i.
$$
We can now exploit the monotonicity of the functions $F^{-}_h$ and $F^{+}_h$ associated with the polymatroidal cut functions $\delta_h$. By Fact~\ref{fact.form-of-lovasz-extension} and Fact~\ref{fct.lovaszmonotone}, we have:
\begin{align*}
    \sum_{h\in E} w_h \overline{\delta}_h(\vx) &=
    \sum_{h\in E} w_h \min_{\nu_h \in \R}\left\{\bar{F}^-_h((\vx_h-\nu_h\vone_h)_+) + \bar{F}^+_h((\vx_h-\nu_h\vone_h)_-)\right\} \\
    & \leq \sum_{h\in E} w_h \left\{\bar{F}^-_h((\vx_h-\eta_h\vone_h)_+) + \bar{F}^+_h((\vx_h-\eta_h \vone_h)_-)\right\}\\
    & \leq \sum_{h\in E} w_h \left(\bar{F}^-_h(\vd^{h,-}) + \bar{F}^+_h(\vd^{h,+})\right) = \kappa.
    \end{align*}
We can now lower bound the denominator in the objective of Equation~\eqref{eq.rc-noncvx} as in Equation~\ref{eq:rounding-denominator}. Finally, we have:
\[
   {\sum_{h \in E} w_h  \bar{\delta}_h (\vx) \over \min_{\gamma \in \R} \|\vx - \gamma\ones\|_{1, \vmu}}
  \leq
  O(\kappa \cdot \sqrt{\log n}).
\]
The vector $\vx$ can then be efficiently rounded to a subset $S\subseteq V$ by Lemma~\ref{lem.continuous-to-discrete-equivalence}.
\end{proof}

\section{Localized Convex Formulations and Ratio-Cut Improvement}
\label{sec.ci}

In this section, we consider the localized formulation~\eqref{eq.rc-local} introduced in Section~\ref{sec.results.ci}, which we rewrite as a convex program. We show that the dual solutions to this problem are hypergraph flows, which can be used to obtain hypergraph flow embeddings into the input graph. Finally, we give algorithms that construct approximate primal-dual solutions for the \eqref{eq.rc-local} problem. These algorithms will play an important part in our cut-matching game.

We start by introducing the \emph{ratio-cut improvement} objective, a local partitioning objective parametrized by a seed $\vs \in R^V$, which is the integral analogue of the \eqref{eq.rc-local} problem:
\begin{definition}[Ratio-Cut Improvement Problem]
\label{def.rc-improve}
Given a weighted submodular hypergraph $G = (V, E, \vw, \vmu)$, cut functions $\{ \delta_h \}_{h \in E}$, and a \emph{seed vector} $\vs \in \R^V$ with $\langle \vs, \ones \rangle_\mu = 0$ and $\lVert \vs \rVert_\infty \leq 1,$ the ratio-cut improvement objective $\Psi_{G,s}$ is defined for all $S \subset V$ as:
\begin{align*}
\Psi_{G, \vs}(S)
&
\defeq
\frac{\sum_{h \in E} w_{h} \cdot \delta_h(S\cap h)} {\max\{0,\langle \vs, \ones^S \rangle_{\mu}\}} \\
\Psi^*_{G, \vs}
&\defeq \min_{S \subseteq V} \Psi_{G, \vs}(S)
\end{align*}
The ratio-cut improvement problem asks to find $S \subseteq V$ achieving $\Psi^*_{G, \vs} = \Psi_{G, \vs}(S)$.
\end{definition}

The ratio-cut improvement problem can be thought of as attempting to find a cut $(S,\overline{S})$ which has both a small boundary and high correlation with the input seed $\vs$ simultaneously. In Appendix~\ref{sec.appendix.andersen-lang}, we demonstrate that the problem of minimizing $\Psi_{G,\vs}$ is precisely the generalization of the cut-improvement problem of Andersen and Lang~\cite{Andersen-Lang} to weighted submodular hypergraphs and fractional seeds.

The relation between the minimum ratio-cut problem, the ratio-cut improvement problem and their continuous formulations ~\eqref{eq.rc-noncvx} and ~\eqref{eq.rc-local} is formalized in the following lemma.

\begin{lemma}
\label{lem.rc-improve.relaxation}
For all $\vx \in \R^V$, and seeds $\vs \in \R^V$ satisfying $\lVert \vs \rVert_{\infty} \leq 1$ and $\langle \vs, \ones \rangle_{\vmu} = 0$, we have:
\begin{equation*}
\bar{\Psi}_G(\vx) \leq \bar{\Psi}_{G, \vs}(\vx) \, .
\end{equation*}
In particular, setting $\vx = \ones^S$ yields $\Psi_{G, \vs}^*$ is a relaxation of $\Psi_G^*$ for any $S \subset V$:
\begin{equation*}
\Psi_G(S) \leq \Psi_{G,\vs}(S) \, .
\end{equation*}
\end{lemma}
\begin{proof}
By the duality of conjugate norms, we have
\begin{equation*}
0 \geq \min_{u \in \R} \lVert \vx - u \ones \rVert_{1, \vmu}
= \min_{u \in \R} \max_{\lVert \vy \rVert_\infty \leq 1} \langle \vy, \vx - u \ones \rangle_{\vmu}
= \max_{\substack{\lVert \vy\rVert_\infty \leq 1 \\ \langle \vy, \ones \rangle_{\vmu} = 0}} \langle \vy, \vx \rangle
\geq  \langle \vs, \vx \rangle_{\vmu}
\end{equation*}
for all $\vx \in \R^V$. Consequently,
\begin{equation*}
\bar{\Psi}_G(\vx)
= \frac{\sum_{h \in E} w_h \bar{\delta}_h(\vx)}{\min_{u \in \R} \lVert \vx - u \ones \rVert_{1, \vmu}}
\leq \frac{\sum_{h \in E} w_h \bar{\delta}_h(\vx)}{\max\{0,\langle \vs, \vx \rangle_{\vmu}\}}
= \bar{\Psi}_{G, \vs}(\vx)
\end{equation*}
as required.
\end{proof}

Of critical, algorithmic, consequence is the fact that computing the value of $\bar{\Psi}_{G, \vs}^*$ can be written as a convex program. We refer to this program as~\eqref{eq.rc-improve.primal}:
\begin{equation*}
\stepcounter{equation}
\tag{\texttt{CI-Primal}}
\label{eq.rc-improve.primal}
\begin{aligned}
& \underset{\vx \in \R^V}{\min}
& & \sum_{h \in E} w_h \cdot \bar{\delta}_{h}(\vx) \\
& \textup{s.t.}
& & \langle \vs, \vx \rangle_{\mu} = 1
\end{aligned}
\end{equation*}

The following lemma is a simple consequence of the positive homogeneity of the Lov\'asz extensions of submodular functions (Proposition 3.1 in~\cite{bach2013learning}).
\begin{lemma}
\label{lem.rc-improve.primal-val}
The optimum value of \eqref{eq.rc-improve.primal} equals $\Psi^*_{G,\vs}$.
\end{lemma}

With this convex program, one can leverage convex duality to construct lower bounds for $\bar{\Psi}^*_{G, \vs}$, and consequently $\Psi^*_{G, \vs}$. We consider a dual formulation of \eqref{eq.rc-improve.primal} as a problem over hypergraph flows:

\begin{equation*}
\stepcounter{equation}
\tag{\texttt{CI-Dual}}
\label{eq.rc-improve.dual}
\begin{aligned}
&  {\max}
& & \alpha \\
& \textup{s.t.}
& & \dem(\vY) = \alpha \cdot \big( \vmu \circ \vs \big) \\
& & & \cong_G(\vY) \leq 1\\
& & &  \alpha \in \R, \vY =\{\vy_h \in \R^h \perp \vone_h\}_{h\in E}
\end{aligned}
\end{equation*}

In this setting, strong duality between \eqref{eq.rc-improve.primal} and \eqref{eq.rc-improve.dual} holds by virtue of the~\eqref{eq.rc-improve.primal} having a finite number of polyhedral constraints. The full proof of the following lemma can be found in Appendix~\ref{sec.appendix.omitted}.
\begin{restatable}{lemma}{improveduality}
\label{lem.rc-improve.strong-duality}
The programs \eqref{eq.rc-improve.primal} and~\eqref{eq.rc-improve.dual} are dual to each other. Strong duality holds.
\end{restatable}
\noindent

From this point onwards, we restrict our attention to the kind of seed vectors that we will encounter in our application to the cut-matching game. For this purpose, we define the following notation for disjoint sets $A, B \subseteq V$ with $\mu(A) \leq \mu(B)$:
$$
\vone^{A,B} = \vone_A - \frac{\mu(A)}{\mu(B)} \vone_B.
$$
Notice that $\vone^{A,B}$ is a valid seed for the local formulation of the minimum ratio-cut problem as $\|\vone^{A,B}\|_\infty \leq 1$ and $\langle \vone^{A,B} , \vone \rangle_{\vmu} = 0.$
For this kind of seed, the dual problem can be interpreted as a maximum concurrent single-commodity flow over the hypergraph $G$, in which we each vertex $i \in A$ is asked to route demand $\mu_i$ to $B$ and each vertex $j$ in $B$ is asked to receive demand $\nicefrac{\mu(A)}{\mu(B)} \cdot \mu_j$ from $A$. For a solution of value $\alpha$, each vertex will be able to route an $\alpha$ fraction of its demand and the total demand routed from $A$ to $B$ will be $\alpha \mu(A).$
Notice that, in the graph case, for $B = \bar{A}$, this construction specializes exactly to the undirected maximum concurrent flow problem used by Andersen and Lang~\cite{Andersen-Lang} to solve the ratio-cut improvement problem.

\paragraph{Approximate Solutions}
In Section~\ref{sec.solvers}, we will discuss how, for many kinds of polymatroidal cut functions, approximately feasible dual solutions to \eqref{eq.rc-improve.dual} can be obtained by approximately solving a maximum flow problem over the factor graph or an augmented version of it.
For this purpose, it is convenient to introduce a notion of approximate dual solution of value $\alpha$ to \eqref{eq.rc-improve.dual}.

\begin{definition}[Approximate dual solution]
\label{def.approximate-primal-dual-solution}
Let $G = \big( V, E_G, \vw^G, \vmu \big)$ be a weighted submodular hypergraph equipped with polymatroidal cut functions $\{\delta_h\}_{h\in E}$, and $A,B \subseteq V$ disjoint sets with $\mu(A) \leq \mu(B)$. An \emph{approximate dual solution of value $\alpha \geq 0$} for $(G,\vone^{A,B})$ is a hypergraph flow $\vY = \big(\vy^h \in \R^h \perp \vone_h\big)$
such that the following hold simultaneously:
\begin{enumerate}
\item $\vY$ has unit congestion, i.e., $\cong_G(\vY) \leq 1$\,,
\item $\vY$ routes a $\nicefrac{1}{2}$-fraction of the required flow from $A$ to $B$, i.e., at least $\nicefrac{1}{2} \cdot \alpha \cdot \mu(A) $  flow,
\item $\vY$ does not exceed the require demand at every vertex, i.e., $|\dem (\mY)| \leq \alpha \cdot |\vmu \circ \vone^{A,B}|$\,.
\end{enumerate}
\end{definition}

\subsection{Graph Certificates from Approximate Dual Solutions}

As discussed in Section~\ref{sec.results.ci}, our plan is to apply the cut-matching game to combine local lower bounds given by approximately optimal solutions to \eqref{eq.rc-improve.dual} into global lower bounds for the non-convex problem $\bar{\Psi}^*_{G}$.
To carry out this plan, we need to express our local lower bounds, i.e. dual solutions to \eqref{eq.rc-improve.dual} in terms of hypergraph embeddings of graphs into the instance hypergraph.
\begin{definition}
\label{def.graph}
Let $G = \big( V, E_G, \vw^G, \vmu \big)$ be a weighted submodular hypergraph equipped with polymatroidal cut functions $\{\delta_h\}_{h\in E}$, and $A,B \subseteq V$ disjoint sets with $\mu(A) \leq \mu(B).$ Let $\alpha = \Psi^*_{G, \vone^{A,B}}.$
An \emph{approximate dual graph certificate of value $\alpha \geq 0$} for the \eqref{eq.rc-local} problem with seed $\vone^{A,B}$ is a directed graph $D=(V,E_D, \vw^D)$ with the following properties:
\begin{enumerate}
    \item $D$ is sparse, i.e., $|E_D| \leq \tilde{O}(\sparsity(G)),$
    \item $\alpha \cdot D$ embeds into $G$ with congestion $1$, or equivalently $D \preceq_{\nicefrac{1}{\alpha}} G,$
    \item $D$ is bipartite from $A$ to $B$, i.e. all arcs of $D$ go from $A$ to $B$,
    \item for all $i \in A, j \in B$, we have the degree bounds $\deg_D(i) \leq \mu_i$ and $\deg_D(j) \leq \nicefrac{\mu(A)}{\mu(B)} \cdot \mu_j$,
    \item $D$ is large, i.e., $w^D(E(A,B)) \geq \nicefrac{1}{2} \cdot \mu(A)$.
\end{enumerate}
Following Theorem~\ref{thm.flow-embedding}, we may assume that $D$ is undirected if all cut functions are symmetric.
\end{definition}
In the cut-matching game framework, we will use  dual graph certificates to iteratively construct lower bounds to $\Psi^*_G$ through the application of Theorem~\ref{thm.flow-embedding}.
For example, given an approximate dual graph certificate $D$ of value $\alpha$, Theorem~\ref{thm.flow-embedding} yields the lower bound:
$$
\forall S \subseteq V\,, \delta_G(S) \geq \alpha \cdot \delta_D(S) = \alpha \cdot w^D(E(S \cap A, \bar{S} \cap B)).
$$
As we built $D$ from a dual solution for \eqref{eq.rc-improve.dual} with seed $\vone^{A,B}$, this lower bound is tighter for cuts $S$ that are well-correlated with $\vone^{A,B}$, reaching its maximum when $A \subseteq S \subseteq \bar{B}.$

We can now apply the flow-path decomposition of Theorem~\ref{theorem.hypergraph-flow-decomposition} to show how to efficiently construct an approximate dual graph certificate from an approximate dual solution. The straightforward proof is found in Appendix~\ref{sec.appendix.omitted}.
\begin{restatable}{lemma}{dualgraphcertificate}
\label{lem.dual-to-embedding}
Under the same assumptions of Definition~\ref{def.graph}, consider an approximate dual solution $\vY$ of value $\alpha$ to the ratio-cut improvement problem on $(G,\vone^{A,B}).$ Let $H$ be the directed graph obtained by applying the flow-path decomposition algorithm of Theorem~\ref{theorem.hypergraph-flow-decomposition} to the hypergraph flow $\vY$. Then, the scaled graph $\frac{1}{\alpha}\cdot H$ is an approximate dual graph certificate of value $\alpha.$
\end{restatable}

\subsection{Approximately Solving \eqref{eq.rc-local}}
\label{sec.solvers}

In this section, we discuss algorithms that construct approximate solutions to \eqref{eq.rc-improve.primal} and~\eqref{eq.rc-improve.dual}, the primal and dual convex formulations of \eqref{eq.rc-local} . Our specific focus is to describe algorithms that implement the following specification:
\begin{definition}
\label{def.oracle}
Given a weighted submodular hypergraph $G = (V,E,\vw, \vmu)$ with polymatroidal cut functions $\{\delta_h\}_{h\in E}$ and disjoint sets $A, B \subset V$ with $\mu(A) \leq \mu(B)$,
an \emph{approximate primal-dual oracle} problem is an algorithm $\cA_{\ci}$ that takes as input $G$ and the seed vector $\vone^{A,B}$.  For some $\alpha \geq 0,$ the algorithm $\cA_{\ci}$ outputs both:
\begin{enumerate}
    \item a cut $S \subseteq V$ with $\Psi_{G,\vone^{A,B}}(S) \leq 3\alpha,$
    \item an approximate dual graph certificate of value $\alpha.$
\end{enumerate}
\end{definition}

We prove the following two theorems in Appendix~\ref{sec.appendix.ci-algs} by a simple application of binary search over  $\alpha.$
The first thereom bounds the complexity of constructing an approximate primal-dual oracle for general polymatroidal cut functions.
\begin{theorem}
\label{thm.general-solver}
For submodular hypergraphs with general polimatroidal cut functions, an approximate primal-dual oracle can be implemented in time $\tilde{O}\big( |V|\cdot\sum_{h\in E}\Theta_{h}\big),$
where $\Theta_h$ is the running time for a quadratic optimization oracle for $\delta_h,$ by solving $O(\log |V|)$-many decomposable submodular minimization problems.
\end{theorem}

Our second theorem shows how this result can be significantly improved in the case of the standard hypergraph cut functions by relying on almost-linear-time maximum flow solvers~\cite{bernstein2022deterministic,chenMaximumFlowMinimumCost2022}.

\begin{theorem}
\label{thm.maxflow-solver}
For a submodular hypergraph $G$ equipped with the directed or standard hypergraph cut functions,
an approximate primal-dual oracle can be implemented in almost linear-time by solving $O(\log |V|)$-many $\nicefrac{1}{2}$-approximate maximum flow problems over a directed, vertex-capacitated version of the factor graph $\hat{G}.$
\end{theorem}

\section{An \texorpdfstring{$O(\log n)$}{O(log n)}-Approximation via Generalized Cut-Matching Games}
\label{sec.alg-cm}

We now describe how to use our cut-matching game to approximate minimum ratio-cut.
Let us begin by stating the precise approximation guarantee garnered by running a game between a good cut strategy, and an approximate primal-dual solver for the ratio-cut improvement problem~\eqref{eq.rc-local}.

\begin{restatable}{theorem}{cmapprox}
\label{thm.cm.approx}
Let $G=(V, E, \vw, \vmu)$ be an input submodular hypergraph equipped with polymatroidal cut functions. Consider an execution of the cut-matching game in which the cut player $\cC$ plays an $\big( f(n), g(n) \big)$-good cut strategy, while the matching player plays the dual graph certificates output by an approximate primal-dual oracle $\cA_{\ci}$ (Definition~\ref{def.oracle}) on $G$.
Let $C_t$ be the cuts output by $\cA_{\ci}$.
Then, with constant probability, we have the following approximation result:
$$
\min \bigg\{\Psi_G(C_t)\bigg\}_{t=1}^{g(n)} \leq O\bigg(\frac{g(n)}{f(n)}\bigg) \cdot \Psi^*_G.
$$
\end{restatable}

\noindent
This is a standard result for cut-matching games and its proof does not deviate too far from that in previous work. We give its proof in Appendix~\ref{sec.cut-strategy.reduction-proof} for completeness.

Theorem~\ref{thm.cm.approx} reduces the task of approximating minimum ratio-cut to constructing a good cut strategy. Hence, this section is devoted towards proving the following theorem.

\begin{restatable}{theorem}{cmGoodCutStrategy}
\label{thm.cm.good-cut-strategy}
Let $H = \big(V, E_H, \vw^H, \vmu \big)$ be the state graph for an $(n, \vmu)$-generalized cut-matching game. There exists a cut strategy $\cA_{\cut}$ satisfying the following:
\begin{enumerate}
\item If $H$ is undirected, then $\cA_{\cut}$ is $\left( \Omega (\log n), O ( \log^2 n ) \right)$-good with probability $O(1)$.

\item If $H$ is directed, then $\cA_{\cut}$ is $\left( \Omega (\log^2 n), O ( \log^3 n ) \right)$-good with probability $O(1)$.
\end{enumerate}
Furthermore, $\cA_{\cut}$ outputs a cut action in time $\tilde{O}\big( \sparsity(H) \big)$.
\end{restatable}

To prove this theorem, we give an optimization-based analysis of the cut-matching game by using regret minimization techniques from online optimization to construct and analyze $\cA_{\cut}$. This approach is very natural: the goal of online optimization is to adaptively produce a sequence of actions that will lead to some terminal outcome despite the presence of adversarial response. For the cut-matching game, the cut player's goal is to produce a short sequence of cut actions such that, no matter what approximate dual responses are given by $\cM$, the state graph $H$ will have large ratio-cut objective.

This approach also provides insight into why earlier combinatorial restrictions on valid cut and matching responses are not necessary when using the cut-matching game to partition (hyper)graphs. An interesting byproduct of our construction, in the case where $\{ \delta_h \}_{h \in E}$ are asymmetric, is that $\cM$ does not necessarily route Eulerian demand graphs. Additionally, no further reductions are required to enforce these combinatorial restrictions when approximating ratio-cut objectives beyond expansion.

To prove Theorem~\ref{thm.cm.good-cut-strategy}, we first describe our construction of $\cA_{\cut}$ for the case where the cut functions defining the ratio-cut objective are symmetric, and $H$ is undirected. Simply considering this case already demonstrates a large fraction of our technical contributions. We then consider asymmetric cut functions and directed $H$, highlighting key differences from the symmetric case.

\subsection{An \texorpdfstring{$\left( \Omega (\log n), O \big( \log^2 n \big) \right)$}{(Omega(log n), O(log\^2 n))}-good Cut Strategy for Symmetric Cut Functions}

We first consider the case where the polymatroidal cut functions $\{ \delta_h \}_{h \in E}$ are symmetric.

Let us start by instantiating the online linear optimization setting in the context of our cut-matching game. For an instance of a $(n, \vmu)$-generalized cut-matching game, fix the action set to be $\cX = \Delta^{n \times n}_{\vmu}$ where
\begin{equation*}
\Delta^{n \times n}_{\vmu}
\defeq \big\{ \mX \succeq \mZero : \big\langle \mL(K_{\vmu}), \mX \big\rangle \geq 1 \big\} \, .
\end{equation*}
At the end of round $t$, the matching player $\cM$ will have produced approximate matching responses $D_1, \ldots, D_t$. Choose the $t$-th feedback matrix to be $\mF_t = \mL(D_t)$. The loss incurred during round $t$ is then
\begin{equation*}
\big\langle \mL(D_t), \mX \big\rangle \, .
\end{equation*}

We instantiate online linear optimization in this way because any bound on regret now translates to a lower bound on the ratio-cut objective for the state graph $H$. When $\{ \delta_h \}_{h \in E}$ are symmetric, $H$ is an undirected graph, and $\Psi_H^*$ is simply the undirected ratio-cut objective, which has the following trivial SDP relaxation:
\begin{lemma}
\label{lem.sym-cut-strat.relaxation}
Given any undirected graph $G = \big( V, E_G, \vw^G, \vmu \big)$, we have $\frac{\opt}{2} \leq \Psi_G^*$ where $\opt$ is the optimal objective value for the following semidefinite program.
\begin{equation*}
\stepcounter{equation}
\tag{\texttt{cond-sdp}}
\label{eq.cond-sdp}
\begin{aligned}
& {\min}
& & \big\langle \mL(G), \mX \big\rangle \\
& \textup{s.t.}
& & \big\langle \mL(K_{\vmu}), \mX \big\rangle \geq 1 \\
& & & \mX \succeq \mZero
\end{aligned}
\end{equation*}
\end{lemma}

\noindent
Our choice of action set thus coincides with the feasible set for~\eqref{eq.cond-sdp}, while the feedback matrices are chosen so that the sum of the losses give the SDP objective for the state graph. This is as intended; we want the cut player to force a state graph with large ratio-cut objective, thus it is crucial to have some lower bound on the ratio-cut objective that is a function of both cut, and matching responses. This is also to be expected; if we wanted regret minimizing strategies to produce approximate solutions for~\eqref{eq.cond-sdp}, then we would have defined the action set and losses in exactly the same way.

Three issues remain if we wish to use a regret bound from Theorem~\ref{thm.regret-bound} to analyze a good cut strategy:
\begin{enumerate}[(1)]
\item \emph{Ensuring Small Width}: We need to ensure that the width term $\big\lVert \mM^{-1/2} \mL(D_t) \mM^{-1/2} \big\rVert_{\spectral}$ is small given any approximate matching response $D_t$.

\item \emph{Ensuring Large Loss}: For our cut strategy to be good, we will need to specify how to produce a cut action such that the loss $\big\langle \mL(D_t), \mX_t \big\rangle$ regardless of what approximate matching response $D_t$ is played.

\item \emph{Implementing $\mmwu_{\eta, \vmu}$ in Nearly-linear Time}: If we want our cut player to run in nearly-linear time, then we will (a) need to compute the update $\mmwu_{\eta, \vmu}$ in nearly-linear time, and (b) ensure that queries on the vector embedding given by the Gram matrix of the update can be computed in sublinear time.
\end{enumerate}
\noindent
Let us address these issues now.

\subsubsection{Ensuring Small Width}
The width term is constant for any approximate matching response. This follows simply from the fact that the degree of any vertex $i$ is bounded by $\mu_i$.

\begin{lemma}
\label{lem.sym-cut-strat.width}
Any approximate matching response $D = \big( V, E_D, \vw^D, \vmu \big)$ satisfies $\mZero \preceq \mL(D) \preceq 2 \cdot \mM$.
\end{lemma}

\noindent
Notice that, since $\mL(D) \succeq \mZero$, the regret guarantee in Theorem~\ref{thm.regret-bound} holds for any step-size $\eta > 0$.

\subsubsection{Ensuring Large Loss}
Let $\big\{ \vv^{(t)}_i \big\}_{i \in V} \subseteq \R^d$ be a vector embedding of $V$ whose Gram matrix is $\mX_t$. If $D_t = \big( V, E_{D_t}, \vw^{D_t} \big)$ is the approximate matching response at time $t$, then the loss $\big\langle \mL_{D_t}, \mX_t \big\rangle$ is equivalent to
\begin{equation*}
\big\langle \mL_{D_t}, \mX_t \big\rangle
= \sum_{(i,j) \in E_{D_t}} w_{ij}^{D_t} \cdot \big\lVert \vv^{(t)}_i - \vv^{(t)}_j \big\rVert^2 \, .
\end{equation*}
For our cut strategy to be good, we need to produce a cut action $(A_t, B_t)$ such that, no matter how $D_t$ is played, the edges of $D_t$ will always be placed across pairs of embedding vectors $\vv^{(t)}_i$, $\vv^{(t)}_j$ that are well-separated under the $\ell_2^2$-norm. This can be difficult to guarantee if the flow routed is approximate. If the primal-dual solver can only route a constant factor approximate flow, then an adversarially chosen constant fraction of the demand may be dropped. If one naively rounds $(A_t, B_t)$ using a random 1-dimensional projection, as in \cite{arora2007combinatorial, orecchia2012balanced_separator}, then the rounding can produce an unbalanced cut, upon which the adversary can choose to drop a fraction of demand that routes between \emph{all} vertex pairs whose embedding vectors are far apart in $\ell_2^2$-distance.

Previous works~\cite{chuzhoyNewAlgorithmDecremental2019, chuzhoy2020deterministic, nanongkai2017dynamic} have addressed this issue by routing a polylogarithmic number of $O(1)$-factor approximate maximum flows during a single iteration of the cut-matching game
However, a technical contribution of this construction is that a \emph{single} approximate primal-dual solve of~\eqref{eq.rc-local} suffices, provided $(A_t, B_t)$ are $\sigma$-separated in the sense of Definition~\eqref{def.separated-sets}.

The following lemma shows that any approximate matching response routing demand across $\sigma$-separated $(A_t, B_t)$ must always route a large amount of flow between embedding vectors that are well separated in $\ell_2^2$-distance. Using $(A_t, B_t)$ as the cut action will thus always result in large loss.

\begin{lemma}
\label{lem.sym-cut-strat.loss}
Given $\vmu \in \Z_{\geq 0}$, let $\mX \in \Delta_{\vmu}^{n \times n}$ and $\{ \vv_i \}_{i \in V} \subseteq \R^d$ be the vector embedding whose Gram matrix is $\mX$. If $(A, B)$ is $\sigma$-well separated with respect to $\vmu$ and $\{ \vv_i \}_{i \in V}$, then for any approximate matching response $D = \big( A \cup B, E_D, \vw^D \big)$ with respect to $(A, B)$, we have:
\begin{equation*}
\sum_{\{ i, j \} \in E_D} w_{ij}^{D} \cdot \lVert \vv_i - \vv_j \rVert_2^2
\geq {\sigma \over 2} \, .
\end{equation*}
\end{lemma}
\begin{proof}
Since $(A, B)$ are $\sigma$-separated we have for all $i \in A$ and $j \in B$:
\begin{equation*}
\frac{\vmu(A) \vmu(B)}{\vmu(V)} \cdot \lVert \vv_i - \vv_j \rVert^2_2
\geq \sigma \cdot \sum_{i,j \in \binom{V}{2}} \frac{\mu_i \mu_j}{\vmu(V)} \cdot \lVert \vv_i - \vv_j \rVert^2_2
= \sigma \cdot \big\langle \mL(K_{\vmu}), \mX \big\rangle
\geq \sigma \, .
\end{equation*}
Here, we use $\mX \in \Delta_{\vmu}^{n \times n}$ to deduce $\big\langle \mL(K_{\vmu}), \mX \big\rangle \geq 1$. Rearranging the above inequality then implies
\begin{equation*}
\lVert \vv_i - \vv_j \rVert_2^2
\geq \frac{\sigma}{\vmu(A)} \cdot \frac{\vmu(V)}{\vmu(B)}
\geq \frac{\sigma}{\vmu(A)}
\qquad\qquad
\forall \, i \in A, j \in B
\end{equation*}
where the last inequality follows by $\vmu(B) \leq \vmu(V)$. Finally, any approximate matching response is bipartite across $(A, B)$, and has large total weight. Thus we can conclude
\begin{equation*}
\sum_{\{ i, j \} \in E_D} w_{ij}^{D} \cdot \lVert \vv_i - \vv_j \rVert_2^2
\geq \frac{\sigma}{\vmu(A)} \cdot \sum_{i \in A, j \in B} w_{ij}^{D}
= \frac{\sigma}{\vmu(A)} \cdot \vw\big( e(A, B) \big)
\geq {\sigma \over 2}
\end{equation*}
as required.
\end{proof}

In section~\ref{sec.separated}, we give an algorithm~\roundcut that produces $\Omega\big( \frac{1}{\log n} \big)$-separated $(A_t, B_t)$. This algorithm will subsequently be used in our cut strategy construction for the symmetric case.

\subsubsection{Implementing \mmwu in Nearly-linear Time}
\label{sec.alg-cm.sym-cut-strat.nearly-linear-mmwu}

Computing $\mmwu_{\eta, \vmu}$ in nearly-linear time, along with $\ell_2^2$-distances and projections in logarithmic time are standard techniques. We briefly sketch them here.

Lemma 6 in~\cite{arora2007combinatorial} states one can compute a matrix-vector product with a matrix exponential by truncating its Taylor expansion. Given any $\mA \in \R^{n \times n}$ and $\vv \in \R^n$, one can compute $\vu \in \R^n$ such that
\begin{equation*}
\big\lVert \vu - \exp(\mA)\vv \big\rVert \leq \epsilon
\end{equation*}
in time $O \big( \cT_{\textsf{mv}} \cdot \max \big\{ \norm{\mM}, \log \big(\frac{1}{\epsilon} \big) \big\} \big)$ where $\cT_{\textsf{mv}}$ is matrix-vector product time. To compute $\mmwu_{\eta, \vmu}$ in nearly-linear time, we take advantage of two details.

First, each loss matrix provided to \mmwu is given by the Laplacian of a graph $D_t$. Consequently, the sparsity of the matrix being exponentiated in \mmwu is bounded by $\sparsity(H_t)$. Consequently, matrix-vector product time~satisfies $\cT_{\textsf{mv}} = O\big( \sparsity(H_T) \big)$ where $T$ is the total number of cut-matching game rounds. Second, Lemma~\ref{lem.sym-cut-strat.width} implies that $\big\lVert \mM^{-1/2} \mL(D_t) \mM^{-1/2} \big\rVert_{\spectral} \leq O(1)$ for every $t \in T$, and thus
\begin{equation*}
\Big\lVert \sum_{k=1}^t \mM^{-1/2} \mL(D_t) \mM^{-1/2} \Big\rVert_{\spectral} \leq O(t) \, .
\end{equation*}
If we run at most $\polylog(n)$ iterations of the cut-matching game, we can compute an $O \big( \frac{1}{\poly(n)} \big)$-additive approximation to $\mmwu$ by truncating the Taylor expansion up to a $\polylog(n)$ number of terms.

To compute $\ell_2^2$-distances, and projections in logarithmic time, one can use the Johnson--Lindenstrauss transform \cite{JohnsonLindenstraussOriginal} to project the embedding into $O\big( \frac{\log n}{\delta^2} \big)$ dimensions for $\delta = O(1)$ (See, e.g. \cite{dasgupta2003elementary,achlioptas2001database,blum2020foundations}). Since the error is multiplicative in terms of pairwise $\ell_2^2$-distances, picking $\delta$ to be constant only causes a constant factor loss in the number of cut strategy plays, and the final approximation ratio. We will ignore this constant in the subsequent calculations~\footnote{Picking $\delta > \frac{1}{10}$ loses a factor $2$ in the final bound on $g(n)$ when proving Theorem~\ref{thm.cm.good-cut-strategy}}.

\subsubsection{Completing the Analysis for Symmetric Cut Functions}

We are now ready to describe the cut strategy which we call $\cA^{\leftrightarrow}_{\cut}$. Fix a step size $\eta > 0$ to be chosen subsequently, and consider round $t$ of a $(n,\vmu)$-generalized cut-matching game. Up to this point, the matching player $\cM$ will have produced approximate matching responses $D_1, \ldots, D_t$.

The cut strategy is straightforward: use Johnson-Lindenstrauss and the truncated Taylor expansion, as in Lemma~6 of~\cite{arora2007combinatorial}, to compute a low-dimensional projection of the embedding whose Gram matrix is given by $\mmwu_{\eta, \vmu}\big( \mL(D_1), \ldots, \mL(D_t) \big)$, execute \roundcut on the resulting embedding to produce $\sigma$-separated sets $(A_t, B_t)$, and output $(A_t, B_t)$ as the cut action. Our strategy is summarized in Algorithm~\ref{alg.sym-cut-strategy}.

\begin{figure}[h]
\centering
\noindent\fbox {
\parbox{45em} {
\alglabel{alg.sym-cut-strategy}
\linespread{1.5}\selectfont
\textbf{Algorithm~\thealgorithm} $\cA^{\leftrightarrow}_{\cut}$.

\textbf{Input:} vertex measure $\vmu \in \Z^V_{\geq 0}$, and approximate matching responses $D_1, \ldots, D_t$.

\textbf{Parameters:} step size $\eta > 0$.

\textbf{Do:} The following.
\begin{quote}
\begin{enumerate}[1.]
\item Sample $d = O \big( \log n \big)$ random unit vectors $\vu_1, \ldots, \vu_d \sim \cS^{n-1}$.

\item Denote $\mW_t = \mmwu_{\eta, \vmu}\big( \mL(D_1), \ldots, \mL(D_t) \big)$, and compute $\big\{ \hat{\vv}^{(t)}_i \big\}_{i \in V} \subseteq \R^d$ given by
\begin{equation*}
\big( \hat{\vv}^{(t)}_i \big)_j = \langle \mW_t \vu_j, \ve_i \rangle
\qquad\qquad
\forall \, i \in V \; j \in [d]
\end{equation*}
using the truncated Taylor series via~\cite{arora2007combinatorial}.

\item Execute $\roundcut\big( d, \big\{ \hat{\vv}^{(t)}_i \big\}_{i \in V} \big)$ to produce the cut $(A_t, B_t)$.
\end{enumerate}
\end{quote}

\textbf{Output:} The cut $(A_t, B_t)$.
}
}

\caption{The cut strategy $\cA_{\cut}^{\leftrightarrow}$ for symmetric polymatroidal cut functions $\{ \delta_h \}_{h \in E}$.}
\end{figure}

Let us now demonstrate that $\cA^{\leftrightarrow}_{\cut}$ is a $\left( \Omega (\log n), O \big( {\log^2 n} \big) \right)$-good cut strategy.

\begin{lemma}
\label{lem.sym-cut-strat.good}
Let $G = (V, E, \vw, \vmu)$ be an input submodular hypergraph equipped with symmetric polymatroidal cut functions $\{ \delta_h \}_{h \in E}$ such that $\lvert V \rvert = n$. Then $\cA^{\leftrightarrow}_{\cut}$ is $\Big( \Omega (\log n), O \big( \log^2 n \big) \Big)$-good with probability $O(1)$.
\end{lemma}
\begin{proof}
Let $T = g(n)$, and for any sequence of weighted undirected approximate matching responses $D_1, \ldots, D_T$, suppose $\cA^{\leftrightarrow}_{\cut}$ plays cut actions $(A_1, B_1), \ldots, (A_T, B_T)$ in response. To show that $\cA^{\leftrightarrow}_{\cut}$ is a $\left( \Omega(\log n), O\big( {\log^2 n} \big) \right)$-good cut strategy, we first show that $\Psi^*_H \geq \Omega(\log n)$ after $T = O\big( \log^2 n \big)$ rounds of the cut-matching game.

Fix $H = \sum_{t=1}^T D_t$ to be the state graph of the cut-matching game. For every $t > 0$, the feedback matrix satisfies $\mL(D_t) \succeq \mZero$, hence we can use the regret bound given by~\eqref{eq.regret-bound.local-norm} in Theorem~\ref{thm.regret-bound}. Lemma~\ref{lem.sym-cut-strat.width} ensures
\begin{equation*}
\big\lVert \mM^{-1/2} \mL(D_t) \mM^{-1/2} \big\rVert_{\spectral} \leq 2
\end{equation*}
Consequently, we have for every $\mU \in \Delta_{\vmu}^{n \times n}$ and $\eta > 0$:
\begin{equation*}
\big\langle \tfrac{1}{T} \cdot \mL(H), \mU \big\rangle
= \frac{1}{T} \sum_{t=1}^T \big\langle \mL(D_t), \mU \big\rangle
\geq \frac{1 - 2\eta}{T} \sum_{t=1}^T \big\langle \mL(D_t), \mX_t \big\rangle - \frac{\log n}{\eta T} \, ,
\end{equation*}
Next, for fixed $t > 0$, let $\big\{ \vv_i^{(t)} \big\}_{i \in V}$ be the vector embedding whose Gram matrix is $\mX_t$. Theorem~\ref{thm:separated} implies \roundcut outputs $\Omega\big( \frac{1}{\log n} \big)$-separated sets $(A_t, B_t)$. Lemma~\ref{lem.sym-cut-strat.loss} then implies
\begin{equation*}
\frac{1 - 2\eta}{T} \sum_{t=1}^T \big\langle \mL(D_t), \mX_t \big\rangle - \frac{\log n}{\eta T}
= \frac{1 - 2\eta}{T} \sum_{t=1}^T \sum_{ij\in E(D_t)}\vw_{ij}^{D_t} \cdot \lVert \vv^{(t)}_i - \vv^{(t)}_j \rVert^{2}_2 - \frac{\log n}{\eta T}
\geq \frac{C_1(1 - 2\eta)}{\log n}  - \frac{\log n}{\eta T} \, .
\end{equation*}
for some absolute constant $C_1 > 0$. Let $C_2$ be constant satisfying $2\sqrt{\frac{2}{C_2}} < \sqrt{C_1}$, then with choice of $\eta = {1\over \sqrt{2 C_1 C_2}}$, and $T = C_2 \cdot {\log^2 n }$, we have that:
\begin{align*}
\frac{C_1(1 - 2\eta)}{\log n} - \frac{\log n}{\eta T}
&={ C_1 \over \log n}- \left(2\sqrt{2 C_1 \over C_2}\right){1 \over \log n} \geq \Omega\left({1 \over \log n}\right),
\end{align*}
and thus establishing for every $\mU \in \Delta_{\vmu}^{n \times n}$
\begin{equation}
\label{eq.cm.good-cut-strategy-1}
\big\langle \mL(H), \mU \big\rangle
\geq \Omega \left({T \over \log n}\right)
= \Omega(\log n)
\end{equation}

Now choose $\mU = \mX^*$, the minimizer to the semidefinite relaxation~\eqref{eq.cond-sdp} for $\Psi_H^*$. Certainly $\mX^* \in \Delta_{\vmu}^{n \times n}$, since $\mX^*$ is feasible. If $\opt_H$ is the optimal objective value of~\eqref{eq.cond-sdp}, then by inequality~\eqref{eq.cm.good-cut-strategy-1}
\begin{equation*}
\opt_H
= \big\langle \mL(H), \mX^* \big\rangle
\geq \Omega(\log n) \, .
\end{equation*}
Finally, Lemma~\ref{lem.sym-cut-strat.relaxation} implies the required
\begin{equation*}
\Psi_H^*
\geq \frac{\opt_H}{2}
\geq \Omega(\log n)
\end{equation*}

To ensure that we obtain the above guarantee with good probability, note that the randomizes Johnson-Lindenstrauss transform fails to produce an embedding with the necessary distortion guarantees with probability $O\left({1\over n}\right)$. By the union bound, the probability that this procedure fails at all over any iteration is then $O\left({\polylog(n) \over n}\right) = o_n(1)$. The only other randomized part of the algorithm is the balanced case of Algorithm~\ref{alg.round-cut}. We know from Theorem~\ref{thm:separated} that one obtains a $\Omega\left(1/\log n\right)$-separated set with $1-O(1/n)$ probability at each iteration, which again by the union bound guarantees that this procedure will fail during the algorithm with probability $O\left({\polylog(n) \over n}\right)$. Hence the algorithm will produce the desired output with probability $1-o(1)$.
\end{proof}

\subsection{An \texorpdfstring{$\left( \Omega (\log^2 n), O ( \log^3 n ) \right)$}{(Omega(log\^2 n), O(log\^3 n))}-good Cut Strategy for Asymmetric Cut Functions}

Let us now consider the case where $\{ \delta_h \}_{h \in E}$ are asymmetric. We describe our setup for the online linear optimization setting.

Given an instance of a $(n, \vmu)$-generalized cut-matching game, we again fix the set of actions to be $\cX = \Delta_{\vmu}^{n \times n}$. For round $t$ of the online optimization setting, we now define the feedback matrix to be the directed Laplacian of the approximate matching response $D_t$: $\mF_t = \mL_{+}(D_t)$ so that the losses become
\begin{equation*}
\big\langle \mL_{+}(D_t), \mX \big\rangle
\end{equation*}
Similar to the symmetric case, our choices of action set and loss functions are guided by the trivial spectral relaxation for the directed ratio-cut of a graph.
\begin{lemma}
\label{lem.asym-cut-strat.relaxation}
Given any weighted, directed graph $G = \big( V, E_G, \vw^G, \vmu \big)$, we have $\frac{\opt}{2} \leq \Psi_G^*$ where $\opt$ is the optimal objective value for the following semidefinite program.
\begin{equation*}
\stepcounter{equation}
\tag{\texttt{dir-cond-sdp}}
\label{eq.dir-cond-sdp}
\begin{aligned}
& {\min}
& & \big\langle \mL_{+}(G), \mX \big\rangle \\
& \textup{s.t.}
& & \big\langle \mL(K_{\vmu}), \mX \big\rangle \geq 1 \\
& & & \mX \succeq \mZero
\end{aligned}
\end{equation*}
\end{lemma}

In order to use a regret bound from Theorem~\ref{thm.regret-bound}, we need only address the first two issues presented in the symmetric case:
\begin{enumerate}[(1)]
\item \emph{Ensuring Small Width}: We need to ensure that the width $\big\lVert \mM^{-1/2} \mL_{+}(D_t) \mM^{-1/2} \big\rVert_{\spectral}$ is small given any approximate matching response $D_t$. Similar to the symmetric case, we show this is bounded by a constant.

\item \emph{Ensuring Large Loss}: We now need to ensure that $\big\langle \mL_{+}(D_t), \mX_t \big\rangle$ is large. A key difference present in the cut strategy for the asymmetric case is that ensuring large loss requires an additional step in rounding algorithm to control for terms in the loss that correspond to the embedding vector lengths.
\end{enumerate}
To compute the matrix exponential update $\mmwu_{\eta, \vmu}\big( \mL_{+}(D_1), \ldots, \mL_{+}(D_t) \big)$ in nearly-linear time, we can reuse the sketch provided in Section~\ref{sec.alg-cm.sym-cut-strat.nearly-linear-mmwu}.

\subsubsection{Ensuring Small Width}

Whe width in the asymmetric case is still bounded by an absolute constant given any directed approximate matching response. We prove the following Lemma in Appendix~\ref{sec.appendix.omitted.alg-cm}

\begin{restatable}{lemma}{asymCutStratWidth}
\label{lem.asym-cut-strat.width}
Suppose $D = \big( A \cup B, E_D, \vw^D, \vmu \big)$ is an approximate matching response with respect to cut action $(A, B)$. Then $- \mM \preceq \mL_{+}(D) \preceq 3 \cdot \mM$.
\end{restatable}

\noindent
This lemma implies that the regret guarantee of Theorem~\ref{thm.regret-bound} only for step-size $0 < \eta < 1$ when the losses are $\big\langle \mL_{+}(D), \mX \big\rangle$. We will subsequently choose an $\eta$ satisfying this.

\subsubsection{Ensuring Large Loss}

In the asymmetric case, we wish to produce a cut action $(A_t, B_t)$ such that no matter how the matching player $\cM$ plays $D$, the edges of $D$ will be weighted such that the loss
\begin{equation*}
\big\langle \mL_{+}(D_t), \mX \big\rangle
= \sum_{(i, j) \in E_{D_t}} w^{D_t}_{ij} \cdot \Big( \lVert \vv_{i} - \vv_{j} \rVert_2^2 + \lVert \vv_{i} \rVert_2^2 - \lVert \vv_{j} \rVert_2^2 \Big)
\end{equation*}
is large. The critical distinction from the symmetric case is the presence of the embedding vector lengths.

This is where one might be inclined to require that valid matching responses be Eulerian directed graphs. If the matching response $D$ is Eulerian, then the graph can be split into two parts where the edges are all oriented from $A_t$ to $B_t$, and from $B_t$ to $A_t$. One can then cancel the embedding vector lengths in the loss using the fact that the in- and out-degrees are identical for each vertex. This reduces the problem of constructing a good cut strategy back to that present in the symmetric case: one need only round a cut $(A_t, B_t)$ such that a large amount of flow is always routed between vertex pairs whose embedding vectors are far in $\ell_2^2$-distance.

Absent of this trick, one then needs to control the length of the embedding vectors. A priori this seems difficult. To ensure that an approximation algorithm based on cut-matching games runs in nearly-linear time, some low-dimensional projection needs to be applied to the embedding. One then cannot hope to simultaneously control for both the distortion of the projected embedding, and the length of the projected vectors, suggesting that the matching response \emph{must} be an Eulerian graph. We remark that if one tries to directly adapt the potential analysis of~\cite{khandekarGraphPartitioningUsing2009} or~\cite{orecchiaPartitioningGraphsSingle2008}, then the same terms in the loss appear in lower bounding the potential progress made when the matching player routes a large flow. This might suggest why current works generalizing the cut-matching game to directed settings~\cite{louis2010cut, lau2023fast} require the matching player to play Eulerian graphs.

However, if one is cognizant of the underlying SDP being used to certify lower bounds on the directed conductance of the state graph, then it is natural to consider rounding algorithms which directly target the SDP. Indeed, applying the rounding algorithm present in~\cite{agarwal2005log} (Algorithm 4) augments a $\sigma$-separated cut $(S, T)$ into a cut action $(A, B)$ such that any matching response across $(A, B)$ must have large loss. We describe this rounding procedure in Algorithm~\ref{alg.directed-round-cut} \dirroundcut.

\begin{figure}[h]
\centering
\noindent\fbox {
\parbox{45em} {
\alglabel{alg.directed-round-cut}
\linespread{1.5}\selectfont
\textbf{Algorithm~\thealgorithm} \dirroundcut.

\textbf{Input:} A cut $(S, T)$ of $V$, and  vector embedding $\{ \vv_i \}_{i \in V} \subseteq \R^d$.
\begin{quote}
\begin{enumerate}[1.]
\item Find $r > 0$ such that $\vmu \big( A \cap \cB_r(0) \big) \geq \frac{\vmu(S)}{2}$ and $\vmu \big( A \cap \overline{\cB_r(0)} \big) \geq \frac{\vmu(S)}{2}$.

\item Define the following sets.
\begin{align*}
S^+ &\defeq \big\{ i \in S : \lVert \vv_i \rVert_2^2 \leq r^2 \big\}
\qquad\qquad
S^- \defeq \big\{ i \in S : \lVert \vv_i \rVert_2^2 \geq r^2 \big\} \\
T^+ &\defeq \big\{ i \in T : \lVert \vv_i \rVert_2^2 \leq r^2 \big\}
\qquad\qquad
T^- \defeq \big\{ i \in T : \lVert \vv_i \rVert_2^2 \geq r^2 \big\}
\end{align*}

\item If $\vmu(T^+) \leq \vmu(T^-)$, then let $A = S^+$ and $B = T^-$. Else, let $A = T^+$ and $B = S^-$.
\end{enumerate}
\end{quote}

\textbf{Output:} the cut $(A, B)$
}
}
\label{fig:directed-round-cut}
\end{figure}

The intuition for \dirroundcut is this: one does not need to route the flow bidirectionally across $(S, T)$, and subsequently produce an Eulerian demand graph. The SDP solution, and specifically the position of the median embedding vector, should dictate which direction flow needs to be routed in a given round the cut-matching game. We now prove our lower bound on the loss

\begin{lemma}
\label{lem.asym-cut-strat.loss}
Let $\vmu \in \Z_{\geq 0}^V$, $\mX \in \Delta_{\vmu}^{n \times n}$, and $\{ \vv_i \}_{i \in V} \subseteq \R^d$ be a vector embedding whose Gram matrix is $\mX$. For $\sigma > 0$, suppose \dirroundcut is given $(S, T)$ that are $\sigma$-separated with respect to $\vmu$ and $\{ \vv_i \}_{i \in V}$ as input, and outputs $(A, B)$. Then, any approximate directed matching response $D = \big( A \cup B, E_D, \vw^D \big)$ with respect to $A, B$ satisfies the following:
\begin{equation*}
\sum_{(i, j) \in E_D} w_{ij}^D \cdot \Big( \lVert \vv_i - \vv_j \rVert_2^2 + \lVert \vv_i \rVert_2^2 - \lVert \vv_j \rVert_2^2 \Big)
\geq \frac{\sigma}{8}
\end{equation*}
\end{lemma}
\begin{proof}
We assume, without loss of generality, that \dirroundcut outputs $(A, B)$ such that $A = S^+$ and $B = T^-$ as the analysis for the case where $A = T^+$ and $B = S^-$ is identical.

In step (1) of \dirroundcut, $r > 0$ is chosen to be the length of the median element in $S$ with respect to measure $\vmu$
, we have $\vmu(A) \geq \frac{\vmu(S)}{2}$. Furthermore, if $A = S^+$ and $B = T^-$, then $\vmu(T^-) \geq \vmu(T^+)$ and hence $\vmu(T^-) \geq \frac{\vmu(T)}{2}$. Consequently,
\begin{equation*}
\vmu(A) \vmu(B) \geq \frac{1}{4} \cdot \vmu(S) \vmu(T) \, .
\end{equation*}
Using the fact that $(S, T)$ are $\sigma$-separated, we derive for all $i \in A, j \in B$:
\begin{equation*}
\frac{\vmu(A) \vmu(B)}{\vmu(V)} \cdot \lVert \vv_i - \vv_j \rVert^2_2
\geq \frac{1}{4} \cdot \frac{\vmu(S) \vmu(T)}{\vmu(V)} \cdot \lVert \vv_i - \vv_j \rVert^2_2
\geq \frac{\sigma}{4} \cdot \sum_{i,j \in \binom{V}{2}} \frac{\mu_i \mu_j}{\vmu(V)} \cdot \lVert \vv_i - \vv_j \rVert_2^2
\end{equation*}
Because $\mX \in \Delta_{\vmu}^{n \times n}$, we have $\big\langle \mL(K_{\vmu}), \mX \big\rangle \geq 1$ and
\begin{equation*}
\frac{\sigma}{4} \cdot \sum_{i,j \in \binom{V}{2}} \frac{\mu_i \mu_j}{\vmu(V)} \cdot \lVert \vv_i - \vv_j \rVert_2^2
= \frac{\sigma}{4} \cdot \big\langle \mL(K_{\vmu}), \mX \big\rangle
\geq \frac{\sigma}{4}
\end{equation*}
Combining the above two calculations, we have the inequality
\begin{equation*}
\frac{\vmu(A) \vmu(B)}{\vmu(V)} \cdot \lVert \vv_i - \vv_j \rVert^2_2
\geq \frac{\sigma}{4}
\qquad\Longleftrightarrow\qquad
\lVert \vv_i - \vv_j \rVert^2_2
\geq \frac{\sigma}{4 \vmu(A)} \cdot \frac{\vmu(V)}{\vmu(B)}
\geq \frac{\sigma}{4 \vmu(A)}
\end{equation*}
where we have used the fact that $\vmu(B) \leq \vmu(V)$. Finally, because any approximate directed matching response is bipartite across $(A, B)$, and the total weight is large, we conclude
\begin{equation*}
\sum_{(i, j) \in E_D} w_{ij}^D \cdot \lVert \vv_i - \vv_j \rVert_2^2
\geq \frac{\sigma}{4 \vmu(A)} \sum_{(i, j) \in E_D} w_{ij}^D
= \frac{\sigma}{4 \vmu(A)} \cdot \vw \big( e(A, B) \big)
\geq \frac{\sigma}{8}
\end{equation*}
as required.
\end{proof}

We also bound the running time of \dirroundcut.

\begin{lemma}
\label{lem.asym-cut-strat.dirroundcut-runtime}
\dirroundcut runs in time $O(nd + n \log n)$ for $n = \lvert V \rvert$.
\end{lemma}
\begin{proof}
Computing the norm of a vector takes time $O(d)$, so all norm computations take time $O(nd)$ overall. The vectors then need to be sorted by norm value, which takes time $O(n \log n)$. After that, completing the algorithm takes $O(n)$ time. This yield the desired runtime.
\end{proof}

\subsubsection{Completing the Analysis for Asymmetric Cut Functions}

We are now ready to give our good cut strategy

\begin{figure}[h]
\centering
\noindent\fbox {
\parbox{45em} {
\alglabel{alg.dir-cut-strategy}
\linespread{1.5}\selectfont
\textbf{Algorithm~\thealgorithm} $\cA^{\rightarrow}_{\cut}$.

\textbf{Input:} vertex measure $\vmu \in \Z^V_{\geq 0}$, and directed approximate matching responses $D_1, \ldots, D_t$.

\textbf{Parameters:} step size $\eta > 0$.

\textbf{Do:} The following.
\begin{quote}
\begin{enumerate}[1.]
\item Sample $d = O \big( \log n \big)$ random unit vectors $\vu_1, \ldots, \vu_d \sim \cS^{n-1}$.

\item Denote $\mW_t = \mmwu_{\eta, \vmu}\big( \mL_{+}(D_1), \ldots, \mL_{+}(D_t) \big)$, and compute $\big\{ \hat{\vv}^{(t)}_i \big\}_{i \in V} \subseteq \R^d$ given by
\begin{equation*}
\big( \hat{\vv}^{(t)}_i \big)_j = \langle \mW_t \vu_j, \ve_i \rangle
\qquad\qquad
\forall \, i \in V \; j \in [d]
\end{equation*}
using the truncated Taylor series via~\cite{arora2007combinatorial}.

\item Execute $\roundcut\big( d, \big\{ \hat{\vv}^{(t)}_i \big\}_{i \in V} \big)$ to produce the cut $(S, T)$

\item Execute $\dirroundcut\big( (S, T), \big\{ \hat{\vv}^{(t)}_i \big\}_{i \in V} \big)$ to produce the cut $(A_t, B_t)$.
\end{enumerate}
\end{quote}

\textbf{Output:} The cut $(A_t, B_t)$.
}
}

\caption{The cut strategy $\cA_{\cut}^{\rightarrow}$ for asymmetric polymatroidal cut functions $\{ \delta_h \}_{h \in E}$.}
\end{figure}

\begin{lemma}
\label{lem.asym-cut-strat.good}
Let $G = (V, E, \vw, \vmu)$ be an input submodular hypergraph equipped with asymmetric polymatroidal cut functions $\{ \delta_h \}_{h \in E}$ such that $\lvert V \rvert = n$. Then $\cA^{\rightarrow}_{\cut}$ is $\Big( \Omega \big( \log^2 n \big), O \big( \log^3 n \big) \Big)$-good with probability $O(1)$.
\end{lemma}
\begin{proof}
We demonstrate that $\Psi^*_H \geq \Omega(\log^2 n)$ after $T = O \big( {\log^3 n} \big)$ rounds of the cut-matching game. Since the feedback matrices $\mL_{+}(D_t)$ are not positive semidefinite, we use regret bound~\eqref{eq.regret-bound.standard} from Theorem~\ref{thm.regret-bound}. Lemma~\ref{lem.asym-cut-strat.width} guarantees $\big\lVert \mM^{-1/2} \mL_{+}(D_t) \mM^{-1/2} \big\rVert_{\spectral} \leq 3$ for every $t > 0$, thus
\begin{align*}
\big\langle \tfrac{1}{T} \cdot \mL_{+}(H), \mU \big\rangle
&\geq \frac{1}{T} \sum_{t=1}^T \langle \mL_{+}(D_t), \mX_t \rangle - \frac{\eta}{T} \sum_{t=1}^T \, \big\lVert \mM^{-1/2} \mL_{+}(D_t) \mM^{-1/2} \big\rVert_{\spectral}^2 - \frac{\log n}{\eta T} \\
&\geq \frac{1}{T} \sum_{t=1}^T \langle \mL_{+}(D_t), \mX_t \rangle - 9 \eta - \frac{\log n}{\eta T}
\end{align*}
holding for every $0 < \eta < 1$.

Next, consider the $t$-th round, and let $\big\{ \vv_i^{(t)} \}_{i \in V}$ be the vector embedding whose Gram matrix is $\mX_t$. Theorem~\ref{thm:separated}
implies \roundcut outputs $\Omega\big( \frac{1}{\log n} \big)$-separated sets $(S, T)$. Consequently, Lemma~\ref{lem.asym-cut-strat.loss} implies that \dirroundcut produces a cut action $(A_t, B_t)$ such that $D_t$ satisfies
\begin{equation*}
\langle \mL_{+}(D_t), \mX_t \rangle
= \sum_{(i, j) \in E_{D_t}} w_{ij}^{D_t} \cdot \Big( \big\lVert \vv_i^{(t)} - \vv_j^{(t)} \big\rVert_2^2 + \big\lVert \vv_i^{(t)} \big\rVert_2^2 - \big\lVert \vv_j^{(t)} \big\rVert_2^2 \Big)
\geq \frac{C}{8 \cdot \log n}
\end{equation*}
for some absolute constant $C_1 > 0$. Consequently, the regret bound becomes
\begin{equation*}
\frac{1}{T} \sum_{t=1}^T \langle \mL_{+}(D_t), \mX_t \rangle - 9 \eta - \frac{\log n}{\eta T}
\geq \frac{C_1}{8 \cdot \log n}  - 9 \eta - \frac{\log n}{\eta T} \, .
\end{equation*}
Let $C_2$ be a constant satisfying: $\frac{C_1}{8} > 6/\sqrt{C_2}$. Choosing $\eta = \frac{1}{3 \sqrt{c_2} \log n}$ and $T = C_2 \cdot \log^3 n$, we have that:
\begin{equation*}
\frac{C_1}{8 \cdot \log n} - 9 \eta - \frac{\log n}{\eta T}
= {1
\over \log n} \left({C_1 \over 8}- {6 \over \sqrt{C_2}}\right) \geq \Omega\left({1 \over \log n}\right)
\end{equation*}
This implies that every $\mU \in \Delta_{\vmu}^{n \times n}$ satisfies:
\begin{equation*}
\big\langle \tfrac{1}{T} \cdot \mL_{+}(H), \mU \big\rangle
\geq \Omega\left({1\over \log n}\right)
\end{equation*}

Now, choosing $\mU = \mX^*$ to be the minimizer to~\eqref{eq.dir-cond-sdp}, and $\opt_H$ the optimal objective value for~\eqref{eq.dir-cond-sdp}, we have
\begin{equation*}
\opt_H
= \big\langle \mL_{+}(H), \mX^* \big\rangle
\geq \Omega\left({T \over \log n}\right) = \Omega\left({\log^2n}\right)
\end{equation*}
Finally, Lemma~\ref{lem.asym-cut-strat.relaxation} implies the required
\begin{equation*}
\Psi_H^* \geq \frac{\opt_H}{2} \geq \Omega\left({\log^2n}\right)
\end{equation*}

These guarantees can be shown to hold with super-constant probability by following the same argument as in Lemma~\ref{lem.sym-cut-strat.good}.
\end{proof}

\subsection{Completing the Cut Strategy: Proof of Theorem~\ref{thm.cm.good-cut-strategy}}

For the sake of completion, let us now prove Theorem~\ref{thm.cm.good-cut-strategy}.

\cmGoodCutStrategy*
\begin{proof}
Items (1) and (2) of statement are given by Lemmas~\ref{lem.sym-cut-strat.good}, and~\ref{lem.asym-cut-strat.good} respectively. To bound the running time, notice that by the sketch in Section~\ref{sec.alg-cm.sym-cut-strat.nearly-linear-mmwu}, the matrix multiplicative weight update can be computed in time bounded by $\tilde{O}\big( \sparsity(H) \big)$. Meanwhile, the \roundcut and \dirroundcut steps run in time $O(nd + n \log n)$ as given by Theorem~\ref{thm:separated}, and lemma~\ref{lem.asym-cut-strat.dirroundcut-runtime}. Since $d \leq O(\log n)$ for both $\cA_{\cut}^{\leftrightarrow}$ and $\cA_{\cut}^{\rightarrow}$, the total runtime is dominated by $\tilde{O}\big( \sparsity(H) \big)$ as required.
\end{proof}

\section{Constructing Separated Sets: Proof of Theorem~\ref{thm:separated}}
\label{sec.separated}

This section is devoted to proving Theorem~\ref{thm:separated} on the construction of separated sets. We will show that the algorithm \roundcut in Figure~\ref{fig:roundcut} satisfies the requirements of the theorem.
We denote by $V = [n]$ the index set for the vector embedding $\{ \vv_i \}_{i \in V} \subseteq \R^d$ in the theorem statement and define the following notions of expected vector and total variance of the embedding:
\begin{align*}
\vv_{\avg}
&\defeq \frac{1}{\mu(V)} \sum_{i \in V} \mu_i \vv_i \\
\Var_{\vmu} (\vv_i)
&\defeq \frac{1}{\mu(V)}\sum_{i \in V} \mu_i \cdot \lVert \vv_i - \vv_{\avg} \rVert_2^2
= \frac{1}{(\mu(V))^2} \sum_{\{i,j\} \in \binom{V}{2}} \mu_i \mu_j \|\vv_i - \vv_j\|^2 .
\end{align*}

To simplify our proof, we will assume, without loss of generality, that our embedding is both \emph{centered},
i.e., $\vv_{\avg} = 0,$
and \emph{normalized}, i.e., $\Var_{\vmu} (\vv_i) = 1.$
Under these assumptions, we will make use of the following lemma, which relates the contribution to the total variance of a set $\bar{S}$ to the that of its complement $S$.
\begin{lemma}[Fact 5.19 from~\cite{orecchia2012balanced_separator}]
\label{lem.variance}
Let $\{ \vv_i \}_{i \in V}$ be a centered and normalized vector embedding. We have:
\begin{equation*}
\sum_{i \in \overline{S}} \mu_i \cdot \lVert \vv_i \rVert_2^2 \geq \mu(S) \cdot  \bigg( 1 - \frac{1}{(\mu(V))^2} \cdot \sum_{i,j \in \binom{S}{2}} \mu_i \mu_j \cdot \lVert \vv_i - \vv_j \rVert_2^2 \bigg).
\end{equation*}
\end{lemma}

Our algorithm handles two cases, depending on whether there exist a set of small measure that contributes most of the variance of the embedding or not.
To distinguish these two cases, we consider long vectors, i.e. vectors with large contributions to the variance of the embedding, and in particular we define the set $R_t \subseteq V$ of vertices with long vector embeddings:
\begin{equation*}
R_t \defeq \Big\{ i \in V : \lVert \vv_i \rVert_2^2 \leq t \Big\}
\end{equation*}

We then say that an embedding is $(t, b)$-balanced if the variance due to vertices in $R_t$ captures a $b$-fraction of the total embedding variance, as per the following definition.

\begin{definition}[$(t, b)$-Balanced Embedding]\label{def.balanced-embedding}
Given $t > 0$, and $b > 0$, a centered normalized vector embedding $\{ \vv_{i} \in \R^d \}_{i \in V}$ is $(t, b)$-\emph{balanced} if the set $R_t$ satisfies:
\begin{equation*}
\frac{1}{(\mu(V))^2} \cdot \sum_{i,j \in \binom{R_t}{2}} \vmu(i) \vmu(j) \cdot \lVert \vv_i - \vv_j \rVert_2^2
\geq b \,.
\end{equation*}
\end{definition}

We distinguish the two cases based on whether the vector embedding is $(3,\nicefrac{1}{10})$-balanced.
We remark that, given $t$ and $b$, one can check if the embedding is $(t, b)$-balanced in $O(nd)$ time by computing distances to the average vector of $R_t$.
We start by describing the unbalanced case, which is where our main technical contribution lies.

\begin{figure}[h]
\centering
\noindent\fbox {
\parbox{45em} {
\alglabel{alg.round-cut}
\linespread{1.5}\selectfont
\textbf{Algorithm~\thealgorithm} \roundcut.

\textbf{Input:} A centered normalized vector embedding $\{ \vv_i \in \R^d \}_{i \in V}.$

\begin{quote}
\begin{enumerate}[1.]
\item If $\{ \vv_i \}_{i \in V} \subseteq \R^d$ is $\big( 3, \frac{1}{10} \big)$-balanced:
\begin{itemize}
\item Sample $\vg$ uniformly at random from $\cS^{d-1}$ and compute $r_i = \sqrt{d} \cdot \langle \vg, \vv_i \rangle$ for each $i \in V$.

\item Sort $r_i$. Assume without loss of generality that $r_1 \geq \ldots \geq r_n$ and define sweep cuts
\begin{equation*}
S_t^{up} = \{ i \in V : r_i \geq r_t \}\text{ and }S_t^{down} = \{ i \in V : r_i \leq r_t \}
\end{equation*}

\item Let $t_{start} = \min \{t : 3\cdot \vmu(S_t^{up})\geq \vmu(V)\}$ and $t_{end} = \max \{t : 3\cdot \vmu(S_t^{down})\geq \vmu(V) \}$

\item Output $\big( S_{t_{start}}^{up}, S_{t_{end}}^{down} \big)$.
\end{itemize}

\item If $\{ \vv_i \}_{i \in V} \subseteq \R^d$ is not $\big( 3, \frac{1}{10} \big)$-balanced:
\begin{itemize}
\item Sort the vectors by their length. Relabel vertices so that $\|\vv_1\| \geq \|\vv_2\| \geq \ldots \geq \|\vv_n\|$.
\item Let $T = \{i \in V : \|\vv_i\|^2 \leq \frac{3}{2}\}.$
\item Define sweep cuts $S_t \defeq \{ i \in V : \|\vv_i\| \geq \|\vv_t\| \}.$  For each sweep cut $S_t$ with $\|\vv_t \| \geq 3$:
\begin{itemize}
    \item if the following condition holds(see Equation~\ref{eq.condition} and Equation~\ref{eq.dist-lower}),
    $$
    \frac{\vmu(S_t)}{\vmu(V)} \cdot \frac{\|\vv_t\|^2}{8} \geq \frac{3}{100 \log(\vmu(V))},
    $$
    output $\frac{1}{100\log(\vmu(V))}$-separated sets $(S_t, T).$
\end{itemize}
\end{itemize}
\end{enumerate}
\end{quote}

}
}
\caption{The \roundcut\xspace algorithm.}
\label{fig:roundcut}
\end{figure}

\paragraph{The unbalanced case}
In the unbalanced case, the \roundcut algorithm sorts the embedding vectors according to their length and outputs separated sets $S$ and $T$
from the set of threshold cuts in this ordering. In particular, the set $T$ is always chosen to equal $R_{3/2}$, while the set $S$ is chosen by sweeping through all cuts in the ordering to find one that satisfies the required separation guarantee. We show that such a cut always exists.
The following lemma, which is implicit in previous works~\cite{ARV2009, agarwal2005log}, follows from Markov's Inequality and Lemma~\ref{lem.variance}. We prove it for completeness in Appendix~\ref{sec.appendix.omitted}.
\begin{restatable}{lemma}{unbalancedprevious}
\label{lem.unbalanced}
Let $\{ \vv_i \in \R^d \}_{i \in V}$ be a centered normalized vector embedding that is not $(3,\nicefrac{1}{10})$-balanced. Then:
$$
\mu(R_{3/2}) \geq \frac{1}{3} \cdot \mu(V) \, \textrm{ and } \, \sum_{i \in \overline{R_3}} \mu_i \cdot \|\vv_i\|^2 \geq \frac{3}{5} \cdot \mu(V).
$$
\end{restatable}
\noindent
Setting $T = R_{3/2}$, it remains to that there exists some $t \geq 3$ such that all pairs $i \in \overline{R_t}$, $j \in T$ satisfy
\begin{equation}
\label{eq.condition}
\lVert \vv_i - \vv_j \rVert_2^2 \geq \sigma \cdot \frac{3\mu(V)}{\mu(\overline{R_t})},
\end{equation}
where $\sigma = \nicefrac{1}{100 \log(\mu(V))} = \Omega(\nicefrac{1}{\log n}).$
For sake of contradiction, assume that no such $t$ exists. Then, for every $t \geq 3,$ there exists a pair $i \in \overline{R_t}, j \in T$ such that the squared distance between $\vv_i$ and $\vv_j$ is less than $\sigma \cdot \nicefrac{3\mu(V)}{\mu(\overline{R_t})}.$
At the same time, we  can lower bound such squared distance as follows:
\begin{equation}
\label{eq.dist-lower}
\|\vv_i - \vv_j\|^2 \geq (\|\vv_i\| - \|\vv_j\|)^2 \geq (\sqrt{t} - \sqrt{\nicefrac{3}{2}})^2 \geq \frac{t}{8}\,.
\end{equation}
This yields, for all $t \geq 3$:
$$
\mu(\overline{R_t}) \leq \frac{24\sigma}{t} \cdot \mu(V).
$$
Integrating over $t \in [3,M]$ where $M \defeq \max_{i \in V, \mu_i > 0} \|\vv_i\|^2$, we have:
$$
\sum_{i \in \overline{R_3} }\mu_i \|\vv_i\|^2 = \int_3^M \mu(\overline{R_t}) \, dt\, +\, 3 \cdot \mu(\overline{R_3}) \leq 24\sigma \cdot \mu(V) \cdot \int_3^M \frac{dt}{t} \, + 8 \sigma \cdot \mu(V) \leq (24 \log(M) + 8) \cdot \sigma \mu(V).
$$
At the same time, Lemma~\ref{lem.unbalanced} lower bounds the right hand side:
$$
\frac{3}{5} \cdot \mu(V) \leq \sum_{i \in \overline{R_3} }\mu_i \|\vv_i\|^2 \leq 32 \cdot \log(\mu(V)) \cdot \sigma\cdot \mu(V)\,,
$$
which implies
$
\sigma \geq \nicefrac{1}{54 \log(\mu(V))},
$
yielding the required contradiction.

On the right hand side, by the normalization assumption and the fact that $\vmu \in \Z^n_{\geq 0}$, we have $
M \leq \sum_{i \in V} \mu_i \|\vv_i\|^2 \leq \mu(V) = \poly(n).
$
 This completes the proof that a set $\overline{R_t}$ with the required separation guarantee exists. To find such cut, the \roundcut algorithm avoids computing the minimum distance between $\overline{R_t}$ and $T$ for each threshold cut by noticing that the guarantee holds if we replace such minimum distance with the quantity $\nicefrac{t}{8}$, as we did in Equation~\ref{eq.dist-lower}.

\paragraph{The balanced case}
In the balanced case, the set $R_3$ of short vectors captures a constant fraction of the vector embedding's variance and a large fraction of the total volume $\vmu(V)$ via Markov's inequality. One can subsequently expect that random projection coupled with a search for a balanced sweep cut will work as a bulk of the vectors will be well separated from one another. The balanced case analysis is thus very similar to the analysis of SDP-based approximation algorithms for balanced separator~\cite{orecchia2012balanced_separator}. In fact, it is simpler as one does not need to bound the conductance of such a cut and so we leave the proof of the following lemma to Appendix~\ref{sec.appendix.omitted.separated}.

\begin{restatable}{lemma}{robustSeparatedBalanced}
\label{lem.robust-separated.balanced}
Given $V = [n]$, let $\{ \vv_i \in \R_d \}_{i \in V} \subseteq \R^d$ be a centered normalized vector embedding. If $\{ \vv_i \}_{i \in V}$ is a $\big( 3, \frac{1}{10} \big)$-balanced embedding, then algorithm \roundcut~outputs a pair $(S,T)$ of $ \Omega\big( \frac{1}{\log n} \big)$-separated sets with high probability.
m\end{restatable}

\paragraph{Running Time} To bound the running time, consider the following. Checking if the embedding is $(3,\nicefrac{1}{10})$-balanced takes $O(nd)$ time. This can seen by noting that the inequality in Definition \ref{def.balanced-embedding} can be rewritten as:
\[
    {1\over \mu(R_t)}\sum_{i\in R_t} \mu(i) \Big\|\vv_i - {1\over \mu(R_t)}\sum_{j\in R_t} \mu(j)\vv_j\Big\|_2^2 \geq b.
\]

Computing $\ell_2^2$-distances and vector projections takes $O(d)$ time per vector, of which there are $n$. Sorting takes $O(n\log n)$ time. There are at most $O(n)$ sweep cuts to check in both the balanced and unbalanced case, and finding the appropriate cut in both cases takes $O(1)$ time per cut. Checking the condition:
\[
    \frac{\mu(S_t)}{\mu(V)} \cdot \frac{\|\vv_t\|^2}{8} \geq \frac{3}{100 \log(\mu(V))},
\]
requires time $O(d)$ per cut. Hence the overall running time is at most $O(nd + n \log n)$.

This completes the proof of Theorem \ref{thm:separated}.

\section{Acknowledgements}

AC is supported by NSF DGE 2140001. Part of this work was also done while AC was a visiting student at Bocconi University. LO is supported by NSF CAREER 1943510. The authors would like to thank anonymous reviewers for suggesting related work on polymatroidal networks, Thatchaphol Saranurak for suggesting related work on the use of cut-matching games beyond partitioning graphs, and Jeffrey Negrea for assisting on topics related to online learning. AC would like to thank Luca Trevisan for discussion regarding HDX construction, flow embedding techniques for certifying expansion of random regular graphs, and support while visiting Bocconi.

\newpage

\bibliographystyle{plain}
\bibliography{main}

\begin{thebibliography}{10}

\bibitem{achlioptas2001database}
Dimitris Achlioptas.
\newblock Database-friendly random projections.
\newblock In {\em Proceedings of the twentieth ACM SIGMOD-SIGACT-SIGART
  symposium on Principles of database systems}, pages 274--281, 2001.

\bibitem{agarwal2005log}
Amit Agarwal, Moses Charikar, Konstantin Makarychev, and Yury Makarychev.
\newblock $\mathcal{O}(\sqrt{\log n})$-approximation algorithms for {Min UnCut,
  Min 2CNF Deletion}, and directed cut problems.
\newblock In {\em Proceedings of the thirty-seventh annual ACM symposium on
  Theory of computing}, pages 573--581, 2005.

\bibitem{agarwalHigherOrderLearning2006}
Sameer Agarwal, Kristin Branson, and Serge Belongie.
\newblock Higher order learning with graphs.
\newblock In {\em Proceedings of the 23rd international conference on machine
  learning}, pages 17--24. {Association for Computing Machinery}, 2006.

\bibitem{ahujaravindrakandmagnantithomaslandorlinjamesbNetworkFlowsTheory1993}
Ravindra~K Ahujia, Thomas~L Magnanti, and James~B Orlin.
\newblock {\em Network Flows: Theory, Algorithms and Applications}.
\newblock Prentice-Hall, 1993.

\bibitem{allen2015spectral}
Zeyuan Allen-Zhu, Zhenyu Liao, and Lorenzo Orecchia.
\newblock Spectral sparsification and regret minimization beyond matrix
  multiplicative updates.
\newblock In {\em Proceedings of the forty-seventh annual ACM symposium on
  Theory of computing}, pages 237--245, 2015.

\bibitem{ameranisEfficientFlowbasedApproximation2023}
Konstantinos Ameranis, Antares Chen, Lorenzo Orecchia, and Erasmo Tani.
\newblock Efficient flow-based approximation algorithms for submodular
  hypergraph partitioning via a generalized cut-matching game.
\newblock {\em arXiv preprint arXiv:2301.08920}, 2023.

\bibitem{Andersen-Lang}
Reid Andersen and Kevin~J. Lang.
\newblock An algorithm for improving graph partitions.
\newblock In {\em Proceedings of the Nineteenth Annual ACM-SIAM Symposium on
  Discrete Algorithms}, SODA '08, page 651–660, USA, 2008. Society for
  Industrial and Applied Mathematics.

\bibitem{aroraSqrtLognApproximation2010}
Sanjeev Arora, Elad Hazan, and Satyen Kale.
\newblock ${{O}}(\sqrt{\log n})$-{Approximation} to {{SPARSEST CUT}} in
  $\tilde{O}(n^2)$ {{Time}}.
\newblock {\em SIAM Journal on Computing}, 39(5):1748--1771, January 2010.

\bibitem{arora2007combinatorial}
Sanjeev Arora and Satyen Kale.
\newblock A combinatorial, primal-dual approach to semidefinite programs.
\newblock In {\em Proceedings of the thirty-ninth annual ACM symposium on
  Theory of computing}, pages 227--236, 2007.

\bibitem{ARV2009}
Sanjeev Arora, Satish Rao, and Umesh Vazirani.
\newblock Expander flows, geometric embeddings and graph partitioning.
\newblock {\em Journal of the ACM (JACM)}, 56(2):1--37, 2009.

\bibitem{bach2013learning}
Francis Bach et~al.
\newblock Learning with submodular functions: A convex optimization
  perspective.
\newblock {\em Foundations and Trends in Machine Learning}, 6(2-3):145--373,
  2013.

\bibitem{bensonHigherorderOrganizationComplex2016}
Austin~R Benson, David~F Gleich, and Jure Leskovec.
\newblock Higher-order organization of complex networks.
\newblock {\em Science}, 353(6295):163--166, 2016.

\bibitem{bernstein2022deterministic}
Aaron Bernstein, Maximilian~Probst Gutenberg, and Thatchaphol Saranurak.
\newblock Deterministic decremental sssp and approximate min-cost flow in
  almost-linear time.
\newblock In {\em 2021 IEEE 62nd Annual Symposium on Foundations of Computer
  Science (FOCS)}, pages 1000--1008. IEEE, 2022.

\bibitem{blum2020foundations}
Avrim Blum, John Hopcroft, and Ravindran Kannan.
\newblock {\em Foundations of data science}.
\newblock Cambridge University Press, 2020.

\bibitem{catalyurekHypergraphpartitioningbasedDecompositionParallel1999}
Umit~V Catalyurek and Cevdet Aykanat.
\newblock Hypergraph-partitioning-based decomposition for parallel
  sparse-matrix vector multiplication.
\newblock {\em IEEE Transactions on parallel and distributed systems},
  10(7):673--693, 1999.

\bibitem{chan2018spectral}
T-H~Hubert Chan, Anand Louis, Zhihao~Gavin Tang, and Chenzi Zhang.
\newblock Spectral properties of hypergraph laplacian and approximation
  algorithms.
\newblock {\em Journal of the ACM (JACM)}, 65(3):1--48, 2018.

\bibitem{chan2012linear}
Yuk~Hei Chan and Lap~Chi Lau.
\newblock On linear and semidefinite programming relaxations for hypergraph
  matching.
\newblock {\em Mathematical programming}, 135(1-2):123--148, 2012.

\bibitem{charikar2006directed}
Moses Charikar, Konstantin Makarychev, and Yury Makarychev.
\newblock Directed metrics and directed graph partitioning problems.
\newblock In {\em SODA}, volume~6, pages 51--60, 2006.

\bibitem{chekuri2012multicommodity}
Chandra Chekuri, Sreeram Kannan, Adnan Raja, and Pramod Viswanath.
\newblock Multicommodity flows and cuts in polymatroidal networks.
\newblock In {\em Proceedings of the 3rd Innovations in Theoretical Computer
  Science Conference}, pages 399--408, 2012.

\bibitem{chenMaximumFlowMinimumCost2022}
Li~Chen, Rasmus Kyng, Yang~P Liu, Richard Peng, Maximilian~Probst Gutenberg,
  and Sushant Sachdeva.
\newblock Maximum flow and minimum-cost flow in almost-linear time.
\newblock In {\em 2022 IEEE 63rd Annual Symposium on Foundations of Computer
  Science (FOCS)}, pages 612--623. IEEE, 2022.

\bibitem{chuzhoyDistancedMatchingGame2023}
Julia Chuzhoy.
\newblock A {{Distanced Matching Game}}, {{Decremental APSP}} in {{Expanders}},
  and {{Faster Deterministic Algorithms}} for {{Graph Cut Problems}}.
\newblock In {\em Proceedings of the 2023 {{Annual ACM-SIAM Symposium}} on
  {{Discrete Algorithms}} ({{SODA}})}, Proceedings, pages 2122--2213. {Society
  for Industrial and Applied Mathematics}, January 2023.

\bibitem{chuzhoy2020deterministic}
Julia Chuzhoy, Yu~Gao, Jason Li, Danupon Nanongkai, Richard Peng, and
  Thatchaphol Saranurak.
\newblock A deterministic algorithm for balanced cut with applications to
  dynamic connectivity, flows, and beyond.
\newblock In {\em 2020 IEEE 61st Annual Symposium on Foundations of Computer
  Science (FOCS)}, pages 1158--1167. IEEE, 2020.

\bibitem{chuzhoyNewAlgorithmDecremental2019}
Julia Chuzhoy and Sanjeev Khanna.
\newblock A new algorithm for decremental single-source shortest paths with
  applications to vertex-capacitated flow and cut problems.
\newblock In {\em Proceedings of the 51st Annual ACM SIGACT Symposium on Theory
  of Computing}, pages 389--400, 2019.

\bibitem{dasgupta2003elementary}
Sanjoy Dasgupta and Anupam Gupta.
\newblock An elementary proof of a theorem of johnson and lindenstrauss.
\newblock {\em Random Structures \& Algorithms}, 22(1):60--65, 2003.

\bibitem{devineParallelHypergraphPartitioning2006}
Karen~D Devine, Erik~G Boman, Robert~T Heaphy, Rob~H Bisseling, and Umit~V
  Catalyurek.
\newblock Parallel hypergraph partitioning for scientific computing.
\newblock In {\em Proceedings 20th IEEE International Parallel \& Distributed
  Processing Symposium}, pages 10--pp. IEEE, 2006.

\bibitem{eneRandomCoordinateDescent2015}
Alina Ene and Huy Nguyen.
\newblock Random coordinate descent methods for minimizing decomposable
  submodular functions.
\newblock In {\em International Conference on Machine Learning}, pages
  787--795. PMLR, 2015.

\bibitem{eneDecomposableSubmodularFunction2017}
Alina Ene, Huy Nguyen, and L{\'a}szl{\'o}~A V{\'e}gh.
\newblock Decomposable submodular function minimization: discrete and
  continuous.
\newblock {\em Advances in neural information processing systems}, 30, 2017.

\bibitem{feng2021hypergraph}
Song Feng, Emily Heath, Brett Jefferson, Cliff Joslyn, Henry Kvinge, Hugh~D
  Mitchell, Brenda Praggastis, Amie~J Eisfeld, Amy~C Sims, Larissa~B Thackray,
  et~al.
\newblock Hypergraph models of biological networks to identify genes critical
  to pathogenic viral response.
\newblock {\em BMC bioinformatics}, 22(1):1--21, 2021.

\bibitem{hazan2012near}
Elad Hazan, Satyen Kale, and Shai Shalev-Shwartz.
\newblock Near-optimal algorithms for online matrix prediction.
\newblock In {\em Conference on Learning Theory}, pages 38--1. JMLR Workshop
  and Conference Proceedings, 2012.

\bibitem{hidakaFastExactSearch2018}
Shohei Hidaka and Masafumi Oizumi.
\newblock Fast and exact search for the partition with minimal information
  loss.
\newblock {\em PLoS One}, 13(9):e0201126, 2018.

\bibitem{jambulapatiPositiveSemidefiniteProgramming2021}
Arun Jambulapati, Yin~Tat Lee, Jerry Li, Swati Padmanabhan, and Kevin Tian.
\newblock Positive semidefinite programming: mixed, parallel, and
  width-independent.
\newblock In {\em Proceedings of the 52nd Annual ACM SIGACT Symposium on Theory
  of Computing}, pages 789--802, 2020.

\bibitem{jiangFasterInteriorPoint2020}
Haotian Jiang, Tarun Kathuria, Yin~Tat Lee, Swati Padmanabhan, and Zhao Song.
\newblock A faster interior point method for semidefinite programming.
\newblock In {\em 2020 IEEE 61st annual symposium on foundations of computer
  science (FOCS)}, pages 910--918. IEEE, 2020.

\bibitem{JohnsonLindenstraussOriginal}
William Johnson and J.~Lindenstrauss.
\newblock Extensions of lipschitz mappings into a hilbert space.
\newblock {\em Conference in Modern Analysis and Probability}, 26:189--206, 01
  1982.

\bibitem{kale2007efficient}
Satyen Kale.
\newblock {\em Efficient algorithms using the multiplicative weights update
  method}.
\newblock Princeton University, 2007.

\bibitem{khandekar2007cut}
Rohit Khandekar, Subhash Khot, Lorenzo Orecchia, and Nisheeth~K Vishnoi.
\newblock On a cut-matching game for the sparsest cut problem.
\newblock {\em Univ. California, Berkeley, CA, USA, Tech. Rep.
  UCB/EECS-2007-177}, 6(7):12, 2007.

\bibitem{khandekarGraphPartitioningUsing2009}
Rohit Khandekar, Satish Rao, and Umesh Vazirani.
\newblock Graph partitioning using single commodity flows.
\newblock {\em Journal of the ACM (JACM)}, 56(4):1--15, 2009.

\bibitem{klamt2009hypergraphs}
Steffen Klamt, Utz-Uwe Haus, and Fabian Theis.
\newblock Hypergraphs and cellular networks.
\newblock {\em PLoS computational biology}, 5(5):e1000385, 2009.

\bibitem{kwok2022cheeger}
Tsz~Chiu Kwok, Lap~Chi Lau, and Kam~Chuen Tung.
\newblock Cheeger inequalities for vertex expansion and reweighted eigenvalues.
\newblock In {\em 2022 IEEE 63rd Annual Symposium on Foundations of Computer
  Science (FOCS)}, pages 366--377. IEEE, 2022.

\bibitem{lau2022cheeger}
Lap~Chi Lau, Kam~Chuen Tung, and Robert Wang.
\newblock Cheeger inequalities for directed graphs and hypergraphs using
  reweighted eigenvalues.
\newblock In {\em Proceedings of the 55th Annual ACM Symposium on Theory of
  Computing}, pages 1834--1847, 2023.

\bibitem{lau2023fast}
Lap~Chi Lau, Kam~Chuen Tung, and Robert Wang.
\newblock Fast algorithms for directed graph partitioning using flows and
  reweighted eigenvalues.
\newblock {\em arXiv preprint arXiv:2306.09128}, 2023.

\bibitem{lawler1982computing}
Eugene~L Lawler and Charles~U Martel.
\newblock Computing maximal “polymatroidal” network flows.
\newblock {\em Mathematics of Operations Research}, 7(3):334--347, 1982.

\bibitem{leightonMulticommodityMaxflowMincut1999}
Tom Leighton and Satish Rao.
\newblock Multicommodity max-flow min-cut theorems and their use in designing
  approximation algorithms.
\newblock {\em Journal of the ACM (JACM)}, 46(6):787--832, 1999.

\bibitem{li2017inhomogeneous}
Pan Li and Olgica Milenkovic.
\newblock Inhomogeneous hypergraph clustering with applications.
\newblock {\em Advances in neural information processing systems}, 30, 2017.

\bibitem{liSubmodularHypergraphsPLaplacians2018}
Pan Li and Olgica Milenkovic.
\newblock Submodular hypergraphs: p-laplacians, cheeger inequalities and
  spectral clustering.
\newblock In {\em International Conference on Machine Learning}, pages
  3014--3023. PMLR, 2018.

\bibitem{long2022near}
Yaowei Long and Thatchaphol Saranurak.
\newblock Near-optimal deterministic vertex-failure connectivity oracles.
\newblock In {\em 2022 IEEE 63rd Annual Symposium on Foundations of Computer
  Science (FOCS)}, pages 1002--1010. IEEE, 2022.

\bibitem{louis2010cut}
Anand Louis.
\newblock Cut-matching games on directed graphs.
\newblock {\em arXiv preprint arXiv:1010.1047}, 2010.

\bibitem{louis2015hypergraph}
Anand Louis.
\newblock Hypergraph markov operators, eigenvalues and approximation
  algorithms.
\newblock In {\em Proceedings of the forty-seventh annual ACM symposium on
  Theory of computing}, pages 713--722, 2015.

\bibitem{louis2016approximation}
Anand Louis and Yury Makarychev.
\newblock Approximation algorithms for hypergraph small-set expansion and
  small-set vertex expansion.
\newblock {\em Theory of Computing}, 12(1):1--25, 2016.

\bibitem{nanongkai2017dynamic}
Danupon Nanongkai and Thatchaphol Saranurak.
\newblock Dynamic spanning forest with worst-case update time: adaptive, las
  vegas, and $\mathcal{O}(n^{1/2-\varepsilon})$-time.
\newblock In {\em Proceedings of the 49th Annual ACM SIGACT Symposium on Theory
  of Computing}, pages 1122--1129, 2017.

\bibitem{narasimhanQClustering2005}
Mukund Narasimhan, Nebojsa Jojic, and Jeff~A Bilmes.
\newblock Q-clustering.
\newblock {\em Advances in Neural Information Processing Systems}, 18, 2005.

\bibitem{LorenzosThesis}
Lorenzo Orecchia.
\newblock {\em Fast Approximation Algorithms for Graph Partitioning using
  Spectral and Semidefinite-Programming Techniques}.
\newblock PhD thesis, EECS Department, University of California, Berkeley, May
  2011.

\bibitem{orecchia2022practical}
Lorenzo Orecchia, Konstantinos Ameranis, Charalampos Tsourakakis, and Kunal
  Talwar.
\newblock Practical almost-linear-time approximation algorithms for hybrid and
  overlapping graph clustering.
\newblock In {\em International Conference on Machine Learning}, pages
  17071--17093. PMLR, 2022.

\bibitem{orecchia2012balanced_separator}
Lorenzo Orecchia, Sushant Sachdeva, and Nisheeth~K Vishnoi.
\newblock Approximating the exponential, the lanczos method and an
  $\tilde{\mathcal{o}}(m)$-time spectral algorithm for balanced separator.
\newblock In {\em Proceedings of the forty-fourth annual ACM symposium on
  Theory of computing}, pages 1141--1160, 2012.

\bibitem{orecchiaPartitioningGraphsSingle2008}
Lorenzo Orecchia, Leonard~J Schulman, Umesh~V Vazirani, and Nisheeth~K Vishnoi.
\newblock On partitioning graphs via single commodity flows.
\newblock In {\em Proceedings of the fortieth annual ACM symposium on Theory of
  computing}, pages 461--470, 2008.

\bibitem{shermanBreakingMulticommodityFlow2009}
Jonah Sherman.
\newblock Breaking the multicommodity flow barrier for $\mathcal{O}(\sqrt{\log
  n})$-approximations to sparsest cut.
\newblock In {\em 2009 50th Annual IEEE Symposium on Foundations of Computer
  Science}, pages 363--372. IEEE, 2009.

\bibitem{svitkinaSubmodularApproximationSamplingbased2010}
Zoya Svitkina and Lisa Fleischer.
\newblock Submodular approximation: Sampling-based algorithms and lower bounds.
\newblock {\em SIAM Journal on Computing}, 40(6):1715--1737, 2011.

\bibitem{tsourakakis2017scalable}
Charalampos~E Tsourakakis, Jakub Pachocki, and Michael Mitzenmacher.
\newblock Scalable motif-aware graph clustering.
\newblock In {\em Proceedings of the 26th International Conference on World
  Wide Web}, pages 1451--1460, 2017.

\bibitem{veldt2023cut}
Nate Veldt.
\newblock Cut-matching games for generalized hypergraph ratio cuts.
\newblock In {\em Proceedings of the ACM Web Conference 2023}, WWW '23, page
  694–704, New York, NY, USA, 2023. Association for Computing Machinery.

\bibitem{veldt2021approximate}
Nate Veldt, Austin~R Benson, and Jon Kleinberg.
\newblock Approximate decomposable submodular function minimization for
  cardinality-based components.
\newblock {\em Advances in Neural Information Processing Systems},
  34:3744--3756, 2021.

\bibitem{veldtHypergraphCutsGeneral2020}
Nate Veldt, Austin~R Benson, and Jon Kleinberg.
\newblock Hypergraph cuts with general splitting functions.
\newblock {\em SIAM Review}, 64(3):650--685, 2022.

\bibitem{yang2017hypergraph}
Wenyin Yang, Guojun Wang, Md~Zakirul~Alam Bhuiyan, and Kim-Kwang~Raymond Choo.
\newblock Hypergraph partitioning for social networks based on information
  entropy modularity.
\newblock {\em Journal of Network and Computer Applications}, 86:59--71, 2017.

\bibitem{yoshidaCheegerInequalitiesSubmodular2018}
Yuichi Yoshida.
\newblock Cheeger inequalities for submodular transformations.
\newblock In {\em Proceedings of the Thirtieth Annual ACM-SIAM Symposium on
  Discrete Algorithms}, pages 2582--2601. SIAM, 2019.

\end{thebibliography}

\newpage
\appendix
\section{Omitted Proofs}
\label{sec.appendix.omitted}

In this section, we prove some technical lemmata used in the body of the paper.

\subsection{\nameref{sec.results}}
\label{sec.appendix.omitted.results}
\lemcontinuous*
\begin{proof}[{\bf Proof of Lemma \ref{lem.continuous-to-discrete-equivalence}}]
Observe that the denominator of the ratio-cut objective is a monotone submodular cut function applied over the entire vertex set
\begin{equation*}
\delta_{\vmu}(S) = \min \{ \vmu(S), \vmu(\bar{S}) \},
\end{equation*}
Consequently, for every $S \subseteq V$:
\begin{equation*}
\Psi_{G}(S)
=  \frac{\sum_{h\in E}w_h\delta_h(S\cap h)}{\min \{ \vmu(S), \vmu(\bar{S}) \}}
= \frac{\delta_G(S)}{\delta_{\vmu}(S)} = { \overline{\delta}_G(\ones^S)\over \overline{\delta}_{\vmu}(\ones^S)},
\end{equation*}

\noindent
which gives:
\[
    \Psi_G^* \geq \min_{\vx \in \R^n} {\overline{\delta}_G(\vx)\over  \overline{\delta}_{\vmu}(\vx)}.
\]
We now prove the opposite inequality by showing the second part of the lemma. Notice that the value of the ratio in the right-hand side of~\eqref{eq.rc-noncvx} is invariant under shifting all the coordinates by a fixed constant.
Furthermore, both the numerator and the denominator of the right-hand side of~\eqref{eq.rc-noncvx} are homogeneous and hence their ratio is invariant under scaling $\vx$ by a constant, i.e., for any $u \in \R$:
\begin{equation*}
{\overline{\delta}_G(u\cdot \vx) \over \overline{\delta}_{\vmu}(u \cdot \vx)} ={u \cdot \overline{\delta}_G(\vx) \over u \cdot \overline{\delta}_{\vmu}(\vx)} = {\overline{\delta}_G(\vx) \over \overline{\delta}_{\vmu}(\vx)}.
\end{equation*}
Hence, for any $\vx$, by letting $\hat{\vx}$ be defined as:
\[
    \hat{x_i} = {x_i - \min_j x_j \over \max_{j,k}(x_j - x_k)},
\]
\noindent
we have  $\max_i \hat{x}_i = 1$, $\min_i \hat{x}_i = 0$, and:
\[
   {\overline{\delta}_G(\vx) \over \overline{\delta}_{\vmu}(\vx)} = {\overline{\delta}_G(\hat{\vx}) \over \overline{\delta}_{\vmu}(\hat{\vx})}.
\]
We then have, as per Propositions 3.1 and 3.1 in \cite{bach2013learning}:
\[
    {\overline{\delta}_G(\vx) \over \overline{\delta}_{\vmu}(\vx)}= {\overline{\delta}_G(\hat{\vx}) \over \overline{\delta}_{\vmu}(\hat{\vx})} = {\E_{t\sim [0,1]} [\delta_G(S_t(\hat{\vx}))]\over \E_{t\sim [0,1]} [\delta_{\vmu}(S_t(\hat{\vx}))]}.
\]
Hence, there must exist some $t \in [0,1]$ such that:
\[
    {\overline{\delta}_G(\vx) \over \overline{\delta}_{\vmu}(\vx)} \geq { \delta_G(S_t(\hat{\vx}))\over  \delta_{\vmu}(S_t(\hat{\vx}))}.
\]
This gives:
\[
    \min_{\vx \in \R^n} {\overline{\delta}_G(\vx) \over \overline{\delta}_{\vmu}(\vx)} = \Psi_G^*,
\]
as needed for the first part of the lemma. Moreover, computing and checking all of the sets $\{S_t(\vx)\}_{t\in [0,1]}$ only requires time $O(n \log n)$ for sorting the entries of $\vx$ and time $O(\sparsity(G) + |V|)$ for evaluating the numerator and denominator.
\end{proof}

\subsection{\nameref{sec.polymatroidal}}
\label{sec.appendix.omitted.polymatroidal}
\submodularsymmetrization*
\begin{proof}
Let $S,T \subseteq V$ and denote with $\overline{S}$ and $\overline{T}$ their complements in $V$. We consider four possible cases, corresponding to four ways of determining $\delta(S)$ and $\delta(T)$. In the first two cases, we exploit the submodularity of $F$ and $G$.
\begin{description}
\item[Case 1:] $F(S) \leq G(\overline{S})$ and $F(T) \leq G(\overline{T})$, which implies $\delta(S) = F(S)$ and $\delta(T) = F(T).$ We have:
\begin{align*}
    \delta(S \cup T) + \delta(S \cap T) &= \min\{F(S \cup T), \; G(\overline{S} \cap \overline{T})\} + \min\{F(S \cap T), G(\overline{S} \cup \overline{T})\}\\
    &\leq F(S\cup T) + F(S\cap T)\\
    &\leq F(S) + F(T) = \delta(S) + \delta(T).
\end{align*}
\item[Case 2:]
$G(\overline{S}) \leq F(S)$ and $G(\overline{T}) \leq F(T)$, which implies $\delta(S) = G(\overline{S})$ and $\delta(T) = G(\overline{T}).$ We have:
\begin{align*}
    \delta(S \cup T) + \delta(S \cap T) &= \min\{F(S \cup T), \; G(\overline{S} \cap \overline{T})\} + \min\{F(S \cap T), G(\overline{S} \cup \overline{T})\}\\
    &\leq G(\overline{S} \cap \overline{T}) + G(\overline{S} \cup \overline{T})\\
    &\leq G(\overline{S}) + G(\overline{T}) = \delta(S) + \delta(T).
\end{align*}

In the last two cases, we rely on the monotonocity of $F$ and $G$.
\item[Case 3:] $F(S) \leq G(\overline{S})$ and $G(\overline{T}) \leq F(T)$, which implies $\delta(S) = F(S)$ and $\delta(T) = G(\overline{T})$. We have:
\begin{align*}
    \delta(S \cup T) + \delta(S \cap T) &= \min\{F(S \cup T), \; G(\overline{S} \cap \overline{T})\} + \min\{F(S \cap T) , G(\overline{S} \cup \overline{T})\}\\
    &\leq G(\overline{S}\cap \overline{T}) + F(S \cap T) \\
    &\leq G(\overline{T}) + F(S)   = \delta(S) + \delta(T).
\end{align*}

\item[Case 4:] $G(\overline{S}) \leq F(S)$ and $F(T) \leq G(\overline{T})$, which implies  $\delta(S) = G(\overline{S})$ and $\delta(T) = F(T)$. We have:
\begin{align*}
    \delta(S \cup T) + \delta(S \cap T) &= \min\{F(S \cup T), \; G(\overline{S} \cap \overline{T})\} + \min\{F(S \cap T) , G(\overline{S} \cup \overline{T})\}\\
    &\leq G(\overline{S}\cap \overline{T}) + F(S \cap T) \\
    &\leq G(\overline{S}) + F(T)   = \delta(S) + \delta(T).
\end{align*}
\end{description}
It follows that $\delta$ is submodular, as we have shown that for every $S,T\subseteq V$:
\[
    \delta(S\cup T) + \delta (S \cap T)  \leq \delta(S) + \delta(T).
\]
\end{proof}

\formoflovaszextension*
\begin{proof}
Since $F^-_h$ and $F^+_h$ are monotone, for any $\vx \in \R^n$ there exists a value $z^* = z^*(\vx)$ such that:
\[
    F^-_h(\{i\in V \mid \vx(i) \geq z\}) \leq F^+_h(\{i\in V \mid \vx(i) < z\}),
\]
for all $z\geq z^*$, and:
\[
    F^-_h(\{i\in V \mid \vx(i) \geq z\}) \geq F^+_h(\{i\in V \mid \vx(i) < z\}),
\]
for all $z < z^*$.
We then have, by Proposition 3.1(c) in \cite{bach2013learning}:
\begin{align*}
\overline{\delta}_h (\vx) &= \int_{-\infty}^\infty \delta(\{i\in V \mid \vx(i) \geq z\}) dz \\
&=\int_{\infty}^{z^*(\vx)} F^+_h(\{i \in V \mid \vx(i) < z\}) dz + \int_{z^*(\vx)}^\infty F^-_h(\{i \in V \mid \vx(i) \geq z\})dz\\
&= \min_{\nu \in \R} \int_{\infty}^0 F^+_h(\{i \in V \mid \vx(i) -\nu < z\}) dz + \int_{0}^\infty F^-_h(\{i \in V \mid \vx(i) -\nu \geq z\})dz\\
&= \min_{\nu \in \R}\bar{F}^-_h((\vx - \nu \vone)_+ ) + \bar{F}^+_h((\vx - \nu \vone)_-),
\end{align*}
as needed.
\end{proof}

\formoflovaszextensionsym
\begin{proof}
    Since $F_h$ is submodular, its Lov\'asz extension $\bar{F}_h$ is convex, and hence, by Fact~\ref{fact.form-of-lovasz-extension}:
    \begin{align*}
        \bar{\delta}_h(\vx) &= \min_{\nu \in \R}\bar{F}_h((\vx - \nu \vone_h)_+ ) + \bar{F}_h((\vx - \nu \vone_h)_-) =  \min_{\nu \in \R} 2 \cdot \left({1\over 2}\bar{F}_h((\vx - \nu \vone_h)_+ ) + {1\over 2}\bar{F}_h((\vx - \nu \vone_h)_-)\right)\\
        &\geq  \min_{\nu \in \R} 2 \cdot \bar{F}_h\left({1\over 2}(\vx - \nu \vone_h)_+  + {1\over 2}(\vx - \nu \vone_h)_-\right)\\
        &=  \min_{\nu \in \R} 2 \cdot \bar{F}_h\left({1\over 2}|\vx - \nu \vone_h|\right)\\
        &=  \min_{\nu \in \R}\bar{F}_h\left(|\vx - \nu \vone_h|\right),
    \end{align*}
    where in the above we applied Jensen's Inequality and the fact that $\bar{F}_h$ is positive homogenous (Proposition 3.1 (e) from \cite{bach2013learning}). On the other hand, since $F_h$ is monotone, we have:
    \begin{align*}
        \bar{\delta}_h(\vx) &= \min_{\nu \in \R}\bar{F}_h((\vx - \nu \vone_h)_+ ) + \bar{F}_h((\vx - \nu \vone_h)_-) \leq \min_{\nu \in \R}\bar{F}_h(|\vx - \nu \vone_h|) + \bar{F}_h(|\vx - \nu \vone_h|) = 2\cdot \min_{\nu \in \R}\bar{F}_h(|\vx - \nu \vone_h|),
    \end{align*}
    where the inequality follow from Fact~\ref{fct.lovaszmonotone}. This completes the proof.
\end{proof}

\flowembedding*
\begin{proof}
Consider the hypergraph flows $\{\vY^e\}_{e \in E^H}$ associated with the embedding $H \preceq_{\rho} G$ and let $\vu \in \R^{E_G}$ be an arbitrary vector over the hyperedges. As each $\vY^{e}$ routes demand $w^H_{e}$ between the endpoints of $e=\{u,v\} \in E_H$, we have
\begin{align*}
w^H_{e} \cdot (\vx_u - \vx_v)_+ &=
(\langle \dem(\vY^e), \vx \rangle)_+ =
\left(\sum_{h \in E_G}  \langle \vy^e_h , \vx_h\rangle\right)_+ =
\left(\sum_{h \in E_G}  \langle \vy^e_h , \vx_h - u_h \vone_h \rangle\right)_+ \\
& \leq \sum_{h \in E_G}  \langle (\vy^e_h)_+ , (\vx_h - u_h \vone_h)_+ \rangle +  \langle (\vy^e_h)_- , (\vx_h - u_h \vone_h)_- \rangle.
\end{align*}
\noindent
where the last equality relies on the fact that $\vy^e_h \in \R^h \perp \vone_h.$
Summing over all edges in $E_h$ yields:
\begin{align*}
\sum_{e=\{i,j\} \in E_H} w^H_{ij} \cdot (\vx_i - \vx_j)_+ \leq
&\sum_{e \in E_H} \sum_{h \in E_G}    \langle (\vy^e_h)_+ , (\vx_h - u_h \vone_h)_+ \rangle + \langle (\vy^e_h)_- , (\vx_h - u_h \vone_h)_- \rangle, \\
= &
 \sum_{h \in E_G}   \left\langle \sum_{e \in E_H}(\vy^e_h)_+ , (\vx_h - u_h \vone_h)_+\right \rangle +  \left\langle \sum_{e \in E_H}(\vy^e_h)_- , (\vx_h - u_h \vone_h)_- \right \rangle\\
\leq & \rho \cdot \sum_{h \in E_G} w^G_h \cdot \left(\bar{F}_h^-((\vx_h - u_h \vone_h)_+) + \bar{F}_h^+((\vx_h - u_h \vone_h)_-)\right)
\end{align*}
The second inequality follows from Definition~\ref{def.flow-embedding} and the characterization of Lov\'asz extensions by the positive submodular polytope (see Proposition 3.4 in~\cite{bach2013learning}).
Finally, the first statement in the theorem follows by setting the arbitrary $\vu$ to be the required minimizer. The specialization to cuts follows by plugging in $\vx = \vone^S$.

In the symmetric case, the capacity constraints in the definition become simply, for all $h \in E_G$:
\begin{equation*}
\sum_{e \in E_h} (\vy^e_h)_+ \,, \sum_{e \in E_h} (\vy^e_h)_-  \in \rho \cdot  w_h \cdot \cP_+(F_h).
\end{equation*}
As a result, the orientation of each arc $e=(i,j) \in E_H$ can be discarded in this case, as the flow $\vY^e$ can be oriented arbitrarily while satisfying the capacity constraint. Therefore, we can repeat the same analysis for the flow embedding given by $\{-\vY^e\}_{e \in E_H}$. Summing the positive and negative parts, we obtain the required statement.
\end{proof}

\flowdecomposition*
\begin{proof}[{\bf Proof Theorem~\ref{theorem.hypergraph-flow-decomposition}}]
Consider the hypergraph flow as a flow over the factor graph $\hat{G}.$ By a standard application of dynamic trees (see for example Lemma 3.8 in~\cite{shermanBreakingMulticommodityFlow2009}) to the computation of flow-path decompositions~\cite{ahujaravindrakandmagnantithomaslandorlinjamesbNetworkFlowsTheory1993} over $\hat{G}$ , we can compute, in time $\tilde{O}(|\hat{E}|_G)=\tilde{O}(\sparsity(G))$, a graph $H=(V,E_H, \vw^H),$  such that $|E_H|=\tilde{O}(\sparsity(G))$ and there exist hypergraph flows $\{\vY^e\}_{e \in E_H}$ in $G$ with the following properties:
\begin{itemize}
    \item $\sum_{e \in E_H} \vY^e = \vY,$
    \item For any $e \in E_H$, for any factor graph edge $\{i,h\}$, the flow $(\vY^e_h)_i$ has the same sign as $(\vY_h)_i$. In particular, for all $h \in E_G$:
    \begin{equation}\label{equation.absolute-decomposition}
    \sum_{e \in E_H} (\vy^e_h)_+  = (\vy_h)_+ \; \textrm{ and }\; \sum_{e \in E_H} (\vy^e_h)_-  = (\vy_h)_-
    \end{equation}
    \item For all $e=(i,j) \in E^H$, we have $\dem_i(\vY) \geq 0$ and $\dem_j(\vY) < 0.$  The flow $\vY^e$ routes $w^H_{e}$ units of flow between vertex $i$ and vertex $j$, i.e.,
    $$
    \dem(\vY^e)= w^H_{e}(\vone_i - \vone_j).
    $$
\end{itemize}
\noindent
The last point proves the required statements about the bipartiteness of $H$ and its degrees. To complete the proof, we verify that $H \preceq_{\cong_G(Y)} G$.
By the definition of congestion in Equation~\ref{eq.congestion}, we have that, for all $h \in E_G:$
$$
\vy_h \in \cong_{G(\vY)}
\cdot w_h \cdot \cB(\delta_h).
$$
By the characterization of the base polytope in Equation~\ref{eq.base-polytope}, this implies that:
$$
(\vy_h)_+ \in \cong_{G(\vY)}
\cdot w_h \cdot \cP(F^-_h) \; \textrm{ and } \; (\vy_h)_- \in \cong_{G(\vY)}
\cdot w_h \cdot \cP(F^+_h).
$$
Together with Equation~\ref{equation.absolute-decomposition} and Definition~\ref{def.flow-embedding}, this completes the proof.
\end{proof}

\subsection{\nameref{sec.sdp-algorithm}}
\label{sec.appendix.omitted.sdp-algorithm}
\symrelaxation*
\begin{proof}
Given any cut $(S,\overline{S})$ with $\mu(S) \leq \mu(\overline{S})$, for each $i \in V$,
set the vector embedding of the vertices to
\begin{equation*}
\vv_i = \begin{cases}
\begin{pmatrix} \sqrt{\frac{\mu(V)}{\mu(S) \mu(\bar{S})}} & 0 & \ldots & 0 \end{pmatrix}^{\top} & \mbox{if } i \in S \\
\vzero & \mbox{otherwise.}
\end{cases}
\end{equation*}

\noindent
And for each hyperedge $h\in E$:
\[
    \vv_h = \begin{cases}
            \boldsymbol{0} & \text{if } \delta_h(S) = F_h(S),\\
            \left(\sqrt{\mu(V) \over \mu(S)\mu(\overline{S})}, 0 , \cdots , 0 \right)^\top &\text{if } \delta_h(S) = F_h(\overline{S}).
        \end{cases}
\]
We then have:
\[
    \vd_i^h = \norm{\vv_i - \vv_h}^2.
\]

It is easy to see that these vector embedding satisfies the triangle inequality constraints. For the variance constraint, we have:
\[
    \sum_{\{i,j\} \subseteq V} {\mu_i \mu_j\over \mu(V)} \norm{\vv_i - \vv_j}^2 = \sum_{\substack{i\in S\\ j\in \overline{S}}} {\mu_i\mu_j\over \mu(V)} \cdot {\mu(V) \over \mu(S) \cdot \mu(\overline{S})} =1.
\]

Furthermore, this solution has objective value:
$$
\sum_{h \in E} w_h \bar{F}_h(\vd^h) = \frac{\mu(V)}{\mu(S) \mu(\bar{S})} \cdot \sum_{h \in E} w_h \delta_h(S\cap h) \leq 2 \cdot \frac{\sum_{h \in E} w_h \delta_h(S \cap h) }{\min\{\mu(S), \mu(\bar{S})\}} = 2 \Psi_G(S).
$$
This shows that the combinatorial solution $(S,\overline{S})$ can be mapped to a feasible solution to the \ref{eqn.vector-program} while preserving the objective to within a factor of two.
\end{proof}

\genrelaxation*
\begin{proof}
    Let $(S,\overline{S})$ be a candidate solution to the minimum hypergraph ratio-cut problem satisfying $\mu(S)
    \leq \mu(\overline{S})$, with objective value $\kappa$, i.e.:
    \[
        \Psi_G(S) = \frac{\sum_{h \in E} w_h \delta_h(S\cap h)}{\mu(S)}  = \kappa.
    \]
    We can construct a solution to the~\ref{eqn.gen-vector-program} program with a objective value $4\kappa$ as follows. Partition $E$ into $E^+ \cup E^-$ where:
    \[
        E^+ \defeq \{h\in E \mid \delta_h(S) = F_h^+(\overline{S})\} \hspace{0.5cm}\text{and}\hspace{0.5cm} E^- \defeq \{h\in E \mid \delta_h(S) = F_h^-(S)\}.
    \]
    We can consider:
    \[
        \vv_i = \begin{cases}
                   \left(\sqrt{\mu(V) \over \mu(S)\mu(\overline{S})}, 0 , \cdots , 0 \right)^\top &\text{if } i\in S,\\
                   \boldsymbol{0} & \text{if }i\in \overline{S},
                \end{cases}
        \hspace{1cm}\text{and}\hspace{1cm}
        \vv_h = \begin{cases}
            \boldsymbol{0} & \text{if } h\in E^-,\\
            \left(\sqrt{\mu(V) \over \mu(S)\mu(\overline{S})}, 0 , \cdots , 0 \right)^\top &\text{if } h \in E^+.
        \end{cases}
    \]
    and for any $h\in E$ let $\vd_h^+$ and $\vd_h^-$ be defined in terms of $\{\vv_i\}_{i\in V}$ and $\{\vv_h\}_{h\in E}$ so as to satisfy the equality constraints in the \ref{eqn.gen-vector-program} program. One can check that:
    \[
        \vd_h^+(i) = \begin{cases}
            2 \cdot{\mu(V)\over\mu(S) \mu(\overline{S})}&\text{if }i\in \overline{S}\text{ and }h\in E^+,\\
            0 & \text{otherwise,}
        \end{cases}
        \hspace{0.5cm}\text{ and }\hspace{0.5cm}\vd_h^-(i) = \begin{cases}
            2 \cdot{\mu(V)\over\mu(S) \mu(\overline{S})}&\text{if }i\in S\text{ and }h\in E^-,\\
            0 & \text{otherwise.}
        \end{cases}
    \]

    It easy to see that all the triangle inequality constraints are satisfied. We can also verify that the variance constraint is satisfied:
    \begin{align*}
        \sum_{\{i,j\} \subseteq V} {\mu(i) \mu(j) \over \mu(V)} \cdot \|\vv_i - \vv_j\|_2^2 &= \sum_{\substack{i \in S\\ j \in \overline{S}}} {\mu(i)\mu(j) \over \mu(V)} \cdot {\mu(V) \over \mu(S)\cdot \mu(\overline{S})} = 1.
    \end{align*}

    Hence, the vectors defined above form a feasible solution to the \ref{eqn.gen-vector-program} program. The objective value of this solution is:
    \begin{align*}
    \sum_{h\in E^-} w_h \overline{F_h^-}(\vd_h^-) + \sum_{h\in E^+} w_h \overline{F_h^+}(\vd_h^+) &=2 \cdot \left(\sum_{h\in E^-} w_h {\mu(V) \over{\mu(S)\mu(\overline{S})}}\overline{F_h^-}(\vone_S) + \sum_{h\in E^+} w_h {\mu(V) \over{\mu(S)\mu(\overline{S})}} \overline{F_h^+}(\vone_{\overline{S}})\right)\\
        &= {2 \mu(V) \over{\mu(S)\mu(\overline{S})}} \cdot \left(\sum_{h\in E^-} w_h {F_h^-}(S\cap h) + \sum_{h\in E^+} w_h {F_h^+}(\overline{{S}}\cap h)\right)\\
        &={2 \mu(V) \over{\mu(S)\mu(\overline{S})}} \cdot \left(\sum_{h\in E} w_h {\delta_h}(S\cap h)\right)\\
        &\leq 4  \cdot {\sum_{h\in E} w_h {\delta_h}(S\cap h)\over{\mu(S)}}\\
        &= 4\kappa.
    \end{align*}
    completing the proof.
\end{proof}

\directedembedding*
\begin{proof}
    Since the relaxation is invariant to translations, we will assume, without loss of generality, that:
    \begin{equation}\label{eqn.zero-mean-embedding}
        \sum_{i\in V} \mu(i)\vv_i = 0,
    \end{equation}
    and that:
    \begin{equation}\label{eq.normalized-variance}
        \sum_{i,j \in V} {\mu(i) \mu(j) \over\mu(V)}\norm{\vv_i - \vv_j}_2^2 = 1.
    \end{equation}

    We divide the proof in two cases corresponding to the case in which the embedding is balanced and unbalanced respectively. In Case 1, the embedding $\{\vv_i\}_{i\in V}$ satisfies:
    \[
        \sum_{\{i,j\}\in \binom{R_3}{2}} {\mu(i) \mu(j) \over \mu(V)} \norm{\vv_i - \vv_j}^2 \geq 1/10,
    \]
    where:
    \[
        R_{t} \defeq \left\{i\in V \mid \norm{\vv_i}^2 \leq {t\over \mu(V)}\right\}.
    \]
    In this case, it is implicit in the results of Arora Rao and Vazirani~\cite{ARV2009}, that there exists polynomial-time computable sets $S,T\subseteq V$ satisfying $\mu(S),\mu(T) = \Omega(\mu(V))$, and:
    \begin{equation}\label{eqn.separation-of-sets}
        \min_{\substack{i\in S\\ j \in T}} \norm{\vv_i - \vv_j}_2^2 \geq \Omega\left({1 \over \mu(V) \sqrt{\log n}}\right).
    \end{equation}

    We construct the map $\phi$ explicitly as follows. Let $S$ and $T$ be the sets satisfying~\eqref{eqn.separation-of-sets}.
    Let $r \in \R$ be a value such that:
    \[
        \mu(\{ i\in S \mid \norm{\vv_i}_2 \leq r\} ) = \mu(\{ i\in S \mid \norm{\vv_i}_2 \geq r\}),
    \]

    \noindent
    and define:
    \[
        S^- \defeq \{ i\in S \mid \norm{\vv_i}_2 \leq r\} , \hspace{1cm} S^+ \defeq \{ i\in S \mid \norm{\vv_i}_2 \leq r\},
    \]
    and:
    \[
        T^- \defeq \{ i\in S \mid \norm{\vv_i}_2 \leq r\} , \hspace{1cm} T^+ \defeq \{ i\in S \mid \norm{\vv_i}_2 \leq r\}.
    \]

    Suppose, without loss of generality, that $\mu(T^+) \geq \mu(T^-)$\footnote{If this condition does not hold, the rest of the argument still holds by exchanging the roles of $T^+$ with $T^-$ and $S^+$ with $S^-$.}. Note, in particular, that this implies that $\mu(S^-) \geq \mu(S)/2$ and $\mu(T^-) \geq \mu(T)/2$. Let:
    \[
        \phi(\vv_i) = \min_{j \in T^+} \norm{\vv_i - \vv_j}^2_2 + \norm{\vv_i}_2^2 -\norm{\vv_j}_2^2.
    \]
     As a simple consequence of the triangle inequality constraints, we have, for any $i,k \in V\cup E$:
    \[
        \phi(\vv_i) - \phi(\vv_k) \leq \norm{\vv_i - \vv_j}_2^2 + \norm{\vv_i}_2^2 - \norm{\vv_j}_2^2.
    \]
    Let $\vx \in \R^V$ be the vector given by $\vx(i) = \phi(\vv_i)$. Then:
    \begin{align*}
        \sum_{\{i,j\}\subseteq V} \mu(i) \mu(j) |\vx(i) - \vx(j)| &\geq \sum_{\{i,j\}\subseteq V} \mu(i) \mu(j) (\vx(i) - \vx(j))_+ \geq \sum_{\substack{i\in S^-\\ j\in T^+}} \mu(i) \mu(j) (\vx(i) - \vx(j))_+\\
        &=\sum_{\substack{i\in S^-\\ j\in T^+}} \mu(i) \mu(j) (\vx(i))_+ = \mu(T^+) \sum_{i\in S^-} \mu(i) (\vx(i))_+\\
        &= \mu(T^+) \sum_{i\in S^-} \mu(i) \left(\min_{j\in T^+}\norm{\vv_i -\vv_j}_2^2 +\norm{\vv_i}_2^2 - \norm{\vv_j}_2^2\right)_+\\
        &\geq \mu(T^+) \sum_{i\in S^-} \mu(i) \left(\min_{j\in T^+}\norm{\vv_i -\vv_j}_2^2 +r^2 - r^2\right)_+\\
        &\geq \mu(T^+) \sum_{i\in S^-} \mu(i) \left(\min_{j\in T^+}\norm{\vv_i -\vv_j}_2^2 \right)_+\\
        &\geq \mu(T^+) \mu(S^-) \cdot \min_{\substack{i\in S^-\\ j\in T^+}}\left(\norm{\vv_i -\vv_j}_2^2
        \right)_+\\
        &= \mu(T^+) \mu(S^-) \cdot\min_{\substack{i\in S^-\\ j\in T^+}}\norm{\vv_i -\vv_j}_2^2\\
        &\geq \mu(T^+) \mu(S^-) \cdot \min_{\substack{i\in S\\ j\in T}}\norm{\vv_i -\vv_j}_2^2\\
        &\geq {1\over 4}\mu(T) \mu(S) \cdot \min_{\substack{i\in S\\ j\in T}}\norm{\vv_i -\vv_j}_2^2\\
        &\geq \Omega(\mu(V)^2) \cdot \Omega\left({1 \over \mu(V) \sqrt{\log n}}\right)\\
        &\geq \Omega \left(\mu(V) \over \sqrt{\log n}\right).
    \end{align*}

    Giving that:
    \[
         {\sum_{\{i,j\}\subseteq V} \mu(i) \mu(j) |\vx(i) - \vx(j)| \over \sum_{\{i,j\}\subseteq V} \mu(i) \mu(j) \norm{\vv_i - \vv_j}^2_2} =  {1  \over \mu(V)}\sum_{\{i,j\}\subseteq V} \mu(i) \mu(j) |\vx(i) - \vx(j)|  \geq \Omega\left({1 \over \sqrt{\log n}}\right).
    \]

    Hence, the map $\phi$ satisfies conditions 1 and 2 in the statement of the lemma. Moreover, it is easy to see that this map can be computed in polynomial time. This completes the proof for Case 1.

    For Case 2, we consider what happens when Case 1 doesn't hold, i.e.:
    \[
        \sum_{\{i,j\}\in \binom{R_3}{2}} {\mu(i) \mu(j)\over \mu(V)} \norm{\vv_i - \vv_j}^2 < 1/10.
    \]
    Here, we have:
    \[
        {\mu(R_{3/2}) \over \mu(V)} \geq {1\over 3}  \hspace{1cm}\text{and}\hspace{1cm}\sum_{i\in V\setminus R_3} \mu_i \cdot \norm{\vv_i}^2 \geq {3\over 5}.
    \]
    Where the first inequality follows by \eqref{eqn.zero-mean-embedding}, \eqref{eq.normalized-variance} and Markov's inequality.
    Let $T = R_{3/2}$ and $S = V\setminus R_{3}$. Consider the embedding:
    \[
        \phi(\vv) = \min_{\vu\in T} \norm{\vv- \vu}^2_2 + \norm{\vv}_2^2 - \norm{\vu}_2^2.
    \]
    Again, it is easy to see that the map $\phi$ satisfies condition 1 of the lemma. Moreover, note that, for any $\vv_i \in S$:
    \[
        |\phi(\vv_i)| \geq |\norm{\vv_i}^2_2 - \norm{\vu}_2^2|,
    \]
    for some $\vu \in T$ and hence:
    \[
        |\phi(\vv_i)| \geq |\norm{\vv_i}^2_2 - \norm{\vu}_2^2| \geq {1\over 2} \norm{\vv_i}^2_2.
    \]
    Let $\vx\in \R^V$ be the vector given by $\vx(i) = \phi(\vv_i)$. We then have:
    \begin{align*}
        {\sum_{\{i,j\}\subseteq V} \mu(i) \mu(j) |\vx(i) - \vx(j)| \over \sum_{\{i,j\}\subseteq V} \mu(i) \mu(j) \norm{\vv_i - \vv_j}^2_2}&= {1\over \mu(V)}\sum_{\{i,j\}\in \binom{V}{2}} {\mu(i) \mu(j)} |\vx(i) -\vx(j)| = {\mu(T)\over \mu(V)} \sum_{i \in S} \mu(i)|\vx(i)| \\
        &\geq {\mu(T) \over 2\mu(V)} \sum_{i \in S} \mu_i\norm{\vv_i}_2^2 \\
        &\geq {1 \over 2}\cdot {1\over 3} \cdot {3 \over 5 } = \Omega(1).    \end{align*}

    Note that the bound obtained in this case is tighter, i.e. $\Omega(\nicefrac{1}{\sqrt{\log n}})$ has been replaced by a constant.
\end{proof}

\subsection{\nameref{sec.ci}}
\label{sec.appendix.omitted.ci}
\improveduality*
\begin{proof}[{\bf Proof of Lemma~\ref{lem.rc-improve.strong-duality}}]
Let us proceed with taking the conjugate dual. Starting with~\eqref{eq.rc-improve.primal}:
\begin{align*}
\min_{\substack{\vx \in \R^V \\ \langle \vs, \vx \rangle_{\mu} = 1}} \sum_{h \in E} w_h \cdot \bar{\delta}_h(\vx)
&= \min_{\vx \in \R^V} \sum_{h \in E} w_h \cdot \bar{\delta}_h(\vx) + \max_{\alpha \in \R} \Big( 1 - \langle \vs, \vx \rangle_{\mu} \Big) \\
&= \min_{\vx \in \R^V} \max_{\alpha \in \R} \bigg\{ \alpha + \sum_{h \in E}w_h \cdot \bar{\delta}_h(\vx) - \big\langle \alpha \big( \mu \circ \vs \big), \vx \big\rangle \bigg\} \, .
\end{align*}
Recall that the Lov\'{a}sz extension is given by
\begin{equation*}
\bar{\delta}_h(\vx)
= \max_{\vf \in \cB(\delta_h)} \langle \vf, \vx \rangle \, ,
\end{equation*}
and hence:
\begin{align*}
\min_{\vx \in \R^V} \max_{\alpha \in \R} \bigg\{ \alpha + \sum_{h \in E}w_h \cdot \bar{\delta}_h(\vx) - \big\langle \alpha \big( \mu \circ \vs \big), \vx \big\rangle \bigg\}
&= \min_{\vx \in \R^V} \max_{\alpha \in \R} \bigg\{
  \alpha + \sum_{h \in E} \max_{\vf_h \in w_h\cdot \cB(\delta_h)} \langle \vf_h, \vx \rangle - \big\langle \alpha \big( \mu \circ \vs \big), \vx \big\rangle
\bigg\} \\
&= \min_{\vx \in \R^V} \max_{\substack{\alpha \in \R, \vf_h \in \R^V \\ \vf_h \in w_h\cdot \cB(\delta_h)}} \bigg\{
  \alpha + \sum_{h \in E} \langle \vf_h, \vx \rangle - \big\langle \alpha \big( \mu \circ \vs \big), \vx \big\rangle
\bigg\} \\
&= \min_{\vx \in \R^V} \max_{\substack{\alpha \in \R, \vf_h \in \R^V \\ \vf_h \in w_h\cdot \cB(\delta_h)}} \bigg\{
  \alpha + \bigg\langle \sum_{h \in E} \vf_h - \alpha \big( \mu \circ \vs \big), \, \vx \bigg\rangle
\bigg\} \, .
\end{align*}
Strong duality always holds when the feasible region is polyhedron. Hence, we have:
\begin{equation*}
\min_{\vx \in \R^V} \max_{\substack{\alpha \in \R, \vf_h \in \R^V \\ \vf_h \in w_h\cdot \cB(\delta_h)}} \bigg\{
  \alpha + \bigg\langle \sum_{h \in E} \vf_h - \alpha \big( \mu \circ \vs \big), \, \vx \bigg\rangle
\bigg\}
= \max_{\substack{\alpha \in \R, \vf_h \in \R^V \\ \vf_h \in w_h\cdot \cB(\delta_h)}} \alpha
  + \min_{\vx \in \R^V} \bigg\langle \sum_{h \in E} \vf_h - \alpha \big( \mu \circ \vs \big), \, \vx \bigg\rangle,
\end{equation*}
which is equivalent to the required linear program by the definition of congestion (Equation~\ref{eq.congestion}) and demand (Equation~\ref{eq.demand})for hypergraph flows.
\end{proof}

\dualgraphcertificate*
\begin{proof}
Let $D = \frac{1}{\alpha}\cdot H.$ The sparsity of $D$ also follows directly from Theorem~\ref{theorem.hypergraph-flow-decomposition}. By the same theorem, we have $H \preceq_{1} G$, as $\cong_G(\vY)\leq 1$ by Property 1 of approximate dual solutions. Hence, the scaled graph $D$ embeds into $G$ with congestion $\frac{1}{\alpha},$ as required.
Taking into account the same scaling by $\frac{1}{\alpha}$,
the statements on the bipartiteness and degree bounds of $D$ follow directly from Property 2 and 3 of approximate dual solutions via Theorem~\ref{theorem.hypergraph-flow-decomposition}. Similarly, the largeness of $D$ follows form Property 2.
\end{proof}

\subsection{\nameref{sec.alg-cm}}
\label{sec.appendix.omitted.alg-cm}
\genpoly*

\polycut*

\begin{proof}
By Definition~\ref{def.graph}, each dual graph certificate $D_t$ played by the matching player has $\tilde{O}(\sparsity(G))$ edges. Hence, $H_t$ has at most $\tilde{O}(\sparsity(G))$ edges for all iterations $t \leq g(n).$ It follows that the running time of the cut strategy at every round is bound by $\tilde{O}(\sparsity(G))$. On the matching side, the running time of the approximate primal-dual oracle is given by Theorem~\ref{thm.general-solver} and ~\ref{thm.maxflow-solver}. Both dominate or are asymptotically equivalent to $\tilde{O}(\sparsity(G))$.
\end{proof}

\asymCutStratWidth*
\begin{proof}
Fix any $\vx \in \R^V$ such that $\big\langle \vx, \mM \vx \big\rangle = 1$. We can write the quadratic form $\big\langle \vx, \mL_{+}(D) \vx \big\rangle$ as
\begin{equation*}
\big\langle \vx, \mL_{+}(D) \vx \big\rangle
= \sum_{(i, j) \in E_D} w_{ij}^{D} \cdot \big( (x_i - x_j)^2 + x_i^2 - x_j^2 \big)
\end{equation*}

For the lower bound, note that
\begin{equation*}
\sum_{(i, j) \in E_D} w_{ij}^{D} \cdot \big( (x_i - x_j)^2 + x_i^2 - x_j^2 \big)
\geq - \sum_{(i, j) \in E_D} w_{ij}^{D} \cdot x_j^2
= - \sum_{j \in B} x_j^2 \sum_{i \in \partial^-(j)} w_{ij}^D
= - \sum_{j \in B} x_j^2 \cdot \deg_D(j)
\end{equation*}
Using the fact that the degrees of $i \in V$ in $D$ is bounded above by $\mu_i$, we have
\begin{equation*}
- \sum_{j \in B} x_j^2 \cdot \deg_D(j)
\leq - \sum_{j \in B} x_j^2 \cdot \mu_j
\leq - \big\langle \vx, \mM \vx \big\rangle
= -1
\end{equation*}

For the upper bound, note that
\begin{equation*}
\sum_{(i, j) \in E_D} w_{ij}^{D} \cdot \big( (x_i - x_j)^2 + x_i^2 - x_j^2 \big)
\leq \sum_{(i, j) \in E_D} w_{ij}^{D} \cdot \big( (x_i - x_j)^2 + x_i^2 \big)
\leq \sum_{(i, j) \in E_D} w_{ij}^{D} \cdot \big( 2 \cdot (x_i^2 + x_j^2) + x_i^2 \big)
\end{equation*}
Let's then exactract the dependence on the degree of each vertex.
\begin{align*}
\sum_{(i, j) \in E_D} w_{ij}^{D} \cdot \big( 2 \cdot (x_i^2 + x_j^2) + x_i^2 \big)
&= 2 \cdot \sum_{(i, j) \in E_D} w_{ij}^{D} \cdot (x_i^2 + x_j^2)
  + \sum_{(i, j) \in E_D} w_{ij}^D \cdot x_i^2 \\
&= 2 \cdot \bigg( \sum_{i \in A} x_i^2 \sum_{j \in \partial^{+}(i)} w_{ij}^{D} + \sum_{j \in B} x_{j}^2 \sum_{i \in \partial^{-}(j)} w_{ij}^{D} \bigg)
  + \sum_{i \in A} x_i^2 \sum_{j \in \partial^+(i)} w_{ij}^D \\
&= 2 \cdot \bigg( \sum_{i \in A} x_i^2 \cdot \deg_D(i) + \sum_{j \in B} x_{j}^2 \cdot \deg_D(j) \bigg)
  + \sum_{i \in A} x_i^2 \cdot \deg_D(i) \\
&= 2 \sum_{i \in V} x_i^2 \cdot \deg_D(i) + \sum_{i \in A} x_i^2 \cdot \deg_D(i)
\end{align*}
Finally, using the fact that the degrees are bounded by $\vmu$, we have
\begin{equation*}
2 \sum_{i \in V} x_i^2 \cdot \deg_D(i) + \sum_{i \in A} x_i^2 \cdot \deg_D(i)
\leq 2 \sum_{i \in V} x_i^2 \cdot \mu_i + \sum_{i \in A} x_i^2 \cdot \mu_i
\leq 3 \sum_{i \in V} x_i^2 \cdot \mu_i
= 3
\end{equation*}
thus establishing the required bound.
\end{proof}

\unbalancedprevious*
\begin{proof}
Because the embedding is normalized, the expect squared length of a vector is $1$. Markov's inequality then yields:
\begin{equation*}
\mu(R_{3/2}) \geq \bigg( 1 - \frac{1}{3/2} \bigg) \cdot \mu(V)
= \frac{1}{3} \cdot \mu(V) \, \textrm{ and } \, \mu(R_{3}) \geq \bigg( 1 - \frac{1}{3} \bigg) \cdot \mu(V) = \frac{2}{3} \cdot \mu(V)\,.
\end{equation*}
For the second part, by Lemma~\ref{lem.variance} and the definition of balanced embedding, we have that:
\begin{equation*}
\sum_{i \in \overline{R_3}} \mu_i \cdot \lVert \vv_i \rVert_2^2
\geq \vmu(R_3) \cdot  \bigg( 1 - \frac{1}{(\vmu(V))^2} \cdot \sum_{i,j \in \binom{R_3}{2}} \mu_i \mu_j \cdot \lVert \vv_i - \vv_j \rVert_2^2 \bigg) \geq \frac{2}{3} \cdot \mu(V) \cdot \frac{9}{10} = \frac{3}{5} \cdot \mu(V).
\end{equation*}
\end{proof}

\subsection{\nameref{sec.separated}}
\label{sec.appendix.omitted.separated}
We now complete the argument that \roundcut~outputs separated sets when the given vector embedding is balanced. We begin by recalling a fact about uniformly sampled vectors on the sphere.

\begin{lemma}
\label{lem.gaussian-facts}
Given $\vv \in \R^{d}$, and $\vg$ a uniformly sampled random vector in $\cS^{d-1}$, the following hold.
\begin{enumerate}
\item $\E_{\vg} \langle \vg, \vv \rangle^2 = \frac{\lVert \vv \rVert^2}{d}$

\item For any $\delta > \frac{d}{16}$, we have $\Pr_{\vg} \Big( \langle \vg, \vv \rangle^2 \geq \delta \cdot \frac{\lVert \vv \rVert^2}{d} \Big) \leq e^{-\delta / 4}$.
\end{enumerate}
\end{lemma}

We also require a lemma proved in~\cite{LorenzosThesis} which states that, with constant probability, one can extract two balanced partitions where projection lengths are well separated when projecting along a random direction.

\begin{lemma}[Lemma 5.5.6 item (3) of~\cite{LorenzosThesis}]
\label{lem.balanced-separated-set}
Given a set $V$ such that $\lvert V \rvert = n$, and a $\big( 3, \frac{1}{10} \big)$-balanced embedding $\{ \vv_i \}_{i \in V} \subseteq \R^d$, let $r_i = \sqrt{d} \cdot \langle \vg, \vv_i \rangle$ for a $\vg \in \cS^{d-1}$ and each $i \in V$. Assume, without loss of generality, that $r_1 \geq \ldots \geq r_n$. Then, there exists an absolute constant $r > 0$ such that, with $O(1)$ probability over a uniformly random choice of $\vg \in \cS^{d-1}$, there exists $1 \leq a < b \leq n$ satisfying the following:
\begin{enumerate}
\item $\mu(\{ 1, \ldots, a \}) \geq \frac{1}{3} \cdot \mu(V)$,

\item $\mu(\{ b, \ldots, n \}) \geq \frac{1}{3} \cdot \mu(V)$,

\item $\big( r_{a} - r_{b} \big)^2 \geq r \cdot \Var_{\mu}(\vv_i)$.
\end{enumerate}
\end{lemma}

Using the above two lemmata, we can now prove that \roundcut~produces a robustly separated set in the balanced case.

\robustSeparatedBalanced*

\begin{proof}
Define $\cE_1$ to be the event that there exist $1 \leq a < b \leq n$ such that the following hold simultaneously: (1) $\mu(\{ 1, \ldots, a \}) \geq \frac{1}{3} \cdot \mu(V)$, (2) $\mu(\{ b, \ldots, n \}) \geq \frac{1}{3} \cdot \mu(V)$, and (3) $(r_a - r_b)^2 \geq \frac{r}{\mu(V)}$. Next, for a constant $c_1 > 0$ to be chosen subsequently, and any $i, j \in V$, define the event $\cE_{2}(i, j)$ as
\begin{equation*}
\cE_{2}(i, j) \defeq \Big\{ (r_i - r_j)^2 < \big( 4 c_1 \cdot \log n \big) \cdot \lVert \vv_i - \vv_j \rVert_2^2 \Big\}
\qquad\textup{and}\qquad
\cE_2 \defeq \bigcap_{i,j \in V} \cE_{2}(i,j).
\end{equation*}

To begin, we observe that if $\cE_1$ and $\cE_2$ occur jointly, then the cuts $S = \{ 1, \ldots, a \}$ and $T = \{ b, \ldots, n \}$ are $\Omega\big( \frac{1}{\log n} \big)$-separated set  and, consequently, \roundcut~outputs the required separated set in the balanced case.
Event $\cE_2$ implies that
\begin{equation*}
\lVert \vv_i - \vv_j \rVert_2^2
> \frac{(r_i - r_j)^2}{4c_1 \cdot \log n}
\end{equation*}
holds for all $i \in S$, and $j \in T$. Item (3) defining event $\cE_1$ then determines:
\begin{equation*}
\frac{(r_i - r_j)^2}{4c_1 \cdot \log n}
\geq r \cdot \Var_{\mu}(\vv_i) \cdot \frac{1}{4c_1 \cdot \log n}
=  \frac{r}{4c_1 \cdot \log n} \, .
\end{equation*}
Finally, item (1) in event $\cE_1$ implies that $\mu(S), \mu(T) \geq \frac{\mu(V)}{3}$ and hence, for all $i \in S, j \in T$:
\begin{equation*}
\frac{\mu(S) \mu(T)}{(\mu(V))^2} \cdot \lVert \vv_i - \vv_j \rVert_2^2
\geq  \frac{r}{36 \cdot c_1 \cdot \log n} = \Omega\bigg( \frac{1}{\log n} \bigg)\,,
\end{equation*}
which establishes the separation guarantee.

Let us next bound the probability that $\cE_1$ and $\cE_2$ occur jointly. By Lemma~\ref{lem.balanced-separated-set}, there exists a constant $p > 0$ such that $\Pr_{\vg}(\cE_1) = p$ and hence $\Pr_{\vg}\big( \overline{\cE_1} \big) = 1 - p$. By Lemma~\ref{lem.gaussian-facts}, we have
\begin{align*}
\Pr_{\vg} \big( \overline{\cE_2(i,j)} \big)
&= \Pr_{\vg} \Big( (r_i - r_j)^2 \geq \big( 4 c_1 \cdot \log n \big) \cdot \lVert \vv_i - \vv_j \rVert_2^2 \Big) \\
&= \Pr_{\vg} \bigg( \langle \vg, \vv_i - \vv_j \rangle^2 \geq \big( 4 c_1 \cdot \log n \big) \cdot \frac{\lVert \vv_i - \vv_j \rVert_2^2}{d} \bigg)
\leq \frac{1}{n^{c_1}}
\end{align*}
Performing a union bound over $\overline{\cE_1}$ and $\overline{\cE_{2}(i,j)}$ for each $i, j \in V$ derives
\begin{equation*}
\Pr_{\vg} \big( \cE_1 \cap \cE_2 \big)
\geq p - \binom{n}{2} \frac{1}{n^{c_1}}
\geq p - \frac{1}{n^{c_1 - 2}}
= \Omega(1)
\end{equation*}
for $c_1$ large enough. We conclude that with constant probability, a single execution of step (1) in \roundcut~outputs $\Omega \big( \frac{1}{\log n} \big)$-separated sets.

Note that upon computing candidate sets, it is possible to check whether these are in fact $\Omega(1/\log n)$-separated. One can therefore boost the probability of success and obtain a high probability guarantee by repeating this procedure up to $O(\log n)$ many times. Note that this clearly affects the runtime of this procedure by at most a logarithmic factor.
\end{proof}

\section{Comparison to the cut-Improvement of Andersen \& Lang}
\label{sec.appendix.andersen-lang}

In this section, we briefly show that our formulation of the ratio-cut improvement problem in Definition~\ref{def.rc-improve} generalizes  the definition of the modified quotient cut score, introduced by Andersen and Lang~\cite{Andersen-Lang}, to potentially non-integral seed vectors (rather than cut seeds) and to arbitrary hypergraphs (rather than graphs).
Consider the following seed vector. Let $A \subseteq V$ be a cut, and define the vector $\vs^{A,\bar{A}} \in \R^V$ to be
\begin{equation}
\label{def.seed-vector}
\vs^{A,\bar{A}} \defeq \ones^A - \nicefrac{\mu(A)}{\mu(\bar{A})} \cdot \ones^{\bar{A}} \, .
\end{equation}
Assuming $\mu(A) \leq \mu(\bar{A})$, note that $\langle \vs^{A,\bar{A}}, \ones \rangle_{\mu} = 0$ and $\lVert \vs^{A,\bar{A}} \rVert_{\infty} \leq 1$, so that $\vs^{A,\bar{A}}$ is a valid seed vector. With this choice of seed vector, the ratio-cut improvement objective on a graph $G$ gives:
\begin{equation*}
\min\left\{\Psi_{G, \vs^{A,\bar{A}}}(S), \Psi_{G, \vs^{A,\bar{A}}}(\bar{S})\right\}
= \frac{\sum_{e \in E} w_h \cdot \delta_e(S)}{\lvert \mu(A \cap S) - \frac{\mu(A)}{\mu(\bar{A})} \cdot \mu(\bar{A} \cap S) \rvert},
\end{equation*}
which is exactly equal to the \emph{modified quotient cut score} of cut $(S,\bar{S})$ as considered by Andersen and Lang.

\section{Proofs of Theorems \ref{thm.general-solver} and \ref{thm.maxflow-solver}}
\label{sec.appendix.ci-algs}

Suppose we have an algorithm $\cA$ that,  on input $\alpha,$ either outputs a cut with a cut $S \subseteq V$ with $\Psi_{G,\vone^{A,B}}(S) \leq 2\alpha$ or an approximate dual solution of value $\alpha.$ Then, $\cA$ can be used to construct an approximate primal-dual oracle by performing binary search on $\alpha$ and converting the final approximate dual solution to an approximate dual graph certificate via the flow decomposition of Theorem~\ref{theorem.hypergraph-flow-decomposition}. Since the range of possible values for $\alpha$ is $[1/\mu(V),\sum_{h\in E} w_h]$, which we assume is polynomial in $|V|,$ it suffices to run $\cA$ for $O(\log |V|)$ iterations.
Hence, in the rest of the section, we focus on implementing $\cA$ for both cases in the theorems.

\paragraph{General Polymatroidal Cut Functions} For a fixed $\alpha \geq 0,$ we consider the integral version of the feasibility problem for \eqref{eq.rc-improve.primal}:
$$
\exists S \subseteq V : \sum_{h \in E} w_h \delta_h(S) - \alpha \cdot \langle \vone_S, \vone^{A,B} \rangle_{\vmu} \leq 0.
$$
The left-hand side of this equation can be minimized exactly by the proximal algorithms for decomposable submodular minimization (DSM) of Ene \etal~\cite{eneRandomCoordinateDescent2015, eneDecomposableSubmodularFunction2017} when run with $\epsilon=\frac{1}{n}$ and rounded to integral solutions.
If the resulting minimum is less than $0$, the minimizer $S$ must have $\Psi_{G,\vone^{A,B}}(S) \leq \alpha,$ as required. Otherwhise, the DSM algorithm produces a matching dual solution, which takes exactly the form of a feasible solution to \eqref{eq.rc-improve.dual} with value $\alpha$. Applying the best result of Ene \etal, achieved by the accelerated coordinated descent method under the assumption that a quadratic optimization oracle is given for each cut function $\delta_h,$ the running time to implement an approximate primal-dual oracle becomes $\tilde{O}\big( |V|\cdot\sum_{h\in E}\Theta_{h}\big),$
where $\Theta_h$ is the running time for a quadratic optimization oracle for $\delta_h$.

\paragraph{Graph-Reducible Polymatroidal Cut Functions}
As we have shown in Section~\ref{sec.polymatroidal},
for directed and standard hypergraph cut functions, the congestion constraint in the hyperflow problem~\ref{eq.rc-improve.dual} corresponds to vertex capacities on the hyperedge nodes of the factor graph $\tilde{G}.$ To implement the primitive $\cA$ described above for a fixed $\alpha$, we construct a network flow problem from this vertex-capacitated factor graph by adding an auxiliary source vertex $s$ and an auxiliary vertex $t$. The vertex $s$ is connected to every vertex in $i \in A$ with an arc of capacity $\alpha \mu_i$, while the vertex $t$ is connected to every vertex $j \in B$ with an arc of capacity $\alpha \cdot \mu_j \cdot \nicefrac{\mu(A)}{\mu(B)}.$ It is easy to see that the capacity of a cut of $V$ in this new network is, for every $S \subseteq V$:
\begin{equation}
\label{eq.cap}
\operatorname{cap}(\{s\} \cup S) \defeq \sum_{h \in E} w_h \delta_h(S) - \alpha \langle \vone_S, \vone^{A,B} \rangle_{\mu} + \alpha \cdot \mu(A)
\end{equation}
Now, apply a $\frac{1}{2}$-approximate maximum flow solver to this network flow problem to obtain a flow a value $F$ and an $s$-$t$ cut $C$ of capacity $2F.$ We have two cases:
\begin{enumerate}
    \item If $F \geq \frac{\alpha}{2} \cdot \mu(A),$ then the flow routed  is an approximate dual solution of value $\alpha$, as required.
    \item If $F <\frac{\alpha}{2} \cdot \mu(A),$ the $s$-$t$ cut has capacity strictly less then $\alpha \cdot \mu(A).$ As a result, we have by Equation~\ref{eq.cap}:
    $$
   \sum_{h \in E} w_h \delta_h(S) < \alpha \langle \vone_S, \vone^{A,B} \rangle_{\mu}.
    $$
    This implies that $\Psi_{G,\vone^{A,B}}(S) \leq \alpha \leq 2\alpha,$ are required for $\cA.$
\end{enumerate}
The almost-linear time result follows by applying the almost-linear time algorithm of Chen \etal~\cite{chenMaximumFlowMinimumCost2022} as our approximate maximum flow solver. For the standard hypergraph cut function, as the edges in the factor graph are undirected, it suffices to use the almost-linear time solver of Bernstein \etal~\cite{bernstein2022deterministic} for undirected vertex-capacitated graphs.

\section{Proof of Theorem~\ref{thm.cm.approx}: the cut-matching game reduction}
\label{sec.cut-strategy.reduction-proof}

\cmapprox*
\begin{proof}[{\bf Proof of Theorem~\ref{thm.cm.approx}}]
First, we check that the dual graph certificates $D_t$ output by $\cA_{ci}$ conform to the definition of an approximate matching response.  Then, the bipartiteness of $D_t$ across $(A_t, B_t)$ follows from requirement 3 in Definition~\ref{def.graph}, while the degree and largeness properties follow from requirements 4 and 5 in the same definition. In the symmetric case, we may assume each $D_t$ is undirected.

In the following, let $\alpha_t$ the value of $\alpha$ achieved by the $t$-invocation of $\cA_{ci}$ and
$$
\alpha^* = \min\{\alpha_t\}_{t=1}^{g(n)}.
$$
By  Lemma~\ref{lem.rc-improve.relaxation} and the definition of approximate primal-dual oracle and, we have that
$$
\min \bigg\{\Psi_G(C_t)\bigg\}_{t=1}^{g(n)} \leq \min \bigg\{\Psi_{G,\vone^{A_t, B_t}}(C_t)\bigg\}_{t=1}^{g(n)}  \leq
3 \cdot \min \bigg\{\alpha_t \bigg\}_{t=1}^{g(n)}= 3 \cdot \alpha^*.
$$

By property 1 in Definition~\ref{def.graph}, we have that $\alpha_t D_t$ embeds in $G$ with congestion $1$ for all $t.$
Averaging this guarantee over all $t$, we have:
$$
\frac{\alpha^*}{g(n)}\cdot H_{g(n)} \preceq_1 \frac{1}{g(n)} \cdot \sum_{t=1}^{g(n)} \alpha_t D_t \preceq_1 G
$$
Applying Theorem~\ref{thm.flow-embedding} and the definition of good strategy, we obtain the following lower bound on cuts of $G$, for all $S \subseteq V:$
$$\alpha^* \cdot \frac{f(n)}{g(n)} \leq
\frac{\alpha^*}{g(n)}\cdot \Psi_{H_{g(n)}}(S) \leq 2 \cdot \Psi_G(S)
$$
In particular, this applies to the minimum-ratio cut in $G$, so that we have:
$$
\min \bigg\{\Psi_G(C_t)\bigg\}_{t=1}^{g(n)} \leq 3 \cdot \alpha^* \leq 6 \cdot \frac{g(n)}{f(n)} \cdot \Psi_G^*.
$$
This completes the proof.
\end{proof}

\end{document}